\documentclass{article}

\usepackage{microtype}
\usepackage{graphicx}
\usepackage{subfigure}
\usepackage{booktabs} 

\usepackage{hyperref}

\usepackage[accepted]{icml2024}

\usepackage{amsmath}
\usepackage{amssymb}
\usepackage{mathtools}
\usepackage{amsthm}
\usepackage{enumerate}
\usepackage{color, soul, threeparttable, multirow, multicol, setspace}
\usepackage{amsmath, amssymb, amsfonts, bbm, caption, tikz, booktabs, colortbl, dcolumn, float}
\usepackage{algorithm}

\definecolor{jcolor}{RGB}{041,122,000}
\definecolor{darkred}{RGB}{100,000,000}
\definecolor{purple}{RGB}{200,000,200}

\newcommand{\cut}[1]{}
\newcommand{\commentout}[1]{}

\newcommand{\bd}{\boldsymbol}
\newcommand{\Ex}{\mathbb{E}}

\newcommand{\mc}{\mathcal}

\newcommand{\bX}{\boldsymbol{X}}
\newcommand{\bx}{\boldsymbol{x}}

\newcommand\independent{\protect\mathpalette{\protect\independenT}{\perp}}
\def\independenT#1#2{\mathrel{\rlap{$#1#2$}\mkern2mu{#1#2}}}

\usepackage[capitalize,noabbrev]{cleveref}

\theoremstyle{plain}
\newtheorem{theorem}{Theorem}[section]

\newtheorem{lemma}[theorem]{Lemma}

\theoremstyle{definition}

\newtheorem{assumption}[theorem]{Assumption}
\theoremstyle{remark}
\newtheorem{remark}[theorem]{Remark}

\usepackage[textsize=tiny]{todonotes}

\icmltitlerunning{Multi-Source Conformal Inference Under Distribution Shift}

\begin{document}

\twocolumn[
\icmltitle{Multi-Source Conformal Inference Under Distribution Shift}





\begin{icmlauthorlist}
\icmlauthor{Yi Liu}{a}
\icmlauthor{Alexander W. Levis}{b}
\icmlauthor{Sharon-Lise Normand}{c}
\icmlauthor{Larry Han}{d}
\end{icmlauthorlist}

\icmlaffiliation{a}{North Carolina State University, Department of Statistics, Raleigh, NC, USA}
\icmlaffiliation{b}{Carnegie Mellon University, Department of Statistics, Pittsburgh, PA, USA}
\icmlaffiliation{c}{Harvard Medical School, Department of Health Care Policy, Boston, MA, USA}
\icmlaffiliation{d}{Northeastern University, Department of Health Sciences, Boston, MA, USA}

\icmlcorrespondingauthor{Larry Han}{lar.han@northeastern.edu}

\vskip 0.3in
]



\printAffiliationsAndNotice{}  

\begin{abstract}
Recent years have experienced increasing utilization of complex machine learning models across multiple sources of data to inform more generalizable decision-making. However, distribution shifts across data sources and privacy concerns related to sharing individual-level data, coupled with a lack of uncertainty quantification from machine learning predictions, make it challenging to achieve valid inferences in multi-source environments. In this paper, we consider the problem of obtaining distribution-free prediction intervals for a target population, leveraging multiple potentially biased data sources. We derive the efficient influence functions for the quantiles of unobserved outcomes in the target and source populations, and show that one can incorporate machine learning prediction algorithms in the estimation of nuisance functions while still achieving parametric rates of convergence to nominal coverage probabilities. Moreover, when conditional outcome invariance is violated, we propose a data-adaptive strategy to upweight informative data sources for efficiency gain and downweight non-informative data sources for bias reduction. We highlight the robustness and efficiency of our proposals for a variety of conformal scores and data-generating mechanisms via extensive synthetic experiments. Hospital length of stay prediction intervals for pediatric patients undergoing a high-risk cardiac surgical procedure between 2016-2022 in the U.S. illustrate the utility of our methodology.
\end{abstract}

\section{Introduction}
\label{sec:introduction}
Conformal inference is a set of methods used to construct distribution-free, nonparametric prediction intervals, for an outcome $Y$ on the basis of covariates $\bX$, with finite-sample marginal coverage guarantees. The framework was first introduced by \citet{vovk2005algorithmic, vovk2009line} and has since been extended to regression settings under covariate shift \citep{lei2018distribution, tibshirani2019conformal, lei2021conformal}. Recently, \citet{yang2024doubly} proposed robust prediction intervals under covariate shift by revealing a connection with the missing data literature, and appealing to modern semiparametric efficiency theory. However, \citet{yang2024doubly} assume only a single data source such that the conditional outcome distribution $Y \mid \bX$---and therefore the conditional distribution of conformal scores---is homogeneous. In general, conformal prediction methods have focused on covariate shift while assuming that conditional outcome distributions are invariant across environments \citep{peters2016causal}. We note, however, that some work has studied label shift settings (e.g., \citet{podkopaev2021}), but this involves the analogously strong assumption that the distribution of $\bX \mid Y$ is homogeneous. We refer the reader to \citet{barber2023conformal} and the extensive literature review therein for other works.

In reality, conditional outcome invariance is unlikely to hold in the real world. In recent years, there has been a huge increase in popularity in using large clinical research networks that facilitate multi-center collaboration. One goal with these networks is to leverage the multiple diverse data sources to mitigate issues related to small or non-representative data, thereby increasing statistical power for probing various scientific hypotheses. However, different clinical sites may be heterogeneous in terms of patient populations, treatment practices, and patient outcomes.  Furthermore, since individual-level data is protected by privacy regulations such as HIPAA and GDPR, direct pooling of data across sites is typically not feasible. Federated transfer learning methods have been proposed as powerful tools for integrating heterogeneous data \citep{duan2020learning, li2023targeting}, and have been applied to yield robust point estimation of the effect of a treatment on a combined population across sites \citep{xiong2023federated, vo2022bayesian}, and for the treatment effect on a specific target population \citep{han2021federated, han2024privacy, han2023multiply, vo2022adaptive}, while accounting for data-sharing constraints and heterogeneity (i.e., covariate shift and different conditional outcome distributions). It has also been applied in problems related to interval estimation, e.g., constructing robust confidence intervals by selecting eligible sites \citep{guo2023robust} with uniform coverage guarantees.
 
 In conformal prediction, \citet{lu2023federated} proposed a notion of partial exchangeability, but the focus of their work is to construct prediction intervals on the combined population across sites, and not any particular target site. Relatedly, \citet{plassier2023conformal} considered federated conformal prediction under label shift via quantile regression, and \citet{humbert2023one} proposed a quantile-of-quantiles estimator for conformal prediction by aggregating multiple quantiles returned by each site.  To date, there are no federated learning methods developed for conformal inference on a missing outcome in a setting with distribution shift across multi-site data, and where data cannot be directly combined due to privacy concerns. When conditional outcome distributions are not the same across sites, there is likely to be poor conformal set performance with existing methods when transferring prediction models (e.g., learned conditional quantiles) from one set to another, e.g., deployment to target distributions that are different from the source distribution \citep{jin2023diagnosing, cai2023diagnosing}. 

Our work differs from recent work by \citet{lee2023distribution} and \citet{dunn2023distribution} in important ways. \citet{lee2023distribution} focus on predicting an outcome on a new subject from a new (unobserved) site. \citet{dunn2023distribution} focus on this same task, and also consider a simplified version of the problem of predicting an outcome on a new subject from an existing (observed) site---they propose an unsupervised method that does not allow for the inclusion of covariates, and leave the supervised version as an open problem. Neither work allows for outcome missingness. In this paper, we fill these methodological and applied gaps by leveraging conformal prediction tools to provide patients with personalized predictions using multi-source data, accounting for missing data and distribution shifts, i.e., covariate shift and heterogeneous conditional outcome distributions. We propose a method to obtain valid prediction intervals, exploiting information from multiple potentially heterogeneous sites, and respecting the privacy of individual-level data when it cannot be shared. Our proposal shares the marginal coverage properties of conformal prediction methods and builds on modern semiparametric efficiency theory and federated learning for more robust and efficient uncertainty quantification. 

\section{Prediction interval construction}\label{sec:meth}

\subsection{Notation and background}\label{sec:setup}
Consider the following multi-site paradigm with missing data. We have data from $K$ sites, and for each subject in each site, we observe a covariate vector $\bX$. Let $T \in \{0,1,...,K-1\}$ denote the study sites, where $T=0$ indicates the target site and the remainder are source sites. Let $R$ be an indicator for observing the outcome $Y$, i.e., $R = 1$ if $Y$ is observed and $R=0$ if $Y$ is missing. The data are assumed to be a random sample of $n$ i.i.d. copies of $\mathcal{O} = (\bX, T, R, RY) \sim \mathbb{P}$. Throughout, let $\mathbb{P}_n(f) \equiv \frac{1}{n}\sum_{i=1}^n f(\mathcal{O}_i)$ be shorthand for the empirical average. To proceed, we make the following standard assumptions.
\begin{assumption}[Missing at random {[MAR]}]\label{ass:MAR}
    \[R \independent Y \mid T, \bX.\]
\end{assumption}
\begin{assumption}[Positivity]\label{ass:pos} For $\epsilon > 0$, 
    $$\mathbb{P}[\mathbb{P}[R = 1 \mid T, \bX] \geq \epsilon] = 1.$$
\end{assumption}
Note that MAR (i.e., Assumption~\ref{ass:MAR}), which asserts that missingness status is not informative about outcomes, given $T$ and $\bX$, and positivity (i.e., Assumption~\ref{ass:pos}), which requires that no subjects have outcomes that could never be observed, are both required for point identification of the distribution of missing outcomes and are standard in this literature \citep{lei2021conformal, yang2024doubly}

We construct prediction intervals of the form $\widehat C_\alpha(\bX)$, for $\alpha\in(0,1)$, such that
$$\mathbb{P}(Y \in \widehat{C}_\alpha(\bX) \mid T = 0, R = 0) \geq 1-\alpha.$$
That is, our predictions should be tailored for missing outcomes \textit{in the target site}, with marginal coverage guarantees. In the spirit of conformal inference, we introduce a conformal score, $S(\bX, Y)$, which for now we assume is fixed. Our predictions will be based on this score, namely $\widehat{C}_{\alpha}(\bX) = \left\{y \in \mathbb{R}: S(\bX, y) \leq \widehat{r}\right\}$, where $\widehat{r}$ is an estimate of $r_0= r_0(\alpha)(\mathbb{P})$, the $(1-\alpha)$-quantile of the conformal score $S(\bX,Y)$ in the target site.

Under MAR, the functional $r_0=r_0(\alpha)(\mathbb{P})$ is identified as the solution to an estimating equation:
{\small
\begin{align*}
    & \mathbb{P}(S(\bX,Y) \leq r_0 \mid T=0, R=0) \\
     & = \mathbb{E}_{\mathbb{P}}(\mathbb{P}(S(\bX,Y) \leq r_0 \mid T=0, \bX, R=1) \mid T=0, R = 0)
    \\ & = 1-\alpha. 
\end{align*}}
Without imposing any further structure, the nonparametric influence function of this functional can be derived \citep{yang2024doubly}. 

\begin{theorem}[\citet{yang2024doubly}] \label{thm:yang}
Under Assumptions~\ref{ass:MAR} and \ref{ass:pos}, the nonparametric influence function of the functional $r_0 = r_0(\alpha)(\mathbb{P})$ is given by
\begin{align}\label{eq:eif-target}
    & \dot{r}_0 (\mc O;\mathbb{P})  \nonumber \\
    & \propto I(T=0)\big[(1-R)\{m_0(r_0,\bX) - (1-\alpha)\} \nonumber \\
    & \quad\quad + R\eta_0(\bX)\{I(S(\bX,Y) \leq r_0) - m_0(r_0,\bX)\}\big] \nonumber \\
    & \eqqcolon \varphi_0(\mc O; r_0, m_0, \eta_0),
\end{align}
where 
$$m_0(r,\bX) =\mathbb{P}(S(\bX,Y) \leq r \mid \bX, T = 0, R=1)$$ is the cumulative distribution function (CDF) of the conformal score, and 
$$\eta_0(\bX) = \frac{\mathbb{P}(R=0 \mid T=0, \bX)}{\mathbb{P}(R=1 \mid T=0, \bX)}$$ is the missingness risk ratio. 
\end{theorem}
\citet{yang2024doubly} propose a robust estimator $\widehat{r}_0$ that solves $0 = \mathbb{P}_n\left[\varphi_0(\mc O; r, \widehat{m}_0, \widehat{\eta}_0)\right]$ for $r$, where $\widehat{m}_0, \widehat{\eta}_0$ are estimated nuisance functions. 

Applying the method of \citet{yang2024doubly} in our multi-source data setting would only use data from the target site $T=0$ itself. To leverage data from the other $K-1$ sites, we make two contributions: (i) we propose a fully efficient estimator of $r_0 $ under further structural assumptions regarding outcome distribution homogeneity (Section~\ref{subsec:ccod}), and (ii) develop (Section~\ref{subsec:hetedis}) and implement (Section~\ref{sec:simu}) a data-adaptive approach when these structural assumptions may be violated.

\subsection{Efficient estimation under homogeneity}\label{subsec:ccod}

When subjects from different data sources are deemed to be similar, it may be reasonable to assert that the outcome distribution is common across them. This idea is formalized with the following structural assumption.
\begin{assumption}[Common conditional outcome distribution [CCOD{]}]\label{ass:ccod}
    $T \independent Y \mid \bX$.
\end{assumption}
Notably, Assumption~\ref{ass:ccod} entails no restriction on the covariate distribution across sites. That is, any level of covariate shift is permitted.  Under CCOD (i.e., Assumption~\ref{ass:ccod}), data from non-target source sites may be leveraged to improve the estimation of the target site quantile $r_0$. Our first result generalizes Theorem~\ref{thm:yang} to the multi-source setting under CCOD.

\begin{theorem}\label{thm:IF-ccod}
    Under Assumptions~\ref{ass:MAR}, \ref{ass:pos}, and \ref{ass:ccod}, the semiparametric efficient influence function (EIF) of $r_0 = r_0(\alpha)(\mathbb{P})$ is given by
    \begin{align} \label{eq:eif-ccod}
        & \dot{r}_0^{\mathrm{CCOD}}(\mc O; \mathbb{P}) \nonumber \\
        & \propto I(T = 0)(1 - R)\left\{\overline{m}(r_0, \bX) - (1 -\alpha)\right\} \nonumber \\
        & \quad + R\overline{\eta}(\bX) q_0(\bX)\left\{I(S(\bX, Y) \leq r_0) - \overline{m}(r_0, \bX)\right\} \nonumber \\
        & \eqqcolon \varphi^{\mathrm{CCOD}}(\mc O; r_0, \overline{m}, \overline{\eta}, q_0),
    \end{align}
    where 
$$\overline{m}(r,\bX) =\mathbb{P}(S(\bX,Y) \leq r \mid \bX, R=1)$$ is the global CDF of the conformal score, 
$$\overline{\eta}(\bX) = \frac{\mathbb{P}(R=0 \mid \bX)}{\mathbb{P}(R=1 \mid \bX)}$$ is the global missingness risk ratio, and \[q_0(\bX) = \mathbb{P}[T = 0 \mid \bX, R = 0]\]
is the target-site propensity.
\end{theorem}

Compared to the nonparametric influence function of the $(1-\alpha)$-quantile of the conformal score (\ref{eq:eif-target}), which uses data from the target site only, the semiparametric EIF (\ref{eq:eif-ccod}) leverages data from all sites with observed outcomes $Y$.
Under CCOD, we propose the estimator $\widehat{r}^{\mathrm{CCOD}}$ which solves $0 = \mathbb{P}_n\left[\varphi^{\mathrm{CCOD}}(\mc O; r, \widehat{\overline{m}}, \widehat{\overline{\eta}}, \widehat{q}_0)\right]$ for $r$. We perform cross-fitting such that the nuisance estimators $(\widehat{\overline{m}}, \widehat{\overline{\eta}}, \widehat{q}_0)$ are estimated on an independent data split from the given estimating equation. The following result demonstrates the marginal coverage properties of the conformal interval $\widehat{C}_{\alpha}^{\mathrm{CCOD}}(\bX) = \{y \in \mathbb{R}: S(\bX, y) \leq \widehat{r}^{\mathrm{CCOD}}\}$.

\begin{theorem}\label{thm:cov-ccod}
    Let $D^n$ denote the training data with which $\widehat{r}^{\mathrm{CCOD}}$ is fit, and let $(\bX, T, R)$ denote a new independent test point with associated outcome $Y$. Assume that $(\widehat{\overline{m}}, \widehat{\overline{\eta}}, \widehat{q}_0)$ are each uniformly bounded, and that $\widehat{\overline{m}}(\, \cdot \,, \bx)$ is a non-decreasing function, for each $\bx$. Under Assumptions~\ref{ass:MAR}, \ref{ass:pos}, and \ref{ass:ccod},
    \begin{align*}
        & \mathbb{P}[Y \in \widehat{C}_{\alpha}^{\mathrm{CCOD}}(\bX) \mid T = 0, R = 0, D^n] \\
        &= (1 - \alpha) + O_{\mathbb{P}}(n^{-1/2} + R_n),
    \end{align*}
    where \[R_n = \left\{\lVert \widehat{\overline{\eta}} - \overline{\eta}\rVert + \lVert \widehat{q}_0 - q_0 \rVert\right\} \sup_{r} \lVert \widehat{\overline{m}}(r, \cdot) - \overline{m}(r, \cdot) \rVert.\] Here $\lVert f \rVert^2 = \mathbb{E}_{\mathbb{P}}(f(\mathcal{O})^2)$ is the squared $L_2(\mathbb{P})$ norm.
\end{theorem}
Theorem \ref{thm:cov-ccod} says that conditional on training data, the proposed prediction interval attains nominal coverage at essentially parametric rates (some authors reserve the term parametric rate for a $o_\mathbb{P}(n^{-1/2})$ remainder), so long as the second order asymptotic bias term $R_n$ converges fast enough to zero. The robustness of our estimator is made clear from inspecting this bias, and Theorem~\ref{thm:cov-ccod} supports flexible (i.e., non- or semi-parametric) estimators for component nuisance functions: $R_n = O_{\mathbb{P}}(n^{-1/2})$ will hold whenever $\overline{m}_0, \overline{\eta}_0, q_0$ are estimated at $O_{\mathbb{P}}(n^{-1/4}$) rates, which may be achievable under smoothness, sparsity, or other structural conditions. Since the bias of our coverage error rate is of the order of the product of two errors, it can be substantially smaller relative to that of related work by \citet{lei2021conformal} (which would include data from the target site only in this setting), which has a bias of the order of the minimum of two errors \citep{yang2024doubly}. We note that the boundedness assumptions in Theorem~\ref{thm:cov-ccod} are standard, and that $\widehat{\overline{m}}(\, \cdot \,, \bx)$ should well be monotone, as it estimates the CDF $\overline{m}(\, \cdot \,, \bx)$.
\vspace{10pt}
\begin{remark}
    Whereas the coverage guarantees for prediction intervals in \citet{lei2021conformal} appear to hold only under a particular choice of conformal score (conditional quantile regression [CQR]), our methodology is not restricted by the choice of conformal score. To highlight the robustness of our procedure to the choice of conformal score, in the numerical experiments of Section \ref{sec:simu}, we evaluate three different conformal scores:
\begin{itemize}\setlength{\itemsep}{0pt}
    \item CQR score (see \citet{lei2021conformal}).
    \item Absolute residual (ASR): $S_{\textrm{ASR}}(x_i, y_i) = |y_i-\widehat\mu(x_i)|$, where $\widehat\mu(\cdot)$ is a regression model to estimate $\mu(x)=\Ex\{Y\mid X= x\}$. 
    \item Locally weighted ASR \citep{lei2018distribution}, defined by 
    $$
    S_{\textrm{local ASR}}(x_i, y_i) = \frac{|y_i-\widehat\mu(x_i)|}{\widehat\rho(x_i)},
    $$
    where $\widehat\rho(x_i)$ is an estimate of the conditional mean absolute deviation (MAD), $\mathbb{E}\{|Y_i-\mu(X_i)|~\big|~X_i=x_i\}$, a function of $x_i$ fitted on $\mc D_{11}$. 
\end{itemize}
\end{remark}

\subsection{Heterogeneous outcome distribution across sites}
\label{subsec:hetedis}

In practical settings, it will often be unreasonable to assume that the conditional outcome distribution is the same across all sites. In such cases, some source sites may provide relevant information for constructing target-site specific prediction intervals, whereas other sites may not. Concretely, the distribution of $Y$ given $(T = k, \bX)$ may be close to that in the target site $T = 0$ for some $k$, but not others. In this section, we present an approach that combines information from target and source sites in a data-adaptive manner. Our approach is also privacy-preserving, in that it involves only minimal data sharing of summary statistics across sites. 

Our proposal is to construct a $(1-\alpha)$-quantile for the target site by taking a weighted average of estimated quantiles $(\widehat{r}_0,\widehat{r}_1,\dots, \widehat{r}_{K-1})$, where  $\widehat{r}_k$ uses data from site $k$ for each $k$. We call the weights in the weighted average \textit{federated weights}. In the following subsections, we describe how the site-specific quantiles are estimated, and how the federated weights are obtained.

\subsubsection{Target site} 

For the target site, we estimate $\widehat{r}_0$ nonparametrically as in Section~\ref{sec:setup}. That is, we use the approach motivated by Theorem~\ref{thm:yang}, and take $\widehat{r}_0$ that solves
$
\mathbb{P}_{n}\left[\varphi_0\left(\mc O; \widehat{r}_{0}, \widehat{m}_{0}, \widehat{\eta}_{0}\right)\right]=0, 
$
where $\varphi_0$ is the nonparametric influence function \eqref{eq:eif-target}. 

\subsubsection{Source sites}
To construct a target-site specific quantile estimate using data from site $k \in \{1, \ldots, K-1\}$, we make a working partial CCOD assumption that outcomes have the same conditional distribution in site $k$ as in the target site. Note that we use this working partial CCOD assumption only to derive the form of the influence function; to aggregate information from source sites, we derive federated weights to account for possible violations of CCOD (Section \ref{agg}). An influence function under this assumption is derived in the following result.

\begin{theorem}\label{thm:IF-partial-ccod}
    Under Assumptions~\ref{ass:MAR}, \ref{ass:pos}, and the partial CCOD assumption $p(y \mid \bX, T = k) \equiv p(y \mid \bX, T = 0)$, an influence function (IF) of $r_0$ is given by
    {\small
    \begin{align*}
        & \dot{r}_k(\mc O; \mathbb{P}) \nonumber \\
        & \propto \frac{I(T=0,R=0)}{\mathbb P(T=0,R=0)}[m_0(r_0,\bX) - (1-\alpha)] \\
& \quad + \frac{I(T=k,R=1)}{\mathbb P(T=k,R=1)}\omega_{k,0}(\bX)[I(S(\bX, Y)\leq r_0) -m_k(r_0,\bX)] \\
        & \eqqcolon \varphi_{k} \left(\mc O ; r_0, m_0, m_{k}, \omega_{k,0}\right),
    \end{align*}
    }
    where $m_k(r,\bX) =\mathbb{P}(S(\bX,Y) \leq r \mid \bX, T = k, R=1)$ is the CDF of the conformal score in site $k$, and
$$
\omega_{k,0}(\bx) = \frac{p(\bx\mid T=0, R=0)}{p(\bx\mid T=k, R=1)}
$$
is a density ratio function of covariates $\bX$ under target site to source site $k$. 
\end{theorem}

Given some nuisance estimators $\widehat m_0$, $\widehat m_k$ and $\widehat\omega_{k,0}$, we take $\widehat{r}_k$ that solves
$$
\mathbb{P}_{n}\left[\varphi_{k}\left(\mc O; \widehat{r}_{k}, \widehat{m}_{0}, \widehat m_{k}, \widehat{\omega}_{k,0}\right)\right]=0. 
$$
By construction, the quantile estimate $\widehat{r}_k$ uses data from both site $k$ and the target site, but note that the principal need for data sharing comes from the estimation of the density ratio $\omega_{k, 0}$. This can be done with the passing of only coarse summary statistics under flexible models \citep{han2021federated}.

\subsubsection{Aggregation across sites}
\label{agg}
To aggregate information from the target and source sites, we first compute the discrepancy measures $\widehat\chi_k = |\widehat r_0-\widehat r_k|$,
then solve for federated weights $\widehat{\bd w} = (\widehat w_0,\widehat w_1,\dots, \widehat w_{K-1})$ that minimize the following loss:
\begin{align}
\label{eq:loss}
Q(\bd w)& =\mathbb{P}_{n}\Bigg[\bigg\{\varphi_{0}(\mc O; \widehat{r}_0, \widehat{m}_{0}, \widehat{\eta}_{0}) \nonumber\\
&\quad -\sum_{k=1}^{K-1} w_{k} {\varphi}_{k}(\mc O_i; \widehat r_0, \widehat m_0, \widehat m_k, \widehat{\omega}_{k,0})\bigg\}^{2}\Bigg] \nonumber\\
& \quad +\frac{1}{n} \lambda \sum_{k=1}^{K-1}\left|w_{k}\right| \widehat{\chi}_{k}^{2},
\end{align}
subject to $0 \leq w_{k} \leq 1$, for all $k \in\{0,1, \ldots, K-1\}, $ and $\sum_{k=0}^{K-1} w_{k}=1$, and $\lambda$ is a tuning parameter chosen by cross-validation. Heuristically, our approach anchors at the nonparametric estimate $\widehat{r}_0$ and weights site $k$ when it is deemed similar enough to the target site \citep{han2021federated}.

Finally, we compute $\widehat{r}_{0,\text{fed}}$ as the weighted average of the site-specific quantiles: $\widehat r_{0,\text{fed}} = \sum_{k=0}^{K-1} \widehat w_k\widehat r_k$. The federated prediction interval is then defined as $\widehat{C}_{\alpha}^{\mathrm{fed}}(\bX) = \{y \in \mathbb{R}: S(\bX, y) \leq \widehat{r}_{0,\mathrm{fed}}\}$. In the following, we provide a coverage guarantee for the prediction interval based on an estimated quantile that is an arbitrary weighted combination of the relevant (i.e., oracle) source sites.

\begin{theorem}[Oracle coverage result]\label{thm:cov-oracle}
    Let \[\mathcal{S}^* = \{k \geq 1 : p(y \mid \boldsymbol{X}, T = k) \equiv p(y \mid \boldsymbol{X}, T = 0)\},\] which may be empty, denote the source sites for which the partial CCOD assumption holds. Let $D^n$ denote the training data with which $\widehat{r}_{0, \mathrm{fed}}$ is fit, and let $(\bX, T, R)$ denote a new independent test point with associated outcome $Y$. Assume that $(\widehat{\eta}_0, \widehat{m}_0)$ and $(\widehat{\omega}_{k, 0},\widehat{m}_k)$, for $k \in \mathcal{S}^*$, are each uniformly bounded, and that $\widehat{m}_k(\, \cdot \,, \bx)$ is a non-decreasing function for $k \in \{0\} \cup \mathcal{S}^*$, for each $\bx$. For any $w^* = (w_0, \ldots, w_{K - 1})$ with $w_k \geq 0$, $\sum_{k = 0}^{K - 1}w_k = 1$, and satisfying $w_k = 0$ for $k \not \in \{0\} \cup \mathcal{S}^*$, define 
    \[\widehat{C}_{\alpha}^{w^*}(\boldsymbol{X}) = \left\{y \in \mathbb{R}: S(\boldsymbol{X}, y) \leq \sum_{k= 0}^{K - 1} w_k \widehat{r}_k\right\}.\] Then under Assumptions 2.1 and 2.2, and conditions (i)--(iii) of Lemma~\ref{lemma:cdf-cont},
    \begin{align*}
        & \mathbb{P}[Y \in \widehat{C}_{\alpha}^{w^*}(\bX) \mid T = 0, R = 0, D^n] \\
        &= (1 - \alpha) + O_{\mathbb{P}}(n^{-1/2} + R_n^*),
    \end{align*}
    where 
        \begin{align*}
      & R_n^* = w_0\left\{\lVert \widehat{\eta}_{0} - \eta_{0}\rVert \cdot \sup_{r} \lVert \widehat{m}_0(r, \cdot) - m_0(r, \cdot) \rVert\right\} \\
      & \quad + \sum_{k = 1}^{K-1} w_k\bigg\{\lVert \widehat{\omega}_{k, 0} - \omega_{k, 0}\rVert\cdot \sup_{r} \lVert \widehat{m}_k(r, \cdot) - m_k(r, \cdot) \rVert
      \\
      & \quad \quad \quad \quad \quad \quad \quad \quad + \sup_{r} \lVert \widehat{m}_k(r, \cdot) - \widehat{m}_0(r, \cdot) \rVert\bigg\}.
    \end{align*}
\end{theorem}

 Note that our penalization procedure in (\ref{eq:loss}) is designed such that $w_k \to 0$ whenever $k \notin \mathcal{S}^*$, akin to adaptive Lasso \citep{zou2006adaptive} and trans-Lasso \citep{fan2024fast}.

\subsection{Estimation with data splitting}
To construct target-site-specific prediction intervals for missing outcomes leveraging information from all sites, we follow the steps described in Algorithm \ref{algo:robust}. In brief, we randomly split the training data $\mc D$ into two equal-sized folds $\mc D_1\cup\mc D_2$. 
We train the models for the putative CDFs of the conformal score $m_k$, $k=0,1,\dots, K-1$ on $\mc D_{11}$. Likewise, we train the density ratio model $\omega_{k,0}$ on $\mc D_1$. We fit all nuisance functions using SuperLearner with the base learners being random forest, elastic net, and generalized linear model (GLM). SuperLearner is a meta-learning algorithm that creates an optimal weighted average of the base learners and is shown to be as accurate
as the best possible prediction algorithm \citep{van2007super}. Density ratio models accommodate flexible basis functions and higher order terms to capture higher-order differences such as variance and skewness. One example we consider is the exponential tilt model, which recovers the entire class of natural exponential family distributions, including the normal distribution with mean shift, Bernoulli distribution for binary covariates, and more \citep{qin1998inferences, duan2020fast}. We predict the values from the trained models on $\mc D_2$ and plug these values into the IFs given in Algorithm \ref{algo:robust}. Figure \ref{fig:flow} provides a visualization of the procedure. Full detail on influence function estimation is given in Algorithm \ref{algo:est-if}.

\begin{figure*}
    \includegraphics[width=1\textwidth]{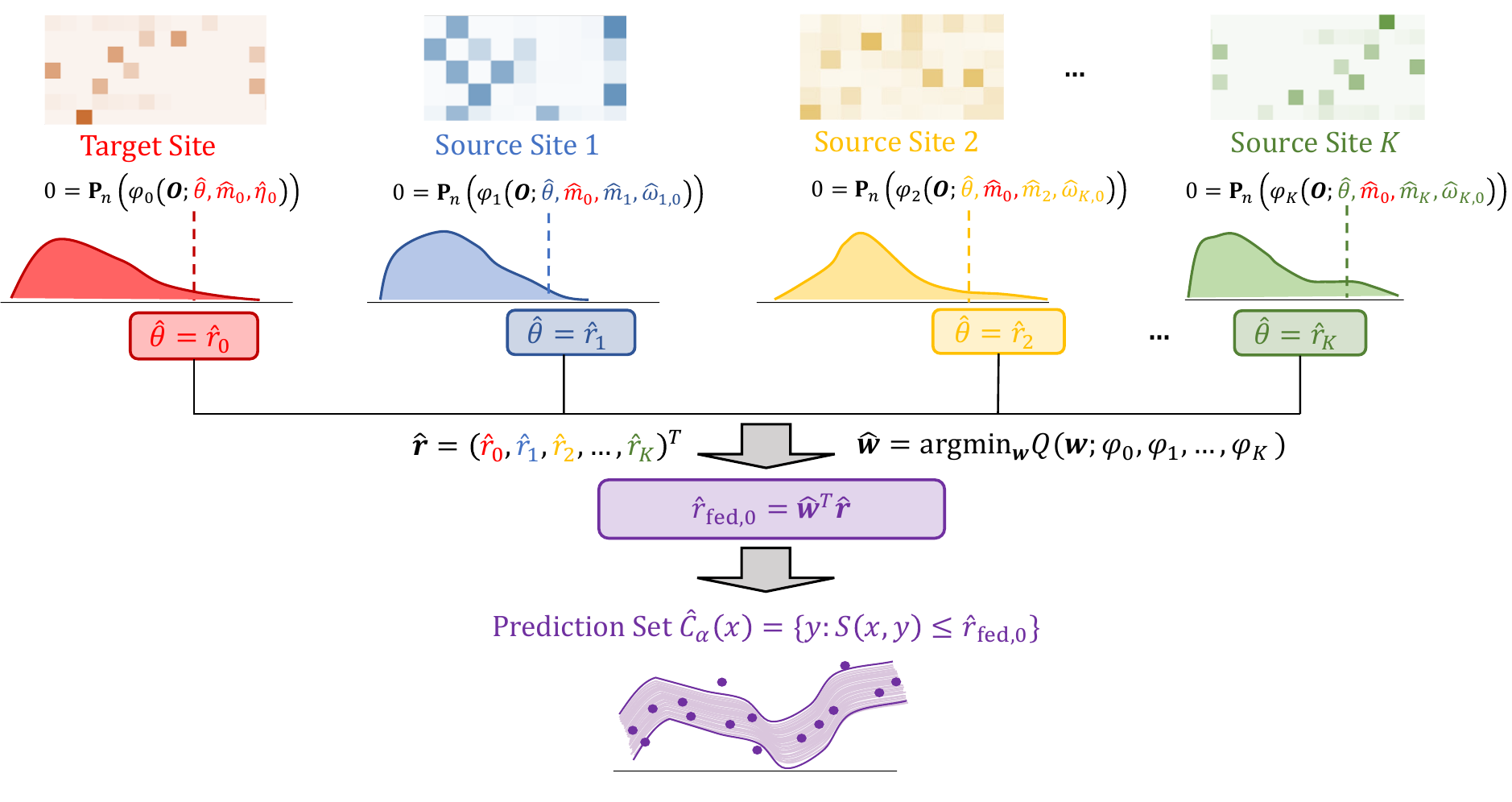}
    \caption{Illustration of the proposed robust algorithm for multi-source conformal prediction. Each $\widehat\theta$ represented by a different color is the estimated $(1-\alpha)$-quantile of the conformal score using data from the site with the same color. $\widehat m_0$ (in red) is the estimated CDF of the conformal score using only the target site data. The other $\widehat m_k~ (k\geq 1)$ are the estimated CDFs of the conformal scores from source sites, and $\widehat\omega_{k,0}~(k\geq 1)$ is the density ratio of site $k$ versus the target site. The federated $\widehat{r}_{\textrm{fed,0}}$ is a weighted average of the site-specific quantiles, with weights given by $\widehat{\bd w}$. The prediction interval $\widehat{C}_{\alpha}(\bx)$ is the set of outcomes $y$ such that the corresponding conformal scores $S(\bx,y)$ in the target are below the threshold $\widehat{r}_{\textrm{fed,0}}$.
 }\label{fig:flow}
\end{figure*}

\begin{algorithm*}
  \caption{Robust multi-source conformal prediction}\label{algo:robust}
  \begin{algorithmic}[1]
  \STATE {\bfseries Input:} {Training data $\mc D = \{\mc O_i = (\bX_i,T_i, R_i,R_iY_i), i=1,\dots, n\}$ with number of sites $K>0$, and the target site is indexed by $T=0$; desired coverage probability $1-\alpha$; estimators of nuisance functions $m_k(\theta, \bX)$, $\eta_0(\bd X)$, and $\omega_{k,0}(\bX)$ for $k=1,\dots,K-1$; a tuning parameter $\lambda$ (in the optimization step); a testing point $\bX=\bx$ from the target site. }
  
  \STATE {\bfseries Output:} {A valid prediction set $\widehat C_\alpha(\bx)$. }
  
  \STATE Split the training data $\mc D$ randomly into $\mc D_1$ and $\mc D_2$, where $\mc D_j=\{\mc O_i\in\mc D, i\in\mc I_j\}$ for $j=1,2$ and $\mc I_1\cup\mc I_2 = \{1,2,\dots, n\}$. 

  \STATE Fit nuisance functions $\widehat m_k$ and $\widehat\omega_{k,0}$ using SuperLearner on $\mc D_1$ and predict them on $\mc D_2$. 
  
  \STATE For the target site $k=0$, find $\widehat\theta = \widehat r_0$ that solves $0 = \dfrac1{|\mc I_2|}\displaystyle\sum_{i\in\mc I_2}\varphi_0(\mc O_i; \widehat\theta, \widehat m_0, \widehat\eta_0).$
  
 \STATE For source sites $k\geq 1$, find $\widehat\theta = \widehat r_k$ that solves $0 = \dfrac1{|\mc I_2|}\displaystyle\sum_{i\in\mc I_2}\varphi_k(\mc O_i; \widehat\theta, \widehat m_0, \widehat m_k, \widehat\omega_{k,0})$. 
 Compute $\widehat\chi_k = |\widehat r_0-\widehat r_k|$. 
  
\STATE Solve for aggregation weights $\widehat{\bd w}=(\widehat w_0, \widehat w_1,\dots\widehat w_{K-1})$ that minimize $Q(\bd w)$ subject to $0\leq w_k\leq 1$  and $\displaystyle\sum_{k=0}^{K-1} w_k = 1$. 

\STATE Compute $\widehat\theta = \widehat r_{0,\text{fed}} = \displaystyle\sum_{k=0}^{K-1} \widehat w_k\widehat r_k$. 
  
  \STATE \textbf{Return: }The prediction set $\widehat C_\alpha(\bx) = \{y: S(\bx,y)\leq \widehat r_{0,\text{fed}}\}$. 
  \end{algorithmic}
\end{algorithm*}

\section{Numerical Experiments}\label{sec:simu}

In this section, we evaluate our proposed method by conducting extensive Monte Carlo simulations, examining aspects such as marginal coverage, conditional coverage, and the width of the prediction interval. In each experiment, we compare our proposed federated method to construct prediction intervals $\widehat C_\alpha(x)$ against (i) the nonparametric efficient method described in \citet{yang2024doubly}, which uses data from the target site only and ignores external source data (target only) and (ii) the method that assumes CCOD holds across sites (pooled sample). In Appendix \ref{apx:simu}, we describe three other methods for learning the federated weights $\hat{w}$ and provide complete simulation results (see details in Appendix \ref{subapx:completeRes}). 

In total, we consider $3$ sample sizes $(300,1000,3000)$ $\times$ $3$ levels of covariate shift (homogeneous, weakly heterogeneous, strongly heterogeneous) $\times$ $2$ types of outcome errors (homoskedastic, heteroskedastic) $\times$ $3$ levels of concept shift (CCOD holds, weak violation, strong violation) $\times$ $3$ different conformal scores (ASR, locally weighted ASR, CQR) $ = 162$ scenarios for our proposed method and the five competitor methods.

\subsection{Data generating process}\label{subsec:DGP}

We generate data from $K=5$ sites, where site 0 is the target site and sites 1 through 4 are source sites, and $T_i \in \{0,\cdots,4\}$ denotes the site of subject $i$. Our goal is to construct valid prediction intervals for a testing point from the target site. We consider the sample size in each site to be $n_k\in\{300,1000,3000\}$, $k=0,...,4$ and generate data over $M=500$ independent Monte Carlo replications. We consider three site-specific covariate data generation scenarios:
\begin{itemize}
    \item Homogeneous covariate distributions: $X_i = \Phi(X_i^*)$ where $X_i^*\sim\mc N(0,1)$, and $\Phi(\cdot)$ is the CDF of the standard normal distribution, for all sites. 
    \item Weakly heterogeneous covariate distributions: $X_i^*\mid T_i\in\{0,1\}\sim\mc N(0,1)$, $X_i^*\mid T_i=2\sim\mc N(2,1)$, $X_i^*\mid T_i=3\sim\mc N(2,4)$, $X_i^*\mid T_i=4\sim\mc N(3,1)$, and  $X_i = \Phi(X_i^*)$. 
    \item Strongly heterogeneous covariate distributions: $X_i^*\mid T_i=0\sim\mc N(0,1)$, $X_i^*\mid T_i=1\sim\mc N(1,1)$, $X_i^*\mid T_i=2\sim\mc N(2,4)$, $X_i^*\mid T_i=3\sim\mc N(3,1)$, $X_i^*\mid T_i=4\sim\mc N(4,4)$, and  $X_i = \Phi(X_i^*)$. 
     \end{itemize}

For each scenario, we generate the propensity score of observing the outcome, i.e., $e(X_i) = P(R_i=1\mid X_i)$, by a logistic regression model, where
\allowdisplaybreaks\begin{align*}
    e(X_i) = \{1+\exp(-0.1+0.5X_i-0.1X_i^2)\}^{-1},
\end{align*}

ensuring that the true propensity score is in $(0.4,0.6)$ to avoid positivity violations. We include additional simulation results where the true propensity score is in the wider range $(0.1,0.9)$. We generate $R_i$ by $\text{Bern}(e(X_i))\in\{0,1\}$ so that outcomes are MAR. 

The outcomes $Y_i$ are generated by 
\allowdisplaybreaks\begin{align}\label{eq:outcome_md}
    Y_i = 5X_i + X_i^2 + \delta(T_i, X_i) + \varepsilon(X_i),
\end{align}
where $\varepsilon(x)\sim N(0, \sigma(x)^2)$. We consider two types of errors: (i) $\sigma(x)=1$ for homoscedastic errors and (ii) $\sigma(x) = -\log(x)$ for heteroscedastic errors. Under both cases, the oracle width of a $90\%$ prediction interval for the outcome is $2\times z_{0.95}\Ex\{\sigma(X_i)\}\approx 3.29$, where $z_{0.95}=1.645$ is the 95th percentile of the standard normal distribution. In addition, note that $\Ex\{\sigma(X_i)\} = \int_0^1\sigma(x)dx = 1$ for both $\sigma(x)=1$ and $\sigma(x)=-\log(x)$.

We also consider varying levels of concept shift corresponding to three cases for $\delta(T_i, X_i)$:
\begin{itemize}
    \setlength{\itemsep}{0pt}
    \item CCOD holds: $\delta(T_i,X_i)=0$, a constant;
    \item Weak violation of CCOD: $\delta(T_i, X_i)=7I(T_i\not=0)$;
    \item Strong violation of CCOD:  $\delta(T_i, X_i)=20I(T_i\not=0)$. 
\end{itemize}

\subsection{Results}\label{subsec:results}
We report the simulation results for $n_k=3000$, $k=0,...,4$ under strongly heterogeneous covariate distributions and strong violation of CCOD in Figure \ref{fig:results_main}. Complete numerical results for all sample sizes, covariate shifts, outcome errors, and concept shifts can be found in Appendix \ref{subapx:completeRes}.

Figure \ref{fig:results_main} summarizes results for (A) marginal coverage, (B) prediction interval width, (C) conditional coverage (C), and (D) weights as a function of discrepancy $\chi^2_k = (\widehat{r}_0 - \widehat{r}_k)^2$ values over $500$ replications. Compared to the target only method, our federated method achieves nominal marginal coverage with tighter dispersion and less variability, shorter prediction interval widths that are close to the oracle interval width (red dashed line), relatively good conditional coverage, and informative weight metrics that indicate how source site quantiles $\widehat{r}_k$ are being weighted as a function of discrepancy compared to the target site quantile $\widehat{r}_0$. The pooled sample method has poor performance for ASR, with overly conservative marginal coverage, interval widths that are on average five times longer than our federated method, and conservative conditional coverage. The performance for local ASR is also poor, with below nominal marginal and conditional coverage. The conditional coverage plots indicate that (1) ASR is not robust, which is consistent with the findings in \citet{lei2018distribution}); (2) both CQR and local ASR have better performance in terms of local coverage, and the results for the target only and our federated method perform similarly with $0.9$ nominal coverage level for many values of $X$. Full conditional coverage plots for all cases are provided in the Appendix (see Figure \ref{fig:localCov}).

\begin{figure*}[htbp!]
    \centering
    \includegraphics[width=0.92\textwidth]{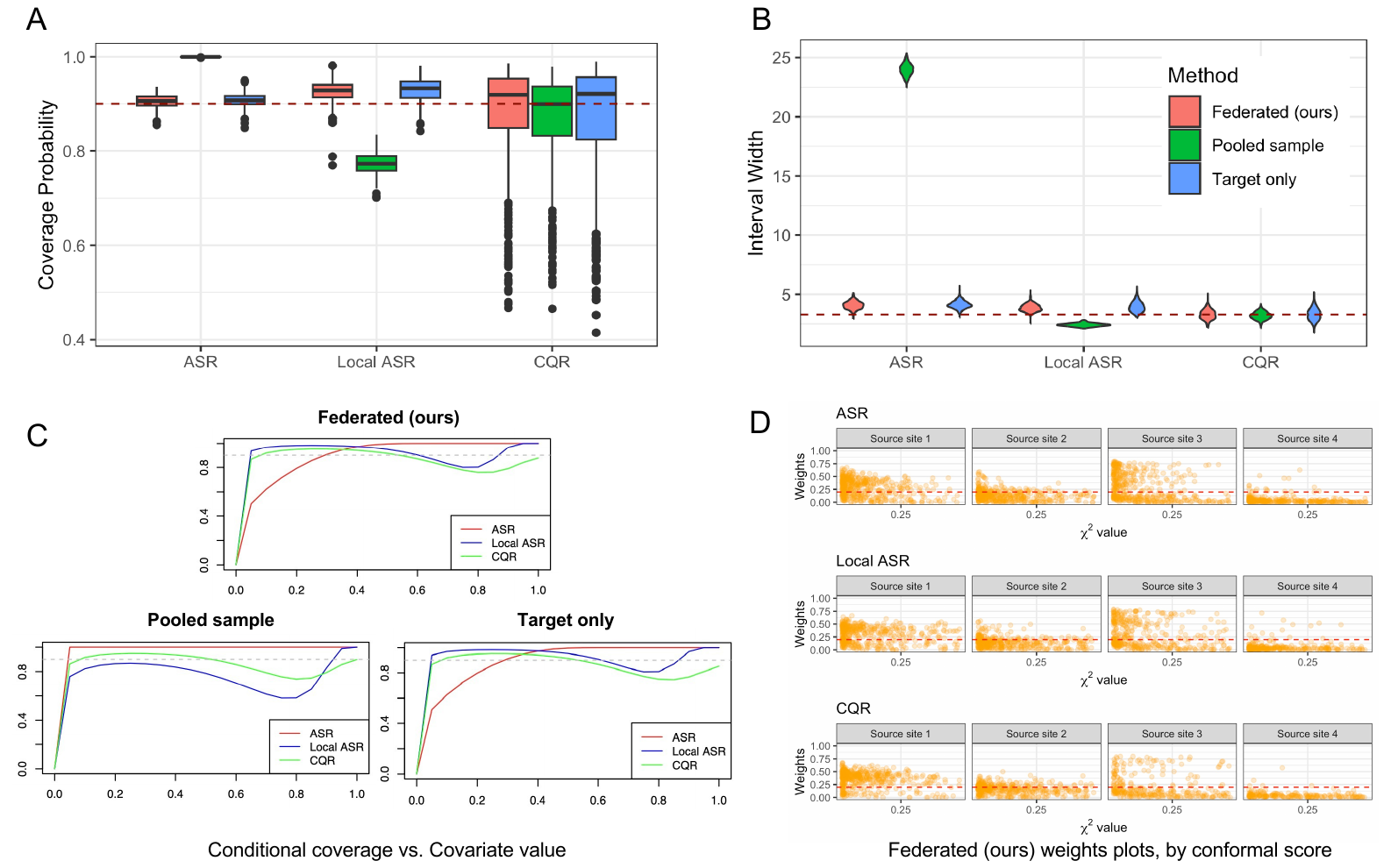}
    \caption{A: Marginal coverage, B: Prediction interval width, C: Conditional coverage, and D: Weights for our proposed federated method compared to the pooled sample and target only methods, where sample size $n_k=3000$, $k=0,...,4$ under strongly heterogeneous covariate distributions and strong violation of CCOD. } 
    \label{fig:results_main}
\end{figure*}

\section{Data Application}\label{sec:data}
Congenital heart defects (CHD) are the most prevalent birth defects in the United States, and over  $40{,}000$ surgeries for CHD are performed each year \citep{pasquali2016congenital}. Prolonged hospital length of stay (LOS) post-surgery places a significant financial burden on families and health care systems and is associated with postoperative morbidity. Moreover, LOS varies geographically, likely due to practice and patient heterogeneity. We utilize data from the Society of Thoracic Surgeons' Congenital Heart Surgery Database (STS-CHSD) which includes audited preoperative, intraoperative, and early postoperative information \citep{overman2019ten} from U.S. congenital heart surgery centers. We identified all Norwood surgeries, which are palliative surgeries for patients with CHD, occurring between January 2016 and June 2022. We used the index operative encounter during a given admission as the unit of observation. There were a total of $3{,}457$ observations, with a median LOS of $40$ days $(\textrm{min: 2}, \textrm{max: }183)$ and $752$ $(21.2\%)$ missing values for LOS. 

Our goal is to provide prediction intervals for LOS for patients in target sites with missing values of LOS. The target site is defined to be one of four mutually exclusive geographic regions according to the U.S. Census Bureau: (i) South, (ii) Midwest,  (iii) West, and (iv) Northeast \citep{census2020}.  We included as confounders demographic factors (e.g. age, race/ethnicity, sex, birthweight, birth height, etc.), genetic syndromes, chromosomal abnormalities, non-cardiac anomalies, pre-operative factors, and a variety of Norwood procedure-specific factors found in the STS-CHSD \citep{tabbutt2012risk}. While the MAR assumption is not testable, it is more likely to be valid in settings such as ours where a rich set of potential confounders are measured prospectively.

Figure \ref{fig:4panels} displays the prediction intervals for hospital LOS following a Norwood procedure for four randomly selected individuals, one in each region, across $\alpha=\{0.1,0.2,0.3,0.4,0.5\}$ and conformal scores $\in \{\textrm{ASR, local ASR, CQR}\}$. For example, using our proposed method and CQR as the conformal score for patient B in the Midwest region, with at least $50\%$ probability, the expected LOS is between $24.3$ to $39.9$ days $(\alpha=0.5)$. Our method generally produces tighter prediction intervals than the target only method of \cite{yang2024doubly}, and the advantage can be practically informative. For example, using local ASR for patient C in the South region, the $80\%$ prediction interval is over $30$ days shorter using our method versus the target only method. The pooled sample method performs similarly to our federated method, suggesting that data-adaptive inference may be nearly as efficient as under full CCOD in this data application.

\begin{figure*}[htbp!]
    \centering
    \includegraphics[width=0.95\textwidth]{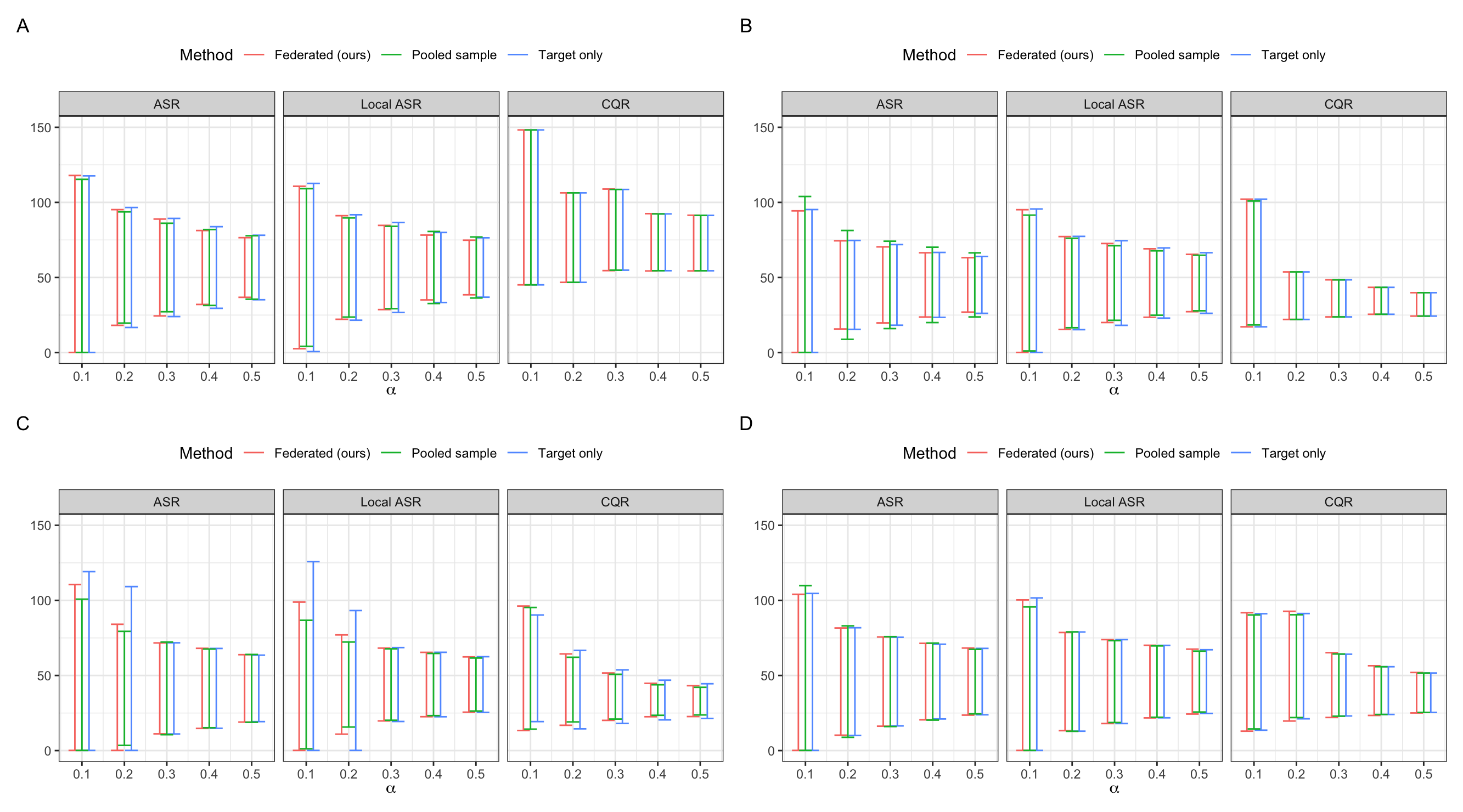}
    \caption{Each panel represents the prediction intervals for hospital LOS for a randomly selected individual following a Norwood procedure across $\alpha=\{0.1,0.2,0.3,0.4,0.5\}$ and conformal score $\in \{\textrm{ASR, local ASR, CQR}\}$ for A: a patient in South, B: a patient in Midwest, C: a patient in West, D: a patient in Northeast. }
    \label{fig:4panels}
\end{figure*}

\section{Discussion}\label{sec:discu}

We proposed a data-driven and distribution-free prediction method to obtain valid prediction intervals for missing outcome data in a target site while exploiting information from multiple potentially heterogeneous sites due to distribution shifts. Our proposal shares the marginal coverage properties of conformal prediction methods and builds on modern semiparametric efficiency theory and federated learning for more robust and efficient uncertainty quantification. When subjects from different data sources are similar, such that one may be willing to assert that the outcome distributions are shared, we derive the efficient influence function leveraging all data sources. In some practical settings, it would be unreasonable to assume that the conditional outcome distribution is the same across sites, i.e., some source sites may provide relevant information for constructing prediction intervals for the target site, whereas other sites may not. In such scenarios, we present a novel approach that combines information from target and source sites in a data-adaptive manner.

Among the three types of conformal scores that we studied, we provide the following recommendations for practitioners. When the sample size is small, e.g., 300 or fewer, we suggest using local ASR, which is more robust against heteroscedasticity compared to ASR and more efficient than CQR, which on average requires larger sample sizes to attain nominal coverage. When sample sizes are larger, CQR provides coverage probabilities close to the nominal level. 

An interesting line of future research concerns the development of covariate-adaptive ensemble weights for aggregating information from multiple sources of data. We conjecture that covariate-adaptive methods could produce prediction intervals that are as efficient as an oracle with knowledge of the optimal prediction interval, although we leave this for future work. Another direction for development is to formalize the framework through a sensitivity analysis approach when the CCOD assumption is violated. There are multiple options for sensitivity analysis, e.g. those working off of the Rosenbaum selection model such as \citet{jin2023sensitivity} and \citet{yin2024conformal}, or through a sensitivity parameter encoding a hypothetical departure from the CCOD assumption via a semiparametric approach \citep{robins2000sensitivity}. Challenges to overcome would be in the estimation of nuisance functions in this case.

\section{Impact Statement}
This paper presents work whose goal is to advance the field of conformal prediction and its applications to precision medicine. There are many potential societal consequences of our work, none of which we feel must be specifically highlighted here.

\section*{Software and Data}

We provide a user-friendly R function \texttt{MuSCI()} implementing the proposed method with an illustrative example, available at: \url{https://github.com/yiliu1998/Multi-Source-Conformal}. 

\section*{Acknowledgements}
This work was supported, in part, by Grant HL5R01HL162893 from the National Heart, Lung, and Blood Institute from the US National Institutes of Health. The data analyzed in this study were provided to the investigators through The Society of Thoracic Surgeons' Task Force on Funded Research Program. The authors thank Sara Pasquali, Meena Nathan, John Mayer, Jr., and Katya Zelevinsky for helpful discussions related to clinical features of CHD and surgical quality.






\bibliography{refs}
\bibliographystyle{icml2024}

\newpage
\appendix
\onecolumn
\section{Technical Details}\label{apx:tech}
\subsection{Proof of Theorem~\ref{thm:IF-ccod}}
Recall from \citet{bkrw1993} and \citet{vandervaart2002} that an influence function $\dot{\chi}(\mc O;\mathbb{P})$ of a functional $\chi(\mathbb{P})$ is a mean-zero finite variance function satisfying the following criterion:
    \[\left. \frac{d}{d\epsilon}\chi(\mathbb{P}_{\epsilon})\right|_{\epsilon = 0} = \mathbb{E}_{\mathbb{P}}\left(\dot{\chi}(\mc O;\mathbb{P}) u(\mc O)\right),\]
    for any regular parametric submodel $\{\mathbb{P}_{\epsilon}: \epsilon \in [0,1)\}$ such that $\mathbb{P}_0 \equiv \mathbb{P}$ with score function $u(\mc O) = \left. \frac{d}{d\epsilon} \log{d \mathbb{P}_{\epsilon}} \right|_{\epsilon = 0}$.
The semiparametric efficient influence function is the unique such function belonging to the tangent space, $\Lambda_{\mathbb{P}}$, which is the closure of the linear span of all scores of regular parametric submodels through $\mathbb{P}$. To find an influence function, we take such a generic submodel, and differentiate an identifying estimating equation with respect to $\epsilon$. Recall that
\[1-\alpha = \mathbb{E}_{\mathbb{P}}(\mathbb{P}(S(\bX,Y) \leq r_0(\alpha)(\mathbb{P}) \mid T=0, \bX, R=1) \mid T=0, R = 0),\]
which holds under Assumptions~\ref{ass:MAR} and \ref{ass:pos}. Under Assumption~\ref{ass:ccod}, we may instead write 
\[1-\alpha = \mathbb{E}_{\mathbb{P}}(\mathbb{P}(S(\bX,Y) \leq r_0(\alpha)(\mathbb{P}) \mid \bX, R=1) \mid T=0, R = 0),\]
since CCOD and MAR together imply $(R, T) \independent Y \mid \bX$. Thus, we have
\begin{align*}
    0 &= \left.\frac{d}{d\epsilon}\mathbb{E}_{\mathbb{P}_{\epsilon}}(\mathbb{P}_{\epsilon}(S(\bX,Y) \leq r_0(\alpha)(\mathbb{P}_{\epsilon}) \mid \bX, R=1) \mid T=0, R = 0)\right|_{\epsilon = 0} \\
    &= \left.\frac{d}{d\epsilon}\mathbb{E}_{\mathbb{P}_{\epsilon}}(\mathbb{P}(S(\bX,Y) \leq r_0(\alpha)(\mathbb{P}) \mid \bX, R=1) \mid T=0, R = 0)\right|_{\epsilon = 0} \\
    & \quad \quad + \left.\frac{d}{d\epsilon}\mathbb{E}_{\mathbb{P}}(\mathbb{P}_{\epsilon}(S(\bX,Y) \leq r_0(\alpha)(\mathbb{P}) \mid \bX, R=1) \mid T=0, R = 0)\right|_{\epsilon = 0} \\
    & \quad \quad + \left.\frac{d}{d\epsilon}\mathbb{E}_{\mathbb{P}}(\mathbb{P}(S(\bX,Y) \leq r_0(\alpha)(\mathbb{P}_{\epsilon}) \mid \bX, R=1) \mid T=0, R = 0)\right|_{\epsilon = 0} 
\end{align*}
Before proceeding, let  $u_{B \mid C}$ be the conditional score function for $B$ given $C$, for arbitrary $B$ and $C$, and note the key properties that (i) $\mathbb{E}_{\mathbb{P}}(u_{B \mid C} \mid C) = 0$, and (ii) $u_{B, C} = u_{B \mid C} + u_C$. Now, for the first of the above three terms, we have
\begin{align*}
    & \left.\frac{d}{d\epsilon}\mathbb{E}_{\mathbb{P}_{\epsilon}}(\overline{m}(r_0, \bX) \mid T=0, R = 0)\right|_{\epsilon = 0} \\
    &= \mathbb{E}_{\mathbb{P}}(\{\overline{m}(r_0, \bX) - (1 - \alpha)\}u_{\bX \mid T = 0, R = 0} \mid T=0, R = 0) \\
    &= \mathbb{E}_{\mathbb{P}}\left(\frac{I(T = 0, R = 0)}{\mathbb{P}[T = 0, R = 0]}\{\overline{m}(r_0, \bX) - (1 - \alpha)\}u_{\bX \mid T, R} \right) \\
    &= \mathbb{E}_{\mathbb{P}}\left(\frac{I(T = 0, R = 0)}{\mathbb{P}[T = 0, R = 0]}\{\overline{m}(r_0, \bX) - (1 - \alpha)\}u(\mathcal{O}) \right),
\end{align*}
where in the last equality we are able to add in $u_{T, R}$ since $I(T = 0, R = 0)\{\overline{m}(r_0, \bX) - (1 - \alpha)\}$ has mean zero given $(T, R)$ by construction, and we can add in $u_{RY \mid \bX, T, R}$ since this has mean zero given $(\bX, T, R)$. Similarly, for the second term above, we have
\begin{align*}
    & \left.\frac{d}{d\epsilon}\mathbb{E}_{\mathbb{P}}(\overline{m}_{\epsilon}(r_0, \bX) \mid T=0, R = 0)\right|_{\epsilon = 0} \\
    & = \mathbb{E}_{\mathbb{P}}\left(\mathbb{E}_{\mathbb{P}}(\{I(S(\bX, Y) \leq r_0) - \overline{m}(r_0, \bX)\}u_{Y \mid \bX, R = 1} \mid \bX, R = 1) \mid T = 0, R = 0\right) \\
    &= \mathbb{E}_{\mathbb{P}}\left(\frac{I(T = 0, R = 0)}{\mathbb{P}[T = 0, R = 0]}\mathbb{E}_{\mathbb{P}}\left(\frac{R}{\mathbb{P}[R = 1 \mid \bX]}\{I(S(\bX, Y) \leq r_0) - \overline{m}(r_0, \bX)\}u_{RY \mid \bX, R} \mid \bX\right)\right) \\
    &= \mathbb{E}_{\mathbb{P}}\left(\frac{\mathbb{P}[T = 0 \mid \bX, R = 0]}{\mathbb{P}[T = 0, R = 0]}\frac{\mathbb{P}[R = 0 \mid \bX]}{\mathbb{P}[R = 1 \mid \bX]}R\{I(S(\bX, Y) \leq r_0) - \overline{m}(r_0, \bX)\}u(\mathcal{O})\right).
\end{align*}

Finally, for the third term above, we have
\begin{align*}
    & \left.\frac{d}{d\epsilon}\mathbb{E}_{\mathbb{P}}(\overline{m}(r_{0}(\alpha)(\mathbb{P}_{\epsilon}), \bX) \mid T=0, R = 0)\right|_{\epsilon = 0} \\
    &= \mathbb{E}_{\mathbb{P}}\left(p_{S \mid \bX, R = 1}(r_0, \bX) \mid T = 0, R = 0\right) \left.\frac{d}{d\epsilon}r_{0}(\alpha)(\mathbb{P}_{\epsilon})\right|_{\epsilon = 0},
\end{align*}
where $p_{S \mid \bX, R = 1}(r_0, \bX)$ is the conditional density of $S(\bX, Y)$ given $\bX, R = 1$, evaluated at $r_0$. Rearranging the original differentiated estimating equation, we have
\[\left.\frac{d}{d\epsilon}r_{0}(\alpha)(\mathbb{P}_{\epsilon})\right|_{\epsilon = 0} = \mathbb{E}_{\mathbb{P}}(\dot{r}_0^{\mathrm{CCOD}}(\mathcal{O}; \mathbb{P}) u(\mathcal{O})),\]
where $\dot{r}_0^{\mathrm{CCOD}}(\mathcal{O}; \mathbb{P}) = -\{\mathbb{P}[T = 0, R = 0] \mathbb{E}_{\mathbb{P}}\left(p_{S \mid \bX, R = 1}(r_0, \bX) \mid T = 0, R = 0\right)\}^{-1}\varphi^{\mathrm{CCOD}}(\mathcal{O}; r_0, \overline{m}, \overline{\eta}, q_0)$, and $\varphi^{\mathrm{CCOD}}$ is as defined in Section~\ref{subsec:ccod}. By Lemma 24 of \citet{rotnitzky2020}, the tangent space of
the semiparametric model at $\mathbb{P}$ is $\Lambda_{\mathbb{P}} = \Lambda_{\bX} \oplus \Lambda_{T \mid \bX} \oplus \Lambda_{R \mid \bX, T} \oplus \Lambda_{RY \mid \bX, R}$, where $\Lambda_{B \mid C} = \{g(B, C) \in L_2(\mathbb{P}): \mathbb{E}(g \mid C) = 0\}$. It is straightforward to verify that $\dot{r}_0^{\mathrm{CCOD}}(\mathcal{O}; \mathbb{P}) \in \Lambda_{\mathbb{P}}$, so it is the semiparametric efficient influence function under Assumption~\ref{ass:ccod}.

\subsection{Proof of Theorem~\ref{thm:cov-ccod}}
Write $\mathbb{P}(f) = \mathbb{E}_{\mathbb{P}}(f(\mathcal{O}) \mid D^n)$, for any $f$. Observe that for any $r$, possibly a function of training data $D^n$, and for $\mathcal{O} \independent D^n$,
\begin{align*}
&\mathbb{P}\left(\varphi^{\mathrm{CCOD}}(\mathcal{O}; r, \widehat{\overline{m}}, \widehat{\overline{\eta}}, \widehat{q}_0)\right) \\
&= \mathbb{P}\left(\mathbb{P}[T = 0, R = 0 \mid \bX]\left\{\widehat{\overline{m}}(r, \bX) - (1 - \alpha)\right\} + \mathbb{P}[R = 1 \mid \bX] \widehat{\overline{\eta}}(\bX)\widehat{q}_0(\bX)\left\{\overline{m}(r, \bX) - \widehat{\overline{m}}(r, \bX)\right\}\right) \\
&= \mathbb{P}\left(\mathbb{P}[R = 1 \mid \bX]\left[\left\{q_0(\bX)\overline{\eta}(\bX) - \widehat{q}_0(\bX)\widehat{\overline{\eta}}(\bX)\right\}\left\{\widehat{\overline{m}}(r, \bX) - \overline{m}(r, \bX)\right\} + q_0(\bX)\overline{\eta}(\bX) \left\{\overline{m}(r, \bX) - (1 - \alpha)\right\}\right] \right).
\end{align*}
Thus, omitting inputs, we have
\begin{equation}\label{eq:prod-bias}
\mathbb{P}\left(\varphi^{\mathrm{CCOD}}(\mathcal{O}; r, \widehat{\overline{m}}, \widehat{\overline{\eta}}, \widehat{q}_0) - \varphi^{\mathrm{CCOD}}(\mathcal{O}; r, \overline{m}, \overline{\eta}, q_0)\right) = \mathbb{P}\left(\mathbb{P}[R = 1 \mid \bX]\left\{q_0\overline{\eta} - \widehat{q}_0\widehat{\overline{\eta}}\right\}\left\{\widehat{\overline{m}}(r, \cdot) - \overline{m}(r, \cdot)\right\}\right)
\end{equation}
On the other hand, by definition,
\begin{align*}
    &\mathbb{P}[Y \in \widehat{C}_{\alpha}^{\mathrm{CCOD}}(\bX) \mid T = 0, R = 0, D^n] - (1 - \alpha) \\
    &= \mathbb{P}[S(\bX, Y) \leq \widehat{r}^{\mathrm{CCOD}} \mid T = 0, R = 0, D^n] - (1 - \alpha)\\
    &= \mathbb{E}_{\mathbb{P}}\left(\overline{m}(\widehat{r}^{\mathrm{CCOD}}, \bX) - (1 - \alpha)\mid T = 0, R = 0, D^n\right) \\
    &= \frac{\mathbb{E}_{\mathbb{P}}\left(\mathbb{P}[R = 1 \mid \bX] q_0(\bX)\overline{\eta}(\bX)\left\{\overline{m}(\widehat{r}^{\mathrm{CCOD}}, \bX) - (1 - \alpha)\right\}\mid D^n\right)}{\mathbb{P}[T = 0, R = 0]} \\
    &= \frac{\mathbb{P}\left(\varphi^{\mathrm{CCOD}}(\mathcal{O}; \widehat{r}^{\mathrm{CCOD}}, \overline{m}, \overline{\eta}, q_0)\right)}{\mathbb{P}[T = 0, R = 0]}
\end{align*}
Finally, we decompose
\begin{align*}
&\mathbb{P}\left(\varphi^{\mathrm{CCOD}}(\mathcal{O}; \widehat{r}^{\mathrm{CCOD}}, \overline{m}, \overline{\eta}, q_0)\right) \\
&= \mathbb{P}_n\left(\varphi^{\mathrm{CCOD}}(\mathcal{O}; \widehat{r}^{\mathrm{CCOD}}, \widehat{\overline{m}}, \widehat{\overline{\eta}}, \widehat{q}_0)\right)
\\ 
& \quad \quad - (\mathbb{P}_n - \mathbb{P})\left(\varphi^{\mathrm{CCOD}}(\mathcal{O}; \widehat{r}^{\mathrm{CCOD}}, \widehat{\overline{m}}, \widehat{\overline{\eta}}, \widehat{q}_0)\right) \\
& \quad \quad - \mathbb{P}\left(\varphi^{\mathrm{CCOD}}(\mathcal{O}; \widehat{r}^{\mathrm{CCOD}}, \widehat{\overline{m}}, \widehat{\overline{\eta}}, \widehat{q}_0) - \varphi^{\mathrm{CCOD}}(\mathcal{O}; \widehat{r}^{\mathrm{CCOD}}, \overline{m}, \overline{\eta}, q_0)\right),
\end{align*}
By construction of $\widehat{r}^{\mathrm{CCOD}}$, the first term above is 0, while the third term is $O_{\mathbb{P}}(R_n)$ by the product bias in~\eqref{eq:prod-bias} and boundedness of $q_0, \overline{\eta}, \widehat{q}_0, \widehat{\overline{\eta}}$. By our assumptions about $\widehat{\overline{m}}$ (monotonicity and boundedness) and $(\widehat{\overline{\eta}}, \widehat{q}_0)$ (boundedness), we can note that $\{\varphi^{\mathrm{CCOD}}(\, \cdot \,; r, \widehat{\overline{m}}, \widehat{\overline{\eta}}, \widehat{q}_0) : r \in \mathbb{R}\}$ is a Donsker class (using similar arguments to Theorem 2 in \citet{yang2024doubly}), so that the second term is $O_{\mathbb{P}}(n^{-1/2})$. Combining these yields the result.

\subsection{Proof of Theorem~\ref{thm:IF-partial-ccod}}\label{subapx:IF-der}
Here, we derive the influence function for a non-target source site $k$, used in Section \ref{subsec:hetedis}, making the working assumption of a common conditional outcome distribution between the target site and site $k$, $p(Y\mid\bX, T=0) = p(Y\mid\bX, T=k)$. Note that our data-adaptive method weights source sites that can violate CCOD; we use this working partial CCOD assumption only to derive the form of the efficient influence function to facilitate downstream analysis. Our derivation is very similar to that in the proof of Theorem~\ref{thm:IF-ccod}. To begin, observe that
\begin{align*}
    1 - \alpha & = \mathbb E_{\mathbb{P}}\left\{{\mathbb{P}}[S(\bX, Y)\leq r_0(\alpha)({{\mathbb{P}}})\mid\bX, T=0, R=1]\mid T=0, R=0\right\}\\
    & = \mathbb E_{\mathbb{P}}\left\{\mathbb P[S(\bX, Y)\leq r_0(\alpha)({\mathbb{P}})\mid\bX, T=k, R=1]\mid T=0, R=0\right\}. 
\end{align*}
In addition, 
\begin{align}
    0 & = \frac{\partial}{\partial\epsilon}(1-\alpha)\Big|_{\epsilon=0}\nonumber \\ 
    & = \frac{\partial}{\partial\epsilon}\mathbb E_{{\mathbb{P}}_\epsilon}\left\{m_{k,\epsilon}(r_0(\alpha)({\mathbb{P}}_\epsilon),\bX)\mid T=0, R=0\right\}\Big|_{\epsilon=0}\nonumber\\
    & = \mathbb E_{\mathbb{P}}\left\{[m_k(r_0(\alpha)({\mathbb{P}}),\bX)-(1-\alpha)]u_{\bX\mid T=0, R=0}\mid T=0, R=0\right\}\label{eq:IF-der-1}\\
    & \quad + \mathbb E_{\mathbb{P}}\big\{\mathbb E_{\mathbb{P}}\{[I(S(\bX, Y)\leq r_0(\alpha)({\mathbb{P}})) - m_k(r_0(\alpha)({\mathbb{P}}), \bX)]u_{Y\mid \bX, T=k, R=1}\mid\bX, T=k, R=1\}\mid T=0, R=0\big\}\label{eq:IF-der-2}\\
    & \quad + \underbrace{\mathbb E_{\mathbb{P}}\{f_{S\mid \bX, T=k, R=1}(r_0(\alpha)({\mathbb{P}})\mid\bX, T=k, R=1)\mid T=0, R=0\}}_{C_{k,0}({\mathbb{P}})}\cdot \frac{\partial}{\partial\epsilon}r_0(\alpha)({\mathbb{P}}_\epsilon)\Big|_{\epsilon=0}, \nonumber
\end{align}
where $f_{S\mid \bX, T=k, R=1}$ is the conditional density function of $S(\bx, y)$, i.e., the derivative of $m_k$. 

Furthermore, we can write,
\begin{align*}
    \eqref{eq:IF-der-1} & = \mathbb E_{\mathbb{P}}\left\{\frac{I(T=0, R=0)}{\mathbb {\mathbb{P}}(T=0, R=0)}[m_k(r_0(\alpha)({\mathbb{P}}),\bX)-(1-\alpha)]u(0)\right\},\\
    \eqref{eq:IF-der-2} & =\mathbb E_{\mathbb{P}}\left\{\frac{I(T=0, R=0)}{\mathbb P(T=0, R=0)}\mathbb E_{\mathbb{P}}\left(\frac{I(T=k, R=1)}{\mathbb P(T=k, R=1\mid\bX)}[I(S(\bX, Y)\leq r_0(\alpha)({\mathbb{P}}))-m_k(r_0(\alpha)({\mathbb{P}}),\bX)]u_{Y\mid\bX,T,R}\mid\bX\right)\right\}\\
    & = \mathbb E_{\mathbb{P}}\left\{\frac{I(T=0, R=0)}{\mathbb P(T=0, R=0)}\frac{\mathbb P(T=0, R=0\mid\bX)}{\mathbb P(T=k, R=1\mid\bX)}[I(S(\bX, Y)\leq r_0(\alpha)({\mathbb{P}}))-m_k(r_0(\alpha)({\mathbb{P}}),\bX)]u(0)\right\},
\end{align*}
by the tower law. Therefore, rearranging the terms, we can obtain
\begin{align*}
    \frac{\partial}{\partial\epsilon}r_0(\alpha)({\mathbb{P}}_\epsilon)\Big|_{\epsilon=0} = -C_{k,0}({\mathbb{P}})^{-1}\{\eqref{eq:IF-der-1} + \eqref{eq:IF-der-2}\}.
\end{align*}
Therefore, an influence function of $r_0(\alpha)(\cdot)$ at ${\mathbb{P}}$ is
\begin{align*}
    \dot{r}_0(\alpha)(\mc O;{\mathbb{P}}) &= -\frac{C_{k,0}({\mathbb{P}})^{-1}}{\mathbb P(T=0, R=0)}\bigg\{I(T=0,R=0)[\underbrace{m_k(r_0(\alpha)({\mathbb{P}}),\bX)}_{=m_0\text{ under our assumption}}-(1-\alpha)] \\
    & \quad\quad\quad + I(T=k,R=1)\frac{\mathbb P(T=0, R=0\mid\bX)}{\mathbb P(T=k, R=1\mid\bX)}[I(S(\bX, Y)\leq r_0(\alpha)({\mathbb{P}})-m_k(r_0(\alpha)({\mathbb{P}}),\bX))]\bigg\} \\
    & = \underbrace{-\frac{C_{k,0}({\mathbb{P}})^{-1}}{\mathbb P(T=0, R=0)}}_{\text{a probability constant}}\varphi_k(\mc O;r_0,m_0,m_k,\omega_{k,0}).
\end{align*}
Observe that, by Bayes' rule, 
$$
\frac{\mathbb P(T=0,R=0\mid\bX)}{\mathbb P(T=k, R=1\mid\bX)} = \underbrace{\frac{\mathbb P(\bX\mid T=0,R=0)}{\mathbb P(\bX\mid T=k, R=1)}}_{\omega_{k,0}(\bX)}\cdot\frac{\mathbb P(T=0,R=0)}{\mathbb P(T=k, R=1)}.
$$
Hence, we can work with
\begin{align*}
    \varphi_k(\mc O;\theta,m_0,m_k,\omega_{k,0}) = \frac{I(T=0,R=0)}{\mathbb P(T=0,R=0)}[m_0(\theta,\bX) - (1-\alpha)] + \frac{I(T=k,R=1)}{\mathbb P(T=k,R=1)}\omega_{k,0}(\bX)[I(S(\bX, Y)\leq\theta) -m_k(\theta,\bX)]. 
\end{align*}

\subsection{Proof of Theorem~\ref{thm:cov-oracle}}

Write $r_{w^*} = \sum_{k=0}^{K-1}w_k \widehat{r}_k$. By
construction of $\varphi_j$, we have
\begin{align}
  \begin{split}\label{eq:decomp}
    \mathbb{P}[Y \in \widehat{C}_{\alpha}^{w^*}(\bX) \mid T = 0, R = 0, D^n] - (1 - \alpha)
    &= \mathbb{P}\left[S(\bX, Y) \leq  r_{w^*} \, \middle| \, T = 0, R = 0, D^n\right] - (1 - \alpha)\\
    &= \frac{\mathbb{P}\left(\varphi_j(\mathcal{O}; r_{w^*}, m_0, m_j, \omega_{j,0})\right)}{\mathbb{P}[T = 0, R = 0]},
  \end{split}
\end{align}
where the last equality holds for all $j \in \mathcal{S}^*$, and the
numerator could also be replaced by
$\mathbb{P}\left(\varphi_0(\mathcal{O}; r_{w^*}, m_0,
  \eta_0)\right)$. Now, see that for any $j \in \mathcal{S}^*$,
\begin{align*}
  &\mathbb{P}\left(\varphi_j(\mathcal{O}; r_{w^*}, m_0, m_j, \omega_{j,0})\right) \\
  &= \mathbb{P}\left(\varphi_j(\mathcal{O}; r_{w^*}, m_0, m_j, \omega_{j,0}) -
    \varphi_j(\mathcal{O};  \widehat{r}_j, m_0, m_j, \omega_{j,0})\right)
    + \mathbb{P}\left(\varphi_j(\mathcal{O}; \widehat{r}_j, m_0, m_j, \omega_{j,0})\right) \\
  &= \mathbb{P}\left(\frac{I(T = 0, R = 0)}{\mathbb{P}(T = 0, R = 0)}
    \left\{m_0(r_{w^*}, \bX) - m_0(\widehat{r}_j, \bX)\right\}\right) +
    \mathbb{P}\left(\varphi_j(\mathcal{O}; \widehat{r}_j, m_0, m_j, \omega_{j,0})\right).
\end{align*}
Further, as in the proof of Theorem~\ref{thm:cov-ccod}, we can decompose the latter term as
\begin{align*}
  \mathbb{P}\left(\varphi_j(\mathcal{O}; \widehat{r}_j, m_0, m_j, \omega_{j,0})\right)
  &= \mathbb{P}_n\left(\varphi_j(\mathcal{O}; \widehat{r}_j,
    \widehat{m}_0, \widehat{m}_j, \widehat{\omega}_{j,0})\right) \\
  & \quad \quad - (\mathbb{P}_n - \mathbb{P})\left(\varphi_j(\mathcal{O}; \widehat{r}_j,
    \widehat{m}_0, \widehat{m}_j, \widehat{\omega}_{j,0})\right) \\
  & \quad \quad - \mathbb{P}\left(\varphi_j(\mathcal{O}; \widehat{r}_j,
    \widehat{m}_0, \widehat{m}_j, \widehat{\omega}_{j,0}) -
    \varphi_j(\mathcal{O}; \widehat{r}_j, m_0, m_j, \omega_{j,0})\right).
\end{align*}
By construction of $\widehat{r}_j$, the first term in this sum is 0,
the second $O_{\mathbb{P}}(n^{-1/2})$ because
$\{\varphi_j(\, \cdot\,; r, \widehat{m}_0, \widehat{m}_j,
\widehat{\omega}_{j,0}): r \in \mathbb{R}\}$ is a Donsker class under
our assumptions, and the third term is
$O_{\mathbb{P}}(R_{n,j}^* + n^{-1/2})$, where
\[R_{n,j}^* = \sup_{r} \lVert \widehat{m}_0(r, \cdot) -
  \widehat{m}_j(r, \cdot) \rVert + \lVert \widehat{\omega}_{j, 0} -
  \omega_{j, 0}\rVert \cdot \sup_{r} \lVert \widehat{m}_j(r, \cdot) -
  m_j(r, \cdot) \rVert,\] since, assuming $\mathbb{P}[T = 0, R = 0]$
is estimated by sample means in the training data (so that
$\widehat{\mathbb{P}}[T = 0, R = 0] - \mathbb{P}[T = 0, R = 0] =
O_{\mathbb{P}}(n^{-1/2})$), we have for any $r$ possibly dependent on
$D^n$,
\begin{align*}
  &\mathbb{P}\left(\varphi_j(\mathcal{O}; r, \widehat{m}_0, \widehat{m}_j, \widehat{\omega}_{j,0}) - \varphi_j(\mathcal{O}; r, m_0, m_j, \omega_{j,0})\right) \\
  &= \mathbb{P}\left(\frac{I(T=0,R=0)}{\mathbb P(T=0,R=0)}[\widehat{m}_0(r,\cdot) - m_0(r, \cdot)] + \frac{I(T=j,R=1)}{\mathbb P(T=j,R=1)}\widehat{\omega}_{j,0}[m_j(r, \cdot)-\widehat{m}_j(r,\cdot)]\right) + O_{\mathbb{P}}(n^{-1/2})\\
  &= \mathbb{P}\left(\frac{\mathbb{P}[T=j, R = 1 \mid \bX]}{\mathbb{P}[T = j, R = 1]}\left\{\omega_{j,0}[\widehat{m}_0(r,\cdot) - m_0(r, \cdot)] + \widehat{\omega}_{j,0}[m_j(r, \cdot)-\widehat{m}_j(r,\cdot)]\right\}\right) + O_{\mathbb{P}}(n^{-1/2}) \\
  &= O_{\mathbb{P}}\left(R_{n,j}^* + n^{-1/2}\right).
\end{align*}

For the target site,
\begin{align*}
  &\mathbb{P}\left(\varphi_0(\mathcal{O}; r_{w^*}, m_0, \eta_0)\right) \\
  &= \mathbb{P}\left(\varphi_0(\mathcal{O}; r_{w^*}, m_0, \eta_0) -
    \varphi_0(\mathcal{O};  \widehat{r}_0, m_0, \eta_0)\right)
    + \mathbb{P}\left(\varphi_0(\mathcal{O}; \widehat{r}_0, m_0, \eta_0)\right) \\
  &= \mathbb{P}\left(\frac{I(T = 0, R = 0)}{\mathbb{P}(T = 0, R = 0)}
    \left\{m_0(r_{w^*}, \bX) - m_0(\widehat{r}_0, \bX)\right\}\right) +
    \mathbb{P}\left(\varphi_0(\mathcal{O}; \widehat{r}_0, m_0, \eta_0)\right)
\end{align*}
and by Theorem 3 in \citet{yang2024doubly},
$\mathbb{P}\left(\varphi_0(\mathcal{O}; \widehat{r}_0, m_0,
  \eta_0)\right) = O_{\mathbb{P}}\left(R_{n,0}^* + n^{-1/2}\right)$,
where
\[R_{n,0}^* = \lVert \widehat{\eta}_0 - \eta_0 \rVert \sup_{r} \lVert
  \widehat{m}_0(r, \cdot) - m_0(r, \cdot) \rVert.\] It remains to
characterize
$\mathbb{P}\left(\frac{I(T = 0, R = 0)}{\mathbb{P}(T = 0, R = 0)}
  \left\{m_0(r_{w^*}, \bX) - m_0(\widehat{r}_j, \bX)\right\}\right)$,
for $j \in \mathcal{S}^* \cup\{0\}$: in Lemma~\ref{lemma:cdf-cont}, we
show that these terms are each
$O_{\mathbb{P}}(R_{n, j}^* + \sum_{k = 0}^{K - 1} w_k R_{n, k}^* +
n^{-1/2})$. Combining all these results, in view of~\eqref{eq:decomp},
we conclude that
\[\mathbb{P}[Y \in \widehat{C}_{\alpha}^{w^*}(\bX) \mid T = 0, R = 0,
  D^n] - (1 - \alpha) = O_{\mathbb{P}}\left(\min_{j \in \mathcal{S}^* \cup
      \{0\}} R_{n,j}^*+ \sum_{k = 0}^{K - 1} w_k R_{n, k}^* +
    n^{-1/2}\right) = O_{\mathbb{P}}\left(R_{n}^* + n^{-1/2}\right),\]
which completes the proof.

\begin{lemma}\label{lemma:cdf-cont}
  Let $F_0(r) = \mathbb{P}[S(\bX, Y) \leq r \mid T = 0, R = 0]$ for
  $r \in \mathbb{R}$, i.e., $F_0$ is the (marginal) cdf of the
  conformal score, given $T = 0, R = 0$. Suppose the conditions of
  Theorem~\ref{thm:cov-oracle} hold, as well as the following conditions:
  \begin{enumerate}[(i)]
  \item $F_0$ is $L$-Lipschitz in a neighborhood around $r_0$,
  \item $\widehat{r}_j \overset{\mathbb{P}}{\to} r_0$,
    $\sup_{r}\lVert \widehat{m}_j(r, \cdot) - m_j(r, \cdot)\rVert =
    o_{\mathbb{P}}(1)$, for all $j \in \mathcal{S}^*\cup\{0\}$,
    $\lVert\widehat{\eta}_0 - \eta_0\rVert = o_{\mathbb{P}}(1)$, and
    $\lVert \widehat{\omega}_{j,0} - \omega_{j,0}\rVert =
    o_{\mathbb{P}}(1)$ for all
    $j \in \mathcal{S}^*$, where the associated rates of convergence
    may be arbitrarily slow,
  \item The maps
    $r \mapsto \mathbb{P}\left(\varphi_j(O; r, m_0, m_j,
      \omega_{j,0})\right)$ for $j \in \mathcal{S}^*$,
    $r \mapsto \mathbb{P}\left(\varphi_0(O;m_0, \eta_0)\right)$ are
    differentiable at $r_0$, uniformly in the nuisance functions, the
    derivative matrices
    $\left.\frac{d}{d r} \mathbb{P}\left(\varphi_j(O; r, m_0, m_j,
        \omega_{j,0})\right)\right|_{r = r_0} \eqqcolon V_j(r_0; m_0,
    m_j, \omega_{j,0})$ and
    $\left.\frac{d}{d r} \mathbb{P}\left(\varphi_0(O; r, m_0,
        \eta_0)\right)\right|_{r = r_0} \eqqcolon V_0(r_0; m_0,
    \eta_0)$ are invertible,
    $V_j(r_0; \widehat{m}_0, \widehat{m}_j, \widehat{\omega}_{j,0})
    \overset{\mathbb{P}}{\to} V_j(r_0; m_0, m_j, \omega_{j,0})$ for
    $r \in \mathcal{S}^*$, and
    $V_0(r_0; \widehat{m}_0, \widehat{\eta_0})
    \overset{\mathbb{P}}{\to} V_0(r_0; m_0, \eta_0)$.
  \end{enumerate}

  Then
  \[\mathbb{P}\left(\frac{I(T = 0, R = 0)}{\mathbb{P}(T = 0, R = 0)}
      \left\{m_0(r_{w^*}, \bX) - m_0(\widehat{r}_j,
        \bX)\right\}\right) = O_{\mathbb{P}}\left(R_{n,j}^* + \sum_{k
        = 0}^{K - 1} w_k R_{n, k}^* + n^{-1/2}\right),\] for all
  $j \in \mathcal{S}^* \cup \{0\}$.
\end{lemma}
\begin{proof}
  Observe that, for $j \in \mathcal{S}^* \cup \{0\}$,
  \begin{align}\label{eq:lipschitz}
    \left|\mathbb{P}\left(\frac{I(T = 0, R = 0)}{\mathbb{P}(T = 0, R = 0)}
    \left\{m_0(r_{w^*}, \bX) - m_0(\widehat{r}_j, \bX)\right\}\right)\right| =
    \left|F_0(r_{w^*}) - F_0(\widehat{r}_j)\right| \lesssim |r_{w^*} - \widehat{r}_j|,
  \end{align}
  by condition (i). Since
  \begin{equation}\label{eq:diff-bound}
    |r_{w^*} - \widehat{r}_j| \leq |r_{w^*} - r_0| + |\widehat{r}_j -
    r_0| \leq |\widehat{r}_j - r_0| +\sum_{k = 0}^{K - 1}w_k
    |\widehat{r}_k - r_0|,
  \end{equation}
  it suffices to analyze
  $|\widehat{r}_j - r_0|$ for each $j \in \mathcal{S}^* \cup \{0\}$.

  As the function classes
  $\{\varphi_j(\, \cdot\,; r, \widehat{m}_0, \widehat{m}_j,
  \widehat{\omega}_{j,0}): r \in \mathbb{R}\}$ and
  $\{\varphi_0(\, \cdot\,; r, \widehat{m}_0, \widehat{\eta}_0): r \in
  \mathbb{R}\}$ are Donsker under the assumptions of
  Theorem~\ref{thm:cov-oracle}, conditions (ii) and (iii) permit
  application of Lemma 3 of \citet{kennedy2023} to obtain
  \[\widehat{r}_j - r_j = O_{\mathbb{P}}(n^{-1/2} + R_{n,j}^*),\]
  for all $j \in \mathcal{S}^* \cup \{0\}$. Combining this
  with~\eqref{eq:lipschitz} and~\eqref{eq:diff-bound} yields the
  result.
\end{proof}

\clearpage

\subsection{Details of Algorithm \ref{algo:robust}}\label{subapx:algo1full}

In this Appendix, we present all details of Algorithm \ref{algo:robust} in Section \ref{sec:meth} in the following algorithm table (Algorithm \ref{fullalgo:robust}). 

\begin{algorithm*}
  \caption{Robust multi-source conformal prediction (complete version of Algorithm \ref{algo:robust})}\label{fullalgo:robust}
  \begin{algorithmic}[1]
  \STATE {\bfseries Input:} {Training data $\mc D = \{\mc O_i = (\bX_i,R_i,R_iY_i,T_i), i=1,\dots, n\}$, where $T_i\in\{0,1,\dots,K\}$ with the target site indexed by $T=0$ and  source sites by $T=k=1,\dots,K-1$; desired coverage probability $1-\alpha$ for an $\alpha\in(0,0.5)$; estimators of the conditional putative cumulative distribution function $m_k(\theta, \bX)$ for the conformal score $\theta$, ratio of the propensity score $\eta_0(\bd X) = \dfrac{\mathbb P(R=0\mid \bX, T=0)}{\mathbb P(R=1\mid \bX, T=0)}$ for the target site, and the density ratio $\omega_{k,0}(\bX) = \dfrac{\mathbb P(\bX\mid T=0, R=0)}{\mathbb P(\bX\mid T=k, R=1)}$ for sites $k=1,\dots,K-1$, denoted by $\widehat m_k(\widehat\theta, \bX)$, $\widehat\eta_0(\bX)$ and $\widehat\omega_{k,0}(\bX)$ (where $\widehat\theta$ is the estimated conformal score), respectively; a tuning parameter $\lambda$ (in the optimization step); a  testing point $\bX=\bx$ from the target site. }
  
  \STATE {\bfseries Output:} {A valid prediction set $\widehat C_\alpha(\bx)$. }
  
  \STATE Split the training data $\mc D$ randomly into $\mc D_1$ and $\mc D_2$, where $\mc D_j=\{\mc O_i\in\mc D, i\in\mc I_j\}$ for $j=1,2$ and $\mc I_1\cup\mc I_2 = \{1,2,\dots, n\}$. 

  \STATE Fit nuisance functions $\widehat m_k$ and $\widehat\omega_{k,0}$ on $\mc D_1$ and predict them on $\mc D_2$. 
  
  \STATE For the target site $k=0$, find the smallest $\widehat\theta = \widehat r_0$ such that 
  \begin{align*}
      0 & = \frac1{|\mc I_2|}\sum_{i\in\mc I_2}\Bigg[ \underbrace{\frac{I(T_i=0, R_i=0)}{\widehat{\mathbb P}(T_i=0, R_i=0)}\{\widehat m_0(\widehat\theta, \bX_i) - (1-\alpha)\} + \frac{I(T_i=0, R_i=1)}{\widehat{\mathbb P}(T_i=0, R_i=1)}\widehat\eta_0(\bd X_i)\{I(S(\bX_i,Y_i)\leq\widehat\theta)-\widehat m_0(\widehat\theta, \bX_i)\} }_{\varphi_0(\mc O_i; \widehat\theta, \widehat m_0, \widehat\eta_0)} \Bigg]. 
  \end{align*}
  
 \STATE For source sites $k\geq 1$, find the smallest $\widehat\theta = \widehat r_k$ that solves 
 \begin{align*}
     0 & =  \frac1{|\mc I_2|}\sum_{i\in\mc I_2} \Bigg[ \underbrace{ \frac{I(T_i=0, R_i=0)}{\widehat{\mathbb P}(T_i=0, R_i=0)}\{\widehat m_0(\widehat\theta, \bX_i) - (1-\alpha)\}
     +  \frac{I(T_i=k, R_i=1)}{\widehat{\mathbb P}(T_i=k, R_i=1)}\widehat\omega_{k,0}(\bX_i)\{I(S(\bX_i,Y_i)\leq\widehat\theta)-\widehat m_k(\widehat\theta, \bX_i)\} }_{\varphi_k(\mc O_i; \widehat\theta, \widehat m_0, \widehat m_k, \widehat\omega_{k,0})} \Bigg]. 
 \end{align*}
 Compute $\widehat\chi_k = |\widehat r_0-\widehat r_k|$. 
  
\STATE Solve for weights $\widehat{\bd w}=(\widehat w_0, \widehat w_1,\dots\widehat w_{K-1})$ that minimize
  \begin{align*}
      Q(\bd w) & = \frac{1}{|\mc I_2|}\sum_{i\in\mc I_2} \Bigg[\sum_{k=1}^{K-1} w_k\{\varphi_0(\mc O_i; \widehat r_0, \widehat m_0, \widehat\eta_0) - \varphi_k(\mc O_i; \widehat r_0, \widehat m_0, \widehat m_k, \widehat\omega_{k,0})\} - \varphi_0(\mc O_i; \widehat r_0, \widehat m_0, \widehat\eta_0)\Bigg]^2  + \frac{1}{|\mc I_2|}\lambda\sum_{k=1}^{K-1} w_k\widehat\chi_k^2,
  \end{align*} 
  subject to $0\leq w_k\leq 1$  and $\displaystyle\sum_{k=0}^{K-1} w_k = 1$. 

   \STATE Compute $\widehat\theta = \widehat r_{0,\text{fed}} = \displaystyle\sum_{k=0}^{K-1} \widehat w_k\widehat r_k$. 
  
  \STATE \textbf{Return: }The prediction set $\widehat C_\alpha(\bx) = \{y: S(\bx,y)\leq \widehat r_{0,\text{fed}}\}$. 
  \end{algorithmic}
\end{algorithm*}

Below, we also present all relevant details about estimating influence functions.

\begin{algorithm*}
  \caption{Estimation of influence functions} \label{algo:est-if}
  \begin{algorithmic}[1]
  \STATE {\bfseries Input:} {Training data $\mc D = \{\mc O_i = (\bX_i,R_i,R_iY_i,T_i), i=1,\dots, n\}$, where $T_i\in\{0,1,\dots,K\}$ with the target site indexed by $T=0$ and  source sites by $T=k=1,\dots,K-1$. \\
  \STATE Desired coverage probability $1-\alpha$ for an $\alpha\in(0,0.5)$. }
  
  \STATE {\bfseries Output:} {Estimates of the target site influence function $\varphi_0(\mc O_i; \widehat\theta, \widehat m_0, \widehat\eta_0)$ and the source site influence functions $\varphi_k(\mc O_i; \widehat\theta, \widehat m_0, \widehat m_k, \widehat\omega_{k,0})$, $k=1,...,K-1$.}
  
  \STATE Randomly split the training data $\mc D$ into two equal-sized folds $\mc D_1\cup\mc D_2$.

  \STATE On the first split $\mc D_1$, fit models to estimate the following nuisance functions via any arbitrary regression model or density ratio model (nonparametric, semiparametric, or parametric):
  \begin{itemize}
    \item Conditional CDF in the target site $m_0(\theta, \bX)$ across a range of values $\theta$  for observations with observed $Y$ ($R=1$): estimate is $\hat{m}_0$
      \item Conditional CDF $m_k(\theta, \bX)$ in source site $k$, $k=1,...,K-1$, across a range of values $\theta$ for observations with observed $Y$ ($R=1$): estimate is $\hat{m}_k$
      \item Ratio of the missingness propensity score $\eta_0(\bd X) = \dfrac{\mathbb P(R=0\mid \bX, T=0)}{\mathbb P(R=1\mid \bX, T=0)}$ for the target site: estimate is $\hat{\eta}_0$
      \item Density ratio $\omega_{k,0}(\bX) = \dfrac{\mathbb P(\bX\mid T=0, R=0)}{\mathbb P(\bX\mid T=k, R=1)}$ for sites $k=1,\dots,K-1$: estimate is $\hat{\omega}_{k,0}$
  \end{itemize}
  We recommend using SuperLearner with the base learners being random forest, elastic net, and generalized linear model (GLM) for the first three nuisance functions and exponential tilting to estimate the density ratio model.

    \STATE On the second split $\mc D_2$, predict the nuisance functions using the models learned ($\hat{m}_k$, $\hat{\eta}_0$, $\hat{\omega}_{k,0}$) from the first split $\mc D_1$.

    \STATE For the target site $k=0$, estimate the influence function as $$\widehat{\varphi}_0(\mc O_i; \widehat\theta, \widehat m_0, \widehat\eta_0) = \frac{I(T_i=0, R_i=0)}{\widehat{\mathbb P}(T_i=0, R_i=0)}\{\widehat m_0(\widehat\theta, \bX_i) - (1-\alpha)\} + \frac{I(T_i=0, R_i=1)}{\widehat{\mathbb P}(T_i=0, R_i=1)}\widehat\eta_0(\bd X_i)\{I(S(\bX_i,Y_i)\leq\widehat\theta)-\widehat m_0(\widehat\theta, \bX_i)\} $$

    \STATE For each of the source sites $k=1,...,K-1$, estimate the influence functions as 
    \begin{align*}
        \widehat{\varphi}_k(\mc O_i; \widehat\theta, \widehat m_0, \widehat m_k, \widehat\omega_{k,0}) &= \frac{I(T_i=0, R_i=0)}{\widehat{\mathbb P}(T_i=0, R_i=0)}\{\widehat m_0(\widehat\theta, \bX_i) - (1-\alpha)\} \\
     &+  \frac{I(T_i=k, R_i=1)}{\widehat{\mathbb P}(T_i=k, R_i=1)}\widehat\omega_{k,0}(\bX_i)\{I(S(\bX_i,Y_i)\leq\widehat\theta)-\widehat m_k(\widehat\theta, \bX_i)\}
    \end{align*}

  \STATE \textbf{Return:} Estimate of target site influence function $\widehat{\varphi}_0$ and source site influence functions $\widehat{\varphi}_k$, $k=1,...,K-1$. 
  \end{algorithmic}
\end{algorithm*}

\section{Additional Simulation Results}\label{apx:simu}

\subsection{An experiment of sample size vs. interval width}\label{subapx:CIwidth}

We first conducted an experiment to assess the relationship between the sample size of the target site vs. the coverage and width of prediction interval. We use only the set-up of target site's DGP but consider two generations for outcomes: the homogeneous variance with $\varepsilon(x)\sim\mc N(0,1)$ and heterogeneous variance with $\varepsilon(x)\sim\mc N(0,[\log(x)]^2)$ for $\varepsilon(X_i)$ defined in \eqref{eq:outcome_md}. Under both cases, the oracle width of a $90\%$ prediction interval for the outcome is $2\times z_{0.95}\Ex\{\sigma(X_i)\}\approx 3.29$, where $z_{0.95}=1.645$ is the 95th percentile of the standard normal distribution. In addition, note that $\Ex\{\sigma(X_i)\} = \int_0^1\sigma(x)dx = 1$ for both $\sigma(x)=1$ and $\sigma(x)=-\log(x)$ (see also \cite{lei2021conformal}). Figure \ref{fig:box-CI-wid} shows the boxplots of interval widths in 500 Monte Carlo simulations. As we can see, the interval width converges to its oracle faster when the variance is homogeneous, by all 3 conformal scores. We can also note that using ASR as the conformal score has an essential bias to oracle width under heterogeneous width, even if the sample size is large enough, while the other two conformal scores are more robust. 

\begin{figure}[H]
    \centering
    \includegraphics[width=0.96\textwidth]{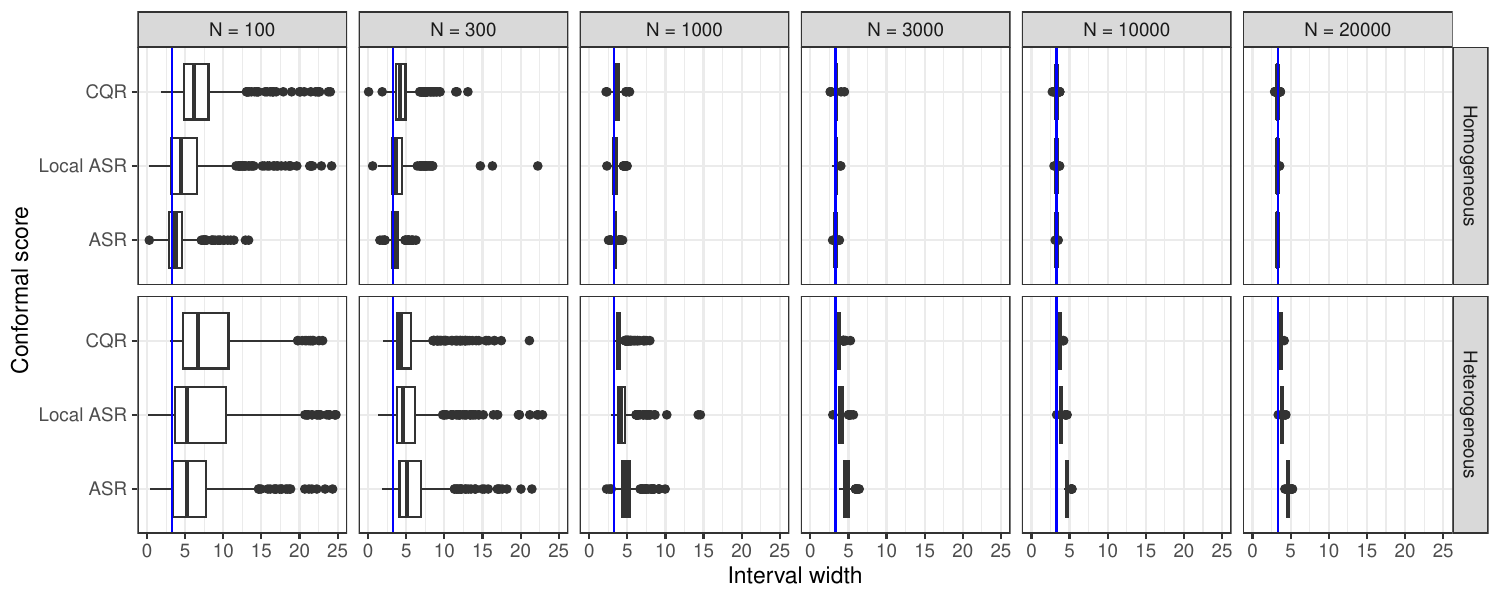}
    \caption{Boxplots of prediction interval widths}
    \label{fig:box-CI-wid}
\end{figure}

\subsection{Complete simulation details and results of coverage probability and interval width, by all sample sizes, variance assumptions, and covariate and outcome distributions across sites}\label{subapx:completeRes}

In this Appendix, we specify more details in our data generating process and competing methods in Section \ref{sec:simu}. In the complete simulation study, we consider 6 methods for constructing prediction interval $\widehat C_\alpha(x)$, where 3 of them (federated (ours), pooled sample and target only) have been described in Section \ref{subsec:DGP}. The additional 3 methods are equal weights, i.e., equally weighting each source site by $1/(K+1)$ (here $=0.2$), and two alternative Federated weights, i.e., Federated I and III shown below. 
\begin{itemize}
    \item Federated I: when solving Step 7 in Algorithm \ref{algo:robust}, set the limit of weight on each source site by $w_k\in[0,1], k=1,\dots,K$, and then the weight of site 0 is $w_0 = 1-\displaystyle\sum_{k=1}^{K-1} w_k$. 
    \item Federated II (ours, and the one in main text): when solving Step 7 in Algorithm \ref{algo:robust}, set the limit of weight on each source site by $w_k\in[0,1], k=1,\dots,K$, and let $w_k^* = w_k\times{K}/{(K+1)} = 0.8w_k$ (here $K=4$), and use $w_k^*$ as the weight of site $k$. Then, $w_0=1-\displaystyle\sum_{k=1}^{K-1} w_k^*$ is the weight of site 0. In this case, $w_0\geq 1/(K+1)=0.2$ in most replications. 
    \item Federated III: setting the limit of weight on each source site by $w_k\in [0,1/(K+1)] = [0,0.2], k=2,\dots, s$, and then $w_0=1-\displaystyle\sum_{k=1}^{K-1}w_k$ for site 0. In this case, $w_0\geq 1/(K+1)$ holds, and thus we always weight the target site the most. 
\end{itemize}

In addition, in this Appendix, we present all simulation findings on coverage probabilities and interval widths via both numerical and visualized results in Tables \ref{tab:homo_n300}--\ref{tab:strg_n3000}, and Figures \ref{fig:box_cov_homo}--\ref{fig:box_len_strg}. We comment, in the following, on patterns and trends we found from these results. 

First, only when CCOD holds does the pooled sample method perform well, where the coverage is close to the nominal level $0.9$ and it is the most efficient one, having the shortest box width (except for CQR under heteroscedastcity). These results make sense as the pooled sample has a larger sample size and when CCOD holds for all sites, all data are directly useful for predicting $Y$ from the target site. However, it can also easily fail when CCOD is violated, either weakly or strongly in our simulation. Compared to other methods, pooled sample can be substantially more sensitive to such violations, which often results in very conservative and wide interval estimations (e.g., from Table \ref{tab:homo_n1000}, the interval widths by pooled sample for ASR under weakly and strongly violations of CCOD are, respectively, $11.19$ and $31.95$ (for homoscedasticity), and $11.27$ and $31.91$ (for heteroscedasticity), which exceeds the oracle width $3.29$ substantially, while other methods always have widths in the range $[3.20,4.10]$. This illustrates that pooling samples from all source sites is not a good strategy in general, especially when there are differences among sites. 

Moreover, the equal weights method can also result in large biases when the distributions of covariates across sites are either weakly or strongly heterogeneous. The biases increase with a stronger difference among covariate distributions. In some cases, it is also less efficient than the federated methods; for instance, in Figure \ref{fig:box_cov_strg}, the boxes of the equal weights method are often wider than those of the three federated methods, as reflected in the corresponding interval width plot, Figure \ref{fig:box_len_strg}. Therefore, it can be biased and less efficient under heterogeneous covariate distributions.

Furthermore, among all methods, only federated weights I, II, and III performed well across settings and exhibited consistent patterns in the coverage probabilities and interval width. The boxplots of coverage probabilities by these federated weights are often situated around the nominal coverage level of $0.9$, and the widths of these boxes are often shorter, indicating higher efficiency in their interval predictions.

Finally, among the three federated weights, there are slight differences with respect to different conformal scores. Federated I and II are less efficient for CQR under both heteroscedasticity and heterogeneous covariate distributions (weakly and strongly). In these cases, federated III for CQR is more efficient, although it is also slightly more conservative (though less biased than the equal weights method). Based on our simulation, we recommend Federated III for CQR, as it offers the optimal choice regarding the bias-variance trade-off among all competing methods. For other cases considered in our simulation, all three federated methods perform similarly and result in valid predictions.

To explore settings in which the propensity score for observing the outcome is allowed to vary more, we provide a comparison by allowing the range of the true propensity score to be in (0.4, 0.6) (panel (a)) and in (0.1, 0.9) (panel (b) of Figure \ref{newfig}. We see that when the propensity score is allowed to have a wider range, our method is even more promising (efficient) than when the propensity score is constrained in (0.4, 0.6). Panel (b) shows the Federated (ours) method provides the overall most efficient interval estimations in ASR and local ASR conformal scores, and the efficiency gains are more obvious than those in panel (a). For CQR, while pooling has the most efficient result, the Federated (ours) also performs well, and it is overall the optimal and safest choice among the three methods compared.


\begin{table}[H]
\centering
\scriptsize
\begin{tabular}{llrrrrrrrrrr}
  \toprule
CFS & CCOD & CP & s.d.(CP) & wd & s.d.(wd) & CP & s.d.(CP) & wd & s.d.(wd) \\ 
  \midrule
  & & \multicolumn{8}{c}{\bf Homoscedasticity where $\sigma(x) = 1$}\\
  \addlinespace
  & & \multicolumn{4}{c}{Federated I}&\multicolumn{4}{c}{Pooled sample}\\
  \cmidrule(lr){3-6}\cmidrule(lr){7-10}
 & holds & 0.893 & 0.038 & 3.27 & 0.35 & 0.900 & 0.018 & 3.30 & 0.16 \\ 
  ASR & weakly violated & 0.894 & 0.040 & 3.29 & 0.37 & 1.000 & 0.000 & 11.16 & 0.55 \\ 
   & strongly violated & 0.894 & 0.039 & 3.29 & 0.37 & 1.000 & 0.000 & 31.99 & 0.97 \\ 
   & holds & 0.904 & 0.033 & 3.38 & 0.36 & 0.901 & 0.023 & 3.32 & 0.22 \\ 
  Local ASR & weakly violated & 0.906 & 0.035 & 3.41 & 0.37 & 0.943 & 0.026 & 3.90 & 0.42 \\ 
   & strongly violated & 0.903 & 0.035 & 3.38 & 0.36 & 0.952 & 0.024 & 4.06 & 0.44 \\ 
   & holds & 0.923 & 0.031 & 3.61 & 0.40 & 0.902 & 0.025 & 3.34 & 0.24 \\ 
  CQR & weakly violated & 0.925 & 0.031 & 3.63 & 0.39 & 0.903 & 0.031 & 3.36 & 0.30 \\ 
   & strongly violated & 0.924 & 0.032 & 3.63 & 0.40 & 0.905 & 0.029 & 3.38 & 0.29 \\ 
   & & \multicolumn{4}{c}{Federated II (ours)}&\multicolumn{4}{c}{Target site only}\\
  \cmidrule(lr){3-6}\cmidrule(lr){7-10}
 & holds & 0.896 & 0.032 & 3.29 & 0.30 & 0.901 & 0.046 & 3.39 & 0.46 \\ 
  ASR & weakly violated & 0.898 & 0.032 & 3.31 & 0.31 & 0.900 & 0.045 & 3.38 & 0.47 \\ 
   & strongly violated & 0.896 & 0.034 & 3.30 & 0.32 & 0.894 & 0.052 & 3.32 & 0.48 \\ 
   & holds & 0.907 & 0.032 & 3.42 & 0.38 & 0.909 & 0.055 & 3.55 & 0.72 \\ 
  Local ASR & weakly violated & 0.909 & 0.032 & 3.44 & 0.38 & 0.908 & 0.057 & 3.56 & 0.78 \\ 
   & strongly violated & 0.905 & 0.032 & 3.39 & 0.35 & 0.900 & 0.060 & 3.44 & 0.63 \\ 
   & holds & 0.925 & 0.029 & 3.63 & 0.39 & 0.922 & 0.054 & 3.71 & 0.67 \\ 
  CQR & weakly violated & 0.926 & 0.030 & 3.65 & 0.39 & 0.917 & 0.060 & 3.66 & 0.70 \\ 
   & strongly violated & 0.925 & 0.031 & 3.63 & 0.39 & 0.917 & 0.059 & 3.66 & 0.67 \\ 
   & & \multicolumn{4}{c}{Federated III}&\multicolumn{4}{c}{Equal weights}\\
  \cmidrule(lr){3-6}\cmidrule(lr){7-10}
 & holds & 0.900 & 0.032 & 3.33 & 0.30 & 0.900 & 0.032 & 3.33 & 0.30 \\ 
  ASR & weakly violated & 0.901 & 0.032 & 3.34 & 0.31 & 0.901 & 0.032 & 3.34 & 0.31 \\ 
   & strongly violated & 0.900 & 0.033 & 3.34 & 0.32 & 0.900 & 0.033 & 3.34 & 0.32 \\ 
   & holds & 0.916 & 0.029 & 3.50 & 0.33 & 0.916 & 0.029 & 3.50 & 0.33 \\ 
  Local ASR & weakly violated & 0.915 & 0.029 & 3.50 & 0.33 & 0.915 & 0.029 & 3.50 & 0.33 \\ 
   & strongly violated & 0.914 & 0.030 & 3.49 & 0.32 & 0.914 & 0.030 & 3.49 & 0.32 \\ 
   & holds & 0.933 & 0.025 & 3.72 & 0.34 & 0.933 & 0.025 & 3.72 & 0.34 \\ 
  CQR & weakly violated & 0.933 & 0.026 & 3.72 & 0.34 & 0.933 & 0.026 & 3.72 & 0.34 \\ 
   & strongly violated & 0.933 & 0.026 & 3.73 & 0.35 & 0.933 & 0.026 & 3.73 & 0.35 \\ 
   \midrule
  & & \multicolumn{8}{c}{\bf Heteroscedasticity where $\sigma(x) = -\log(x)$}\\
  \addlinespace
  & & \multicolumn{4}{c}{Federated I}&\multicolumn{4}{c}{Pooled sample}\\
  \cmidrule(lr){3-6}\cmidrule(lr){7-10}
 & holds & 0.904 & 0.034 & 4.22 & 0.94 & 0.900 & 0.015 & 3.94 & 0.35 \\ 
  ASR & weakly violated & 0.903 & 0.033 & 4.20 & 0.98 & 0.991 & 0.002 & 11.23 & 0.45 \\ 
   & strongly violated & 0.907 & 0.033 & 4.31 & 1.06 & 1.000 & 0.000 & 31.85 & 1.00 \\ 
   & holds & 0.914 & 0.046 & 4.08 & 2.17 & 0.903 & 0.031 & 3.55 & 0.46 \\ 
  Local ASR & weakly violated & 0.914 & 0.048 & 4.24 & 2.55 & 0.915 & 0.032 & 3.74 & 0.49 \\ 
   & strongly violated & 0.909 & 0.048 & 4.14 & 3.42 & 0.924 & 0.027 & 3.85 & 0.42 \\ 
   & holds & 0.926 & 0.029 & 3.40 & 0.31 & 0.902 & 0.024 & 3.21 & 0.15 \\ 
  CQR & weakly violated & 0.926 & 0.029 & 3.39 & 0.26 & 0.904 & 0.025 & 3.22 & 0.16 \\ 
   & strongly violated & 0.928 & 0.029 & 3.42 & 0.34 & 0.905 & 0.026 & 3.23 & 0.16 \\
   & & \multicolumn{4}{c}{Federated II (ours)}&\multicolumn{4}{c}{Target site only}\\
  \cmidrule(lr){3-6}\cmidrule(lr){7-10}
 & holds & 0.905 & 0.030 & 4.21 & 0.87 & 0.898 & 0.042 & 4.16 & 1.21 \\ 
  ASR & weakly violated & 0.905 & 0.029 & 4.20 & 0.86 & 0.900 & 0.043 & 4.21 & 1.22 \\ 
   & strongly violated & 0.908 & 0.030 & 4.30 & 0.92 & 0.902 & 0.042 & 4.27 & 1.26 \\ 
   & holds & 0.915 & 0.044 & 4.14 & 2.50 & 0.909 & 0.061 & 4.37 & 4.93 \\ 
  Local ASR & weakly violated & 0.916 & 0.047 & 4.29 & 2.70 & 0.910 & 0.065 & 4.52 & 3.97 \\ 
   & strongly violated & 0.911 & 0.046 & 4.16 & 3.27 & 0.905 & 0.065 & 4.23 & 3.20 \\ 
   & holds & 0.928 & 0.028 & 3.41 & 0.31 & 0.912 & 0.070 & 3.45 & 0.58 \\ 
  CQR & weakly violated & 0.929 & 0.027 & 3.42 & 0.28 & 0.916 & 0.059 & 3.48 & 0.71 \\ 
   & strongly violated & 0.930 & 0.027 & 3.43 & 0.33 & 0.916 & 0.064 & 3.48 & 0.58 \\ 
   & & \multicolumn{4}{c}{Federated III}&\multicolumn{4}{c}{Equal weights}\\
  \cmidrule(lr){3-6}\cmidrule(lr){7-10}
 & holds & 0.893 & 0.066 & 3.39 & 0.66 & 0.900 & 0.028 & 3.32 & 0.28 \\ 
  & holds & 0.910 & 0.031 & 4.35 & 0.90 & 0.910 & 0.031 & 4.35 & 0.90 \\ 
  ASR & weakly violated & 0.908 & 0.030 & 4.31 & 0.95 & 0.908 & 0.030 & 4.31 & 0.95 \\ 
   & strongly violated & 0.912 & 0.031 & 4.43 & 1.01 & 0.912 & 0.031 & 4.43 & 1.01 \\ 
   & holds & 0.923 & 0.041 & 4.41 & 5.69 & 0.923 & 0.041 & 4.41 & 5.69 \\ 
  Local ASR & weakly violated & 0.924 & 0.044 & 4.38 & 2.54 & 0.924 & 0.044 & 4.38 & 2.54 \\ 
   & strongly violated & 0.919 & 0.043 & 4.32 & 3.50 & 0.919 & 0.043 & 4.32 & 3.50 \\ 
   & holds & 0.939 & 0.026 & 3.51 & 0.33 & 0.939 & 0.026 & 3.51 & 0.33 \\ 
  CQR & weakly violated & 0.940 & 0.025 & 3.52 & 0.30 & 0.940 & 0.025 & 3.52 & 0.30 \\ 
   & strongly violated & 0.941 & 0.025 & 3.54 & 0.34 & 0.941 & 0.025 & 3.54 & 0.34 \\ 
   \bottomrule
\end{tabular}
\begin{tablenotes}
       \item CFS: conformal score; CCOD: common conditional outcome distribution
       \item CP: coverage probability; wd: width; s.d.: standard deviation (over 500 replications)
   \end{tablenotes}
\caption{$n_k=300$, homogeneous covariate distribution}\label{tab:homo_n300}
\end{table}

\begin{table}[H]
\centering
\scriptsize 
\begin{tabular}{llrrrrrrrrrr}
  \toprule
CFS & CCOD & CP & s.d.(CP) & wd & s.d.(wd) & CP & s.d.(CP) & wd & s.d.(wd) \\ 
  \midrule
  & & \multicolumn{8}{c}{\bf Homoscedasticity where $\sigma(x) = 1$}\\
  \addlinespace
  & & \multicolumn{4}{c}{Federated I}&\multicolumn{4}{c}{Pooled sample}\\
  \cmidrule(lr){3-6}\cmidrule(lr){7-10}
 & holds & 0.899 & 0.021 & 3.30 & 0.20 & 0.900 & 0.010 & 3.30 & 0.09 \\ 
  ASR & weakly violated & 0.899 & 0.022 & 3.29 & 0.20 & 1.000 & 0.000 & 11.19 & 0.28 \\ 
   & strongly violated & 0.898 & 0.021 & 3.29 & 0.20 & 1.000 & 0.000 & 31.95 & 0.50 \\ 
   & holds & 0.901 & 0.019 & 3.32 & 0.19 & 0.900 & 0.014 & 3.30 & 0.13 \\ 
  Local ASR & weakly violated & 0.901 & 0.020 & 3.31 & 0.19 & 0.947 & 0.014 & 3.89 & 0.22 \\ 
   & strongly violated & 0.902 & 0.020 & 3.33 & 0.18 & 0.952 & 0.012 & 3.99 & 0.22 \\ 
   & holds & 0.905 & 0.018 & 3.37 & 0.17 & 0.900 & 0.015 & 3.31 & 0.13 \\ 
  CQR & weakly violated & 0.906 & 0.019 & 3.36 & 0.18 & 0.901 & 0.018 & 3.31 & 0.17 \\ 
   & strongly violated & 0.905 & 0.018 & 3.36 & 0.17 & 0.900 & 0.018 & 3.31 & 0.16 \\ 
   & & \multicolumn{4}{c}{Federated II (ours)}&\multicolumn{4}{c}{Target site only}\\
  \cmidrule(lr){3-6}\cmidrule(lr){7-10}
 & holds & 0.900 & 0.017 & 3.31 & 0.16 & 0.901 & 0.023 & 3.32 & 0.22 \\ 
  ASR & weakly violated & 0.900 & 0.018 & 3.30 & 0.17 & 0.900 & 0.023 & 3.31 & 0.22 \\ 
   & strongly violated & 0.899 & 0.018 & 3.30 & 0.16 & 0.901 & 0.023 & 3.32 & 0.22 \\ 
   & holds & 0.901 & 0.018 & 3.32 & 0.17 & 0.901 & 0.030 & 3.34 & 0.30 \\ 
  Local ASR & weakly violated & 0.902 & 0.018 & 3.32 & 0.17 & 0.901 & 0.031 & 3.34 & 0.30 \\ 
   & strongly violated & 0.903 & 0.019 & 3.34 & 0.17 & 0.904 & 0.030 & 3.37 & 0.31 \\ 
   & holds & 0.906 & 0.018 & 3.37 & 0.17 & 0.905 & 0.033 & 3.39 & 0.33 \\ 
  CQR & weakly violated & 0.906 & 0.018 & 3.37 & 0.17 & 0.905 & 0.032 & 3.40 & 0.32 \\ 
   & strongly violated & 0.905 & 0.018 & 3.37 & 0.17 & 0.905 & 0.033 & 3.40 & 0.33 \\ 
   & & \multicolumn{4}{c}{Federated III}&\multicolumn{4}{c}{Equal weights}\\
  \cmidrule(lr){3-6}\cmidrule(lr){7-10}
 & holds & 0.902 & 0.018 & 3.32 & 0.18 & 0.902 & 0.018 & 3.32 & 0.18 \\ 
  ASR & weakly violated & 0.901 & 0.019 & 3.31 & 0.18 & 0.901 & 0.019 & 3.31 & 0.18 \\ 
   & strongly violated & 0.900 & 0.019 & 3.31 & 0.18 & 0.900 & 0.019 & 3.31 & 0.18 \\ 
   & holds & 0.905 & 0.017 & 3.35 & 0.18 & 0.905 & 0.017 & 3.35 & 0.18 \\ 
  Local ASR & weakly violated & 0.905 & 0.018 & 3.35 & 0.18 & 0.905 & 0.018 & 3.35 & 0.18 \\ 
   & strongly violated & 0.906 & 0.018 & 3.37 & 0.18 & 0.906 & 0.018 & 3.37 & 0.18 \\ 
   & holds & 0.910 & 0.017 & 3.41 & 0.16 & 0.910 & 0.017 & 3.41 & 0.16 \\ 
  CQR & weakly violated & 0.910 & 0.017 & 3.41 & 0.17 & 0.910 & 0.017 & 3.41 & 0.17 \\ 
   & strongly violated & 0.909 & 0.017 & 3.41 & 0.16 & 0.909 & 0.017 & 3.41 & 0.16 \\ 
   \midrule
  & & \multicolumn{8}{c}{\bf Heteroscedasticity where $\sigma(x) = -\log(x)$}\\
  \addlinespace
  & & \multicolumn{4}{c}{Federated I}&\multicolumn{4}{c}{Pooled sample}\\
  \cmidrule(lr){3-6}\cmidrule(lr){7-10}
 & holds & 0.901 & 0.020 & 3.99 & 0.50 & 0.900 & 0.010 & 3.91 & 0.20 \\ 
  ASR & weakly violated & 0.903 & 0.021 & 4.04 & 0.54 & 0.991 & 0.002 & 11.27 & 0.24 \\ 
   & strongly violated & 0.901 & 0.020 & 3.99 & 0.50 & 1.000 & 0.000 & 31.91 & 0.57 \\ 
   & holds & 0.907 & 0.031 & 3.67 & 0.71 & 0.899 & 0.020 & 3.47 & 0.23 \\ 
  Local ASR & weakly violated & 0.907 & 0.032 & 3.66 & 0.82 & 0.921 & 0.019 & 3.76 & 0.28 \\ 
   & strongly violated & 0.906 & 0.032 & 3.67 & 0.96 & 0.924 & 0.017 & 3.81 & 0.24 \\ 
   & holds & 0.907 & 0.017 & 3.23 & 0.12 & 0.900 & 0.015 & 3.19 & 0.11 \\ 
  CQR & weakly violated & 0.908 & 0.018 & 3.23 & 0.14 & 0.901 & 0.017 & 3.19 & 0.13 \\ 
   & strongly violated & 0.908 & 0.017 & 3.24 & 0.13 & 0.901 & 0.016 & 3.20 & 0.12 \\ 
   & & \multicolumn{4}{c}{Federated II (ours)}&\multicolumn{4}{c}{Target site only}\\
  \cmidrule(lr){3-6}\cmidrule(lr){7-10}
 & holds & 0.901 & 0.018 & 3.98 & 0.43 & 0.899 & 0.020 & 3.94 & 0.48 \\ 
  ASR & weakly violated & 0.903 & 0.019 & 4.03 & 0.48 & 0.900 & 0.022 & 3.98 & 0.52 \\ 
   & strongly violated & 0.901 & 0.018 & 3.98 & 0.44 & 0.900 & 0.021 & 3.97 & 0.51 \\ 
   & holds & 0.908 & 0.030 & 3.67 & 0.73 & 0.907 & 0.038 & 3.71 & 0.89 \\ 
  Local ASR & weakly violated & 0.908 & 0.031 & 3.67 & 0.84 & 0.907 & 0.040 & 3.71 & 1.06 \\ 
   & strongly violated & 0.907 & 0.031 & 3.68 & 0.93 & 0.908 & 0.039 & 3.73 & 0.94 \\ 
   & holds & 0.907 & 0.017 & 3.23 & 0.12 & 0.903 & 0.031 & 3.23 & 0.18 \\ 
  CQR & weakly violated & 0.908 & 0.018 & 3.23 & 0.13 & 0.907 & 0.034 & 3.26 & 0.19 \\ 
   & strongly violated & 0.909 & 0.017 & 3.24 & 0.13 & 0.907 & 0.032 & 3.26 & 0.18 \\ 
   & & \multicolumn{4}{c}{Federated III}&\multicolumn{4}{c}{Equal weights}\\
  \cmidrule(lr){3-6}\cmidrule(lr){7-10}
  & holds & 0.902 & 0.019 & 4.01 & 0.46 & 0.902 & 0.019 & 4.01 & 0.46 \\ 
  ASR & weakly violated & 0.904 & 0.020 & 4.07 & 0.51 & 0.904 & 0.020 & 4.07 & 0.51 \\ 
   & strongly violated & 0.902 & 0.019 & 4.01 & 0.47 & 0.902 & 0.019 & 4.01 & 0.47 \\ 
   & holds & 0.911 & 0.030 & 3.71 & 0.70 & 0.911 & 0.030 & 3.71 & 0.70 \\ 
  Local ASR & weakly violated & 0.911 & 0.031 & 3.71 & 0.83 & 0.911 & 0.031 & 3.71 & 0.83 \\ 
   & strongly violated & 0.910 & 0.031 & 3.71 & 0.93 & 0.910 & 0.031 & 3.71 & 0.93 \\ 
   & holds & 0.911 & 0.016 & 3.25 & 0.13 & 0.911 & 0.016 & 3.25 & 0.13 \\ 
  CQR & weakly violated & 0.913 & 0.017 & 3.25 & 0.14 & 0.913 & 0.017 & 3.25 & 0.14 \\ 
   & strongly violated & 0.914 & 0.017 & 3.26 & 0.13 & 0.914 & 0.017 & 3.26 & 0.13 \\ 
   \bottomrule
\end{tabular}
\begin{tablenotes}
       \item CFS: conformal score; CCOD: common conditional outcome distribution
       \item CP: coverage probability; wd: width; s.d.: standard deviation (over 500 replications)
   \end{tablenotes}
\caption{$n_k=1000$, homogeneous covariate distribution}\label{tab:homo_n1000}
\end{table}

\begin{table}[H]
\centering
\scriptsize 
\begin{tabular}{llrrrrrrrrrr}
  \toprule
CFS & CCOD & CP & s.d.(CP) & wd & s.d.(wd) & CP & s.d.(CP) & wd & s.d.(wd) \\ 
  \midrule
  & & \multicolumn{8}{c}{\bf Homoscedasticity where $\sigma(x) = 1$}\\
  \addlinespace
  & & \multicolumn{4}{c}{Federated I}&\multicolumn{4}{c}{Pooled sample}\\
  \cmidrule(lr){3-6}\cmidrule(lr){7-10}
 & holds & 0.900 & 0.014 & 3.29 & 0.12 & 0.900 & 0.008 & 3.29 & 0.05 \\ 
  ASR & weakly violated & 0.899 & 0.013 & 3.28 & 0.11 & 1.000 & 0.000 & 11.19 & 0.17 \\ 
   & strongly violated & 0.900 & 0.014 & 3.30 & 0.12 & 1.000 & 0.000 & 31.99 & 0.29 \\ 
   & holds & 0.900 & 0.014 & 3.30 & 0.12 & 0.900 & 0.010 & 3.29 & 0.08 \\ 
  Local ASR & weakly violated & 0.901 & 0.013 & 3.30 & 0.11 & 0.947 & 0.010 & 3.88 & 0.14 \\ 
   & strongly violated & 0.901 & 0.014 & 3.31 & 0.12 & 0.954 & 0.009 & 3.99 & 0.14 \\ 
   & holds & 0.901 & 0.013 & 3.31 & 0.11 & 0.899 & 0.011 & 3.29 & 0.09 \\ 
  CQR & weakly violated & 0.901 & 0.013 & 3.31 & 0.11 & 0.899 & 0.013 & 3.29 & 0.11 \\ 
   & strongly violated & 0.901 & 0.012 & 3.32 & 0.11 & 0.899 & 0.012 & 3.30 & 0.11 \\ 
   & & \multicolumn{4}{c}{Federated II (ours)}&\multicolumn{4}{c}{Target site only}\\
  \cmidrule(lr){3-6}\cmidrule(lr){7-10}
 & holds & 0.900 & 0.012 & 3.29 & 0.10 & 0.899 & 0.015 & 3.28 & 0.13 \\ 
  ASR & weakly violated & 0.899 & 0.012 & 3.29 & 0.10 & 0.900 & 0.014 & 3.30 & 0.12 \\ 
   & strongly violated & 0.901 & 0.012 & 3.30 & 0.10 & 0.901 & 0.014 & 3.31 & 0.13 \\ 
   & holds & 0.900 & 0.013 & 3.29 & 0.11 & 0.899 & 0.019 & 3.29 & 0.17 \\ 
  Local ASR & weakly violated & 0.901 & 0.013 & 3.31 & 0.11 & 0.902 & 0.018 & 3.32 & 0.17 \\ 
   & strongly violated & 0.901 & 0.013 & 3.31 & 0.12 & 0.901 & 0.018 & 3.32 & 0.17 \\ 
   & holds & 0.901 & 0.012 & 3.31 & 0.11 & 0.900 & 0.020 & 3.31 & 0.18 \\ 
  CQR & weakly violated & 0.901 & 0.013 & 3.31 & 0.11 & 0.902 & 0.020 & 3.33 & 0.19 \\ 
   & strongly violated & 0.902 & 0.012 & 3.32 & 0.11 & 0.902 & 0.019 & 3.33 & 0.18 \\  
   & & \multicolumn{4}{c}{Federated III}&\multicolumn{4}{c}{Equal weights}\\
  \cmidrule(lr){3-6}\cmidrule(lr){7-10}
 & holds & 0.901 & 0.013 & 3.30 & 0.11 & 0.901 & 0.013 & 3.30 & 0.11 \\ 
  ASR & weakly violated & 0.900 & 0.012 & 3.29 & 0.10 & 0.900 & 0.012 & 3.29 & 0.10 \\ 
   & strongly violated & 0.901 & 0.013 & 3.31 & 0.11 & 0.901 & 0.013 & 3.31 & 0.11 \\ 
   & holds & 0.901 & 0.013 & 3.31 & 0.11 & 0.901 & 0.013 & 3.31 & 0.11 \\ 
  Local ASR & weakly violated & 0.902 & 0.013 & 3.32 & 0.11 & 0.902 & 0.013 & 3.32 & 0.11 \\ 
   & strongly violated & 0.902 & 0.013 & 3.32 & 0.12 & 0.902 & 0.013 & 3.32 & 0.12 \\ 
   & holds & 0.903 & 0.012 & 3.33 & 0.11 & 0.903 & 0.012 & 3.33 & 0.11 \\ 
  CQR & weakly violated & 0.903 & 0.012 & 3.33 & 0.11 & 0.903 & 0.012 & 3.33 & 0.11 \\ 
   & strongly violated & 0.903 & 0.012 & 3.33 & 0.10 & 0.903 & 0.012 & 3.33 & 0.10 \\ 
   \midrule
  & & \multicolumn{8}{c}{\bf Heteroscedasticity where $\sigma(x) = -\log(x)$}\\
  \addlinespace
  & & \multicolumn{4}{c}{Federated I}&\multicolumn{4}{c}{Pooled sample}\\
  \cmidrule(lr){3-6}\cmidrule(lr){7-10}
 & holds & 0.901 & 0.013 & 3.94 & 0.29 & 0.900 & 0.007 & 3.91 & 0.12 \\ 
  ASR & weakly violated & 0.901 & 0.013 & 3.94 & 0.30 & 0.991 & 0.002 & 11.28 & 0.14 \\ 
   & strongly violated & 0.901 & 0.012 & 3.95 & 0.29 & 1.000 & 0.000 & 31.96 & 0.35 \\ 
   & holds & 0.902 & 0.023 & 3.48 & 0.31 & 0.901 & 0.015 & 3.45 & 0.15 \\ 
  Local ASR & weakly violated & 0.899 & 0.021 & 3.48 & 0.27 & 0.921 & 0.015 & 3.77 & 0.18 \\ 
   & strongly violated & 0.901 & 0.023 & 3.50 & 0.30 & 0.925 & 0.013 & 3.80 & 0.17 \\ 
   & holds & 0.902 & 0.013 & 3.19 & 0.11 & 0.900 & 0.011 & 3.18 & 0.11 \\ 
  CQR & weakly violated & 0.902 & 0.012 & 3.20 & 0.11 & 0.900 & 0.012 & 3.19 & 0.11 \\ 
   & strongly violated & 0.902 & 0.013 & 3.20 & 0.11 & 0.899 & 0.012 & 3.18 & 0.11 \\  
   & & \multicolumn{4}{c}{Federated II (ours)}&\multicolumn{4}{c}{Target site only}\\
  \cmidrule(lr){3-6}\cmidrule(lr){7-10}
& holds & 0.901 & 0.011 & 3.94 & 0.25 & 0.900 & 0.013 & 3.91 & 0.28 \\ 
  ASR & weakly violated & 0.901 & 0.012 & 3.94 & 0.27 & 0.901 & 0.013 & 3.96 & 0.30 \\ 
   & strongly violated & 0.901 & 0.011 & 3.95 & 0.26 & 0.900 & 0.013 & 3.94 & 0.30 \\ 
   & holds & 0.901 & 0.022 & 3.48 & 0.30 & 0.900 & 0.027 & 3.47 & 0.34 \\ 
  Local ASR & weakly violated & 0.899 & 0.021 & 3.48 & 0.26 & 0.900 & 0.025 & 3.50 & 0.32 \\ 
   & strongly violated & 0.901 & 0.023 & 3.50 & 0.30 & 0.901 & 0.027 & 3.51 & 0.35 \\ 
   & holds & 0.902 & 0.013 & 3.19 & 0.11 & 0.900 & 0.022 & 3.19 & 0.13 \\ 
  CQR & weakly violated & 0.902 & 0.012 & 3.20 & 0.11 & 0.901 & 0.020 & 3.20 & 0.13 \\ 
   & strongly violated & 0.902 & 0.012 & 3.20 & 0.11 & 0.901 & 0.020 & 3.20 & 0.13 \\ 
   & & \multicolumn{4}{c}{Federated III}&\multicolumn{4}{c}{Equal weights}\\
  \cmidrule(lr){3-6}\cmidrule(lr){7-10}
  & holds & 0.901 & 0.012 & 3.95 & 0.27 & 0.901 & 0.012 & 3.95 & 0.27 \\ 
  ASR & weakly violated & 0.901 & 0.013 & 3.95 & 0.29 & 0.901 & 0.013 & 3.95 & 0.29 \\ 
   & strongly violated & 0.901 & 0.012 & 3.96 & 0.27 & 0.901 & 0.012 & 3.96 & 0.27 \\ 
   & holds & 0.903 & 0.022 & 3.50 & 0.30 & 0.903 & 0.022 & 3.50 & 0.30 \\ 
  Local ASR & weakly violated & 0.900 & 0.021 & 3.49 & 0.26 & 0.900 & 0.021 & 3.49 & 0.26 \\ 
   & strongly violated & 0.902 & 0.023 & 3.51 & 0.30 & 0.902 & 0.023 & 3.51 & 0.30 \\ 
   & holds & 0.903 & 0.013 & 3.20 & 0.11 & 0.903 & 0.013 & 3.20 & 0.11 \\ 
  CQR & weakly violated & 0.904 & 0.012 & 3.20 & 0.11 & 0.904 & 0.012 & 3.20 & 0.11 \\ 
   & strongly violated & 0.904 & 0.012 & 3.20 & 0.11 & 0.904 & 0.012 & 3.20 & 0.11 \\  
   \bottomrule
\end{tabular}
\begin{tablenotes}
       \item CFS: conformal score; CCOD: common conditional outcome distribution
       \item CP: coverage probability; wd: width; s.d.: standard deviation (over 500 replications)
   \end{tablenotes}
\caption{$n_k=3000$, homogeneous covariate distribution}\label{tab:homo_n3000}
\end{table}

\begin{figure}[H]
    \centering
    \includegraphics[width=0.96\textwidth]{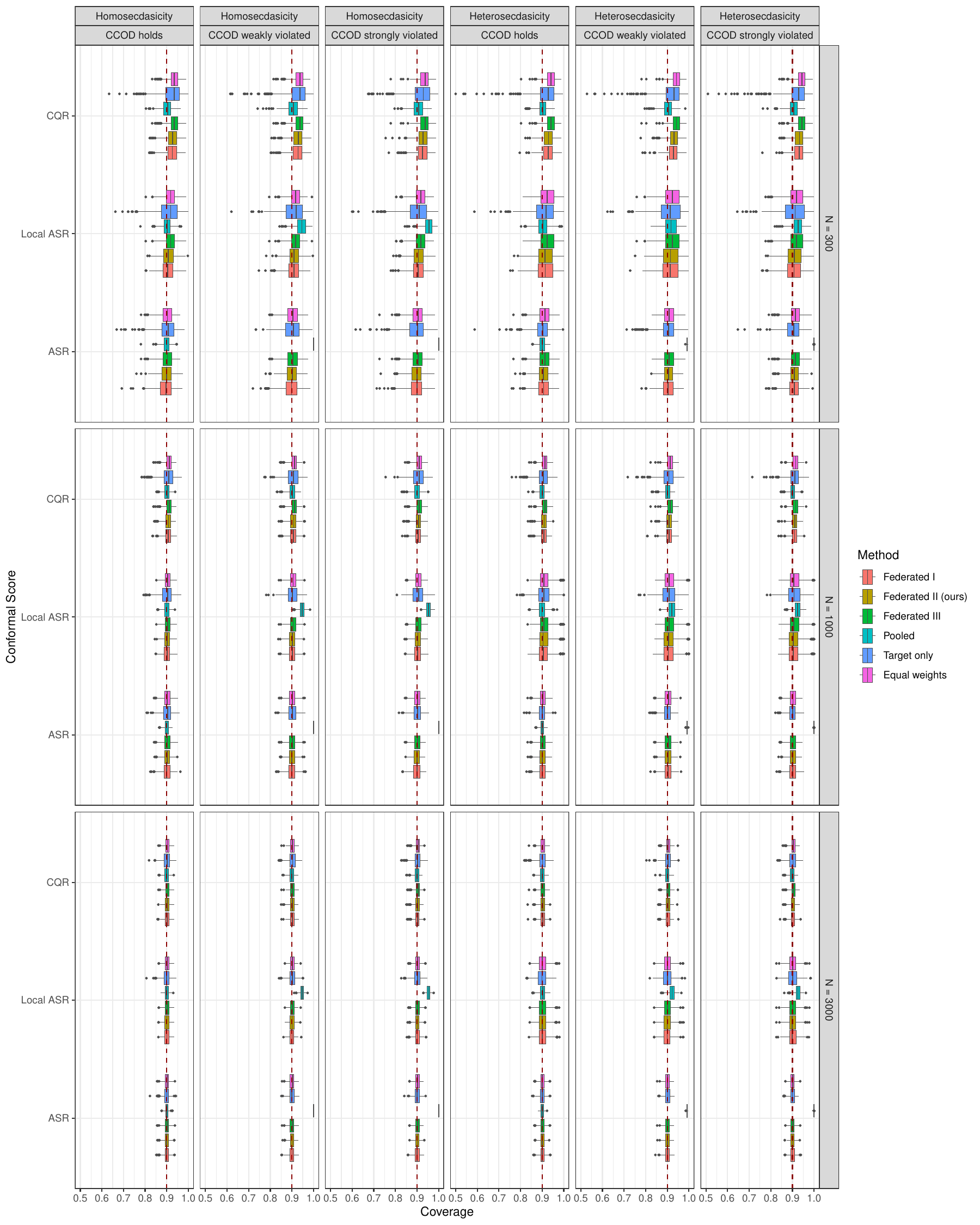}
    \caption{Boxplots of coverage probability, under homogeneous covariate distributions}
    \label{fig:box_cov_homo}
\end{figure}

\begin{figure}[H]
    \centering
    \includegraphics[width=0.96\textwidth]{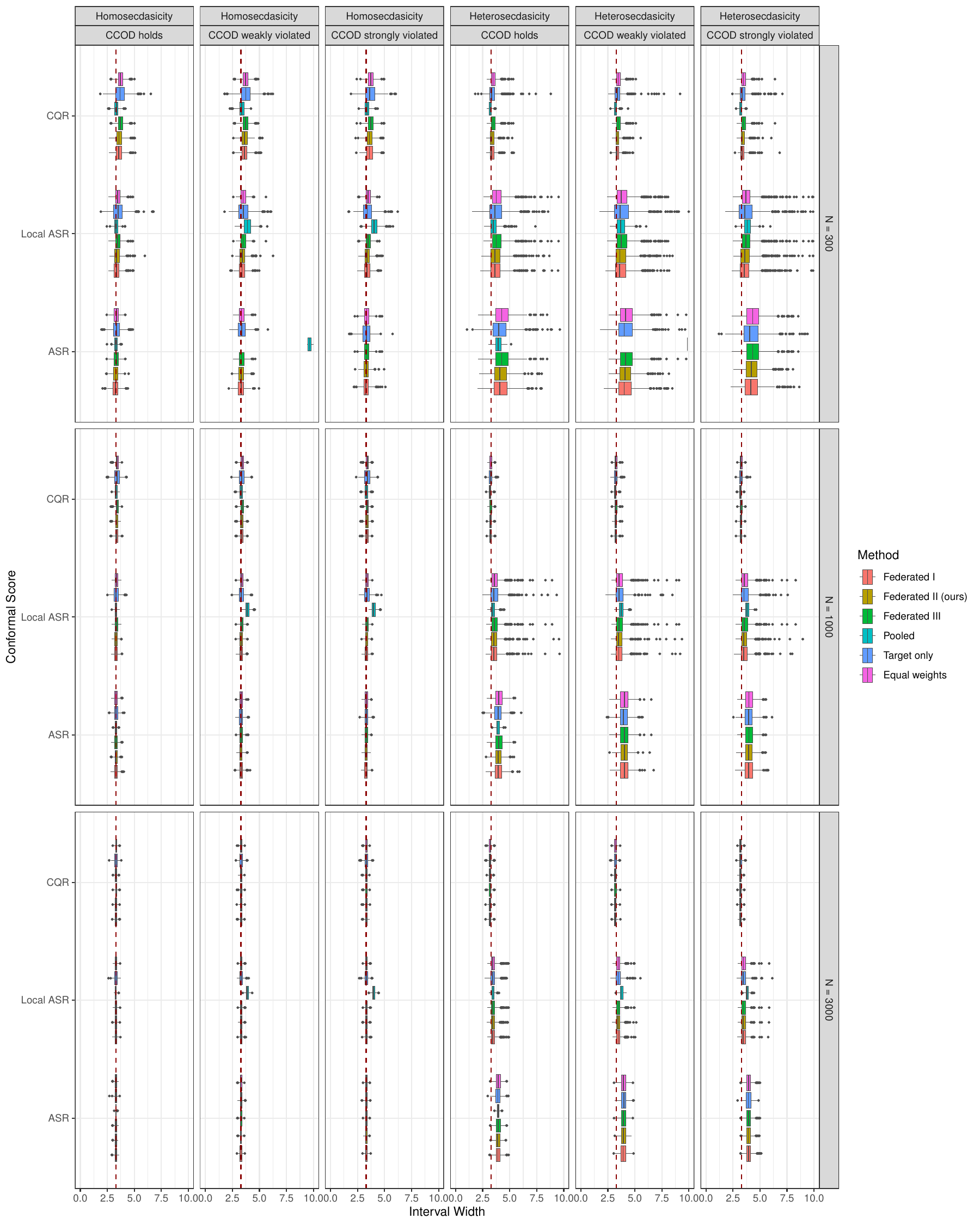}
    \caption{Boxplots of prediction interval width, under homogeneous covariate distributions}
    \label{fig:box_len_homo}
\end{figure}


\begin{table}[H]
\centering
\scriptsize 
\begin{tabular}{llrrrrrrrrrr}
  \toprule
CFS & CCOD & CP & s.d.(CP) & wd & s.d.(wd) & CP & s.d.(CP) & wd & s.d.(wd) \\ 
  \midrule
  & & \multicolumn{8}{c}{\bf Homoscedasticity where $\sigma(x) = 1$}\\
  \addlinespace
  & & \multicolumn{4}{c}{Federated I}&\multicolumn{4}{c}{Pooled sample}\\
  \cmidrule(lr){3-6}\cmidrule(lr){7-10}
 & holds & 0.892 & 0.043 & 3.28 & 0.44 & 0.899 & 0.018 & 3.29 & 0.16 \\ 
  ASR & weakly violated & 0.892 & 0.042 & 3.27 & 0.38 & 1.000 & 0.000 & 9.39 & 0.56 \\ 
   & strongly violated & 0.890 & 0.042 & 3.26 & 0.40 & 1.000 & 0.000 & 25.09 & 1.74 \\ 
   & holds & 0.898 & 0.038 & 3.33 & 0.37 & 0.899 & 0.024 & 3.31 & 0.23 \\ 
  Local ASR & weakly violated & 0.897 & 0.042 & 3.33 & 0.40 & 0.841 & 0.041 & 2.85 & 0.28 \\ 
   & strongly violated & 0.897 & 0.044 & 3.34 & 0.42 & 0.756 & 0.058 & 2.36 & 0.31 \\ 
   & holds & 0.925 & 0.043 & 6.42 & 17.84 & 0.901 & 0.025 & 3.33 & 0.24 \\ 
  CQR & weakly violated & 0.925 & 0.045 & 4.68 & 7.49 & 0.905 & 0.038 & 3.41 & 0.39 \\ 
   & strongly violated & 0.927 & 0.042 & 5.56 & 26.79 & 0.905 & 0.041 & 3.43 & 0.47 \\  
   & & \multicolumn{4}{c}{Federated II (ours)}&\multicolumn{4}{c}{Target site only}\\
  \cmidrule(lr){3-6}\cmidrule(lr){7-10}
  & holds & 0.895 & 0.037 & 3.30 & 0.38 & 0.901 & 0.045 & 3.38 & 0.43 \\ 
  ASR & weakly violated & 0.896 & 0.036 & 3.29 & 0.34 & 0.902 & 0.046 & 3.40 & 0.47 \\ 
   & strongly violated & 0.894 & 0.035 & 3.29 & 0.35 & 0.901 & 0.051 & 3.41 & 0.50 \\ 
   & holds & 0.902 & 0.035 & 3.37 & 0.36 & 0.907 & 0.055 & 3.50 & 0.61 \\ 
  Local ASR & weakly violated & 0.902 & 0.037 & 3.37 & 0.39 & 0.909 & 0.059 & 3.56 & 0.72 \\ 
   & strongly violated & 0.902 & 0.040 & 3.38 & 0.42 & 0.908 & 0.061 & 3.58 & 0.73 \\ 
   & holds & 0.929 & 0.037 & 5.88 & 14.24 & 0.921 & 0.053 & 3.70 & 0.65 \\ 
  CQR & weakly violated & 0.929 & 0.039 & 4.49 & 5.97 & 0.920 & 0.062 & 3.72 & 0.72 \\ 
   & strongly violated & 0.931 & 0.036 & 5.19 & 21.43 & 0.920 & 0.063 & 3.72 & 0.70 \\ 
   & & \multicolumn{4}{c}{Federated III}&\multicolumn{4}{c}{Equal weights}\\
  \cmidrule(lr){3-6}\cmidrule(lr){7-10}
 & holds & 0.899 & 0.034 & 3.33 & 0.35 & 0.907 & 0.037 & 3.43 & 0.41 \\ 
  ASR & weakly violated & 0.900 & 0.033 & 3.33 & 0.33 & 0.904 & 0.035 & 3.39 & 0.38 \\ 
   & strongly violated & 0.898 & 0.034 & 3.32 & 0.34 & 0.904 & 0.035 & 3.39 & 0.38 \\ 
   & holds & 0.913 & 0.036 & 3.51 & 0.49 & 0.924 & 0.039 & 3.90 & 2.17 \\ 
  Local ASR & weakly violated & 0.915 & 0.035 & 3.53 & 0.50 & 0.923 & 0.038 & 4.05 & 4.30 \\ 
   & strongly violated & 0.913 & 0.036 & 3.50 & 0.43 & 0.923 & 0.038 & 3.74 & 0.93 \\ 
   & holds & 0.949 & 0.033 & 4.98 & 5.52 & 0.962 & 0.035 & 8.94 & 33.91 \\ 
  CQR & weakly violated & 0.947 & 0.032 & 4.38 & 2.66 & 0.960 & 0.035 & 6.23 & 6.98 \\ 
   & strongly violated & 0.948 & 0.032 & 4.61 & 7.20 & 0.961 & 0.037 & 6.22 & 8.50 \\ 
   \midrule
  & & \multicolumn{8}{c}{\bf Heteroscedasticity where $\sigma(x) = -\log(x)$}\\
  \addlinespace
  & & \multicolumn{4}{c}{Federated I}&\multicolumn{4}{c}{Pooled sample}\\
  \cmidrule(lr){3-6}\cmidrule(lr){7-10}
 & holds & 0.904 & 0.034 & 4.21 & 0.90 & 0.906 & 0.017 & 4.12 & 0.42 \\ 
  ASR & weakly violated & 0.903 & 0.036 & 4.21 & 1.06 & 0.985 & 0.004 & 9.69 & 0.64 \\ 
   & strongly violated & 0.905 & 0.034 & 4.25 & 0.96 & 1.000 & 0.000 & 25.21 & 1.76 \\ 
   & holds & 0.915 & 0.060 & 3.87 & 0.87 & 0.926 & 0.033 & 3.91 & 0.51 \\ 
  Local ASR & weakly violated & 0.920 & 0.049 & 4.11 & 3.12 & 0.861 & 0.043 & 3.10 & 0.41 \\ 
   & strongly violated & 0.920 & 0.048 & 3.97 & 0.90 & 0.779 & 0.056 & 2.49 & 0.37 \\ 
   & holds & 0.923 & 0.111 & 8.19 & 27.88 & 0.851 & 0.156 & 3.31 & 0.58 \\ 
  CQR & weakly violated & 0.919 & 0.120 & 8.36 & 57.29 & 0.858 & 0.156 & 3.43 & 0.70 \\ 
   & strongly violated & 0.922 & 0.122 & 7.36 & 40.15 & 0.841 & 0.175 & 3.46 & 0.91 \\ 
   & & \multicolumn{4}{c}{Federated II (ours)}&\multicolumn{4}{c}{Target site only}\\
  \cmidrule(lr){3-6}\cmidrule(lr){7-10}
 & holds & 0.906 & 0.030 & 4.22 & 0.82 & 0.902 & 0.043 & 4.27 & 1.24 \\ 
  ASR & weakly violated & 0.905 & 0.031 & 4.23 & 0.93 & 0.902 & 0.044 & 4.31 & 1.30 \\ 
   & strongly violated & 0.908 & 0.029 & 4.28 & 0.83 & 0.905 & 0.044 & 4.36 & 1.24 \\ 
   & holds & 0.919 & 0.052 & 3.93 & 0.86 & 0.906 & 0.083 & 4.14 & 1.63 \\ 
  Local ASR & weakly violated & 0.924 & 0.045 & 4.14 & 2.55 & 0.909 & 0.081 & 4.29 & 1.93 \\ 
   & strongly violated & 0.925 & 0.045 & 4.03 & 0.89 & 0.910 & 0.094 & 4.29 & 1.62 \\ 
   & holds & 0.931 & 0.101 & 7.36 & 22.26 & 0.872 & 0.192 & 4.05 & 1.46 \\ 
  CQR & weakly violated & 0.933 & 0.099 & 7.50 & 45.81 & 0.863 & 0.209 & 4.12 & 1.57 \\ 
   & strongly violated & 0.934 & 0.102 & 6.71 & 32.09 & 0.863 & 0.210 & 4.12 & 1.58 \\ 
   & & \multicolumn{4}{c}{Federated III}&\multicolumn{4}{c}{Equal weights}\\
  \cmidrule(lr){3-6}\cmidrule(lr){7-10}
 & holds & 0.912 & 0.029 & 4.39 & 0.85 & 0.916 & 0.028 & 4.55 & 0.91 \\ 
  ASR & weakly violated & 0.911 & 0.029 & 4.38 & 0.91 & 0.915 & 0.028 & 4.51 & 0.91 \\ 
   & strongly violated & 0.914 & 0.029 & 4.46 & 0.90 & 0.917 & 0.029 & 4.59 & 0.97 \\ 
   & holds & 0.938 & 0.041 & 4.28 & 0.95 & 0.946 & 0.038 & 4.68 & 1.91 \\ 
  Local ASR & weakly violated & 0.940 & 0.042 & 4.76 & 6.03 & 0.948 & 0.038 & 5.05 & 6.76 \\ 
   & strongly violated & 0.938 & 0.043 & 4.34 & 1.04 & 0.946 & 0.039 & 4.63 & 1.48 \\ 
   & holds & 0.970 & 0.054 & 6.53 & 11.45 & 0.977 & 0.047 & 10.45 & 22.37 \\ 
  CQR & weakly violated & 0.966 & 0.064 & 6.51 & 20.21 & 0.974 & 0.052 & 10.32 & 28.37 \\ 
   & strongly violated & 0.969 & 0.059 & 6.02 & 14.23 & 0.977 & 0.050 & 9.86 & 19.04 \\ 
   \bottomrule
\end{tabular}
\begin{tablenotes}
       \item CFS: conformal score; CCOD: common conditional outcome distribution
       \item CP: coverage probability; wd: width; s.d.: standard deviation (over 500 replications)
   \end{tablenotes}
\caption{$n_k=300$, weakly heterogeneous covariate distribution}\label{tab:weak_n300}
\end{table}

\begin{table}[H]
\centering
\scriptsize 
\begin{tabular}{llrrrrrrrrrr}
  \toprule
CFS & CCOD & CP & s.d.(CP) & wd & s.d.(wd) & CP & s.d.(CP) & wd & s.d.(wd) \\ 
  \midrule
  & & \multicolumn{8}{c}{\bf Homoscedasticity where $\sigma(x) = 1$}\\
  \addlinespace
  & & \multicolumn{4}{c}{Federated I}&\multicolumn{4}{c}{Pooled sample}\\
  \cmidrule(lr){3-6}\cmidrule(lr){7-10}
& holds & 0.897 & 0.022 & 3.28 & 0.21 & 0.899 & 0.011 & 3.29 & 0.09 \\ 
  ASR & weakly violated & 0.896 & 0.025 & 3.27 & 0.22 & 1.000 & 0.000 & 9.23 & 0.26 \\ 
   & strongly violated & 0.894 & 0.027 & 3.26 & 0.24 & 1.000 & 0.000 & 24.51 & 0.96 \\ 
   & holds & 0.899 & 0.024 & 3.30 & 0.22 & 0.899 & 0.015 & 3.29 & 0.12 \\ 
  Local ASR & weakly violated & 0.897 & 0.027 & 3.28 & 0.23 & 0.832 & 0.022 & 2.76 & 0.14 \\ 
   & strongly violated & 0.897 & 0.024 & 3.29 & 0.22 & 0.736 & 0.030 & 2.25 & 0.14 \\ 
   & holds & 0.909 & 0.035 & 3.62 & 3.28 & 0.899 & 0.015 & 3.30 & 0.13 \\ 
  CQR & weakly violated & 0.911 & 0.037 & 3.56 & 0.91 & 0.903 & 0.020 & 3.34 & 0.20 \\ 
   & strongly violated & 0.907 & 0.036 & 3.53 & 1.28 & 0.900 & 0.023 & 3.32 & 0.22 \\   
   & & \multicolumn{4}{c}{Federated II (ours)}&\multicolumn{4}{c}{Target site only}\\
  \cmidrule(lr){3-6}\cmidrule(lr){7-10}
 & holds & 0.899 & 0.019 & 3.29 & 0.18 & 0.901 & 0.025 & 3.33 & 0.23 \\ 
  ASR & weakly violated & 0.898 & 0.021 & 3.28 & 0.19 & 0.900 & 0.025 & 3.31 & 0.23 \\ 
   & strongly violated & 0.896 & 0.022 & 3.27 & 0.20 & 0.899 & 0.025 & 3.31 & 0.24 \\ 
   & holds & 0.900 & 0.022 & 3.32 & 0.21 & 0.903 & 0.032 & 3.37 & 0.31 \\ 
  Local ASR & weakly violated & 0.899 & 0.024 & 3.29 & 0.22 & 0.902 & 0.031 & 3.35 & 0.31 \\ 
   & strongly violated & 0.899 & 0.022 & 3.30 & 0.20 & 0.901 & 0.032 & 3.35 & 0.32 \\ 
   & holds & 0.910 & 0.030 & 3.58 & 2.62 & 0.905 & 0.036 & 3.40 & 0.35 \\ 
  CQR & weakly violated & 0.912 & 0.032 & 3.53 & 0.72 & 0.904 & 0.034 & 3.39 & 0.34 \\ 
   & strongly violated & 0.908 & 0.030 & 3.50 & 1.01 & 0.904 & 0.035 & 3.39 & 0.35 \\ 
   & & \multicolumn{4}{c}{Federated III}&\multicolumn{4}{c}{Equal weights}\\
  \cmidrule(lr){3-6}\cmidrule(lr){7-10}
 & holds & 0.902 & 0.021 & 3.33 & 0.20 & 0.905 & 0.024 & 3.37 & 0.25 \\ 
  ASR & weakly violated & 0.901 & 0.021 & 3.31 & 0.20 & 0.903 & 0.023 & 3.34 & 0.23 \\ 
   & strongly violated & 0.899 & 0.022 & 3.30 & 0.20 & 0.902 & 0.024 & 3.34 & 0.23 \\ 
   & holds & 0.906 & 0.023 & 3.39 & 0.30 & 0.913 & 0.027 & 3.53 & 0.88 \\ 
  Local ASR & weakly violated & 0.906 & 0.022 & 3.37 & 0.22 & 0.913 & 0.027 & 3.48 & 0.54 \\ 
   & strongly violated & 0.905 & 0.021 & 3.36 & 0.21 & 0.912 & 0.024 & 3.45 & 0.38 \\ 
   & holds & 0.923 & 0.027 & 3.65 & 1.09 & 0.936 & 0.032 & 3.94 & 1.28 \\ 
  CQR & weakly violated & 0.925 & 0.027 & 3.64 & 0.52 & 0.936 & 0.032 & 3.92 & 0.94 \\ 
   & strongly violated & 0.922 & 0.026 & 3.61 & 0.47 & 0.935 & 0.033 & 3.94 & 1.02 \\ 
   \midrule
  & & \multicolumn{8}{c}{\bf Heteroscedasticity where $\sigma(x) = -\log(x)$}\\
  \addlinespace
  & & \multicolumn{4}{c}{Federated I}&\multicolumn{4}{c}{Pooled sample}\\
  \cmidrule(lr){3-6}\cmidrule(lr){7-10}
& holds & 0.907 & 0.022 & 4.16 & 0.58 & 0.908 & 0.011 & 4.12 & 0.24 \\ 
  ASR & weakly violated & 0.905 & 0.022 & 4.12 & 0.54 & 0.985 & 0.003 & 9.53 & 0.32 \\ 
   & strongly violated & 0.901 & 0.025 & 4.03 & 0.57 & 1.000 & 0.000 & 24.61 & 0.95 \\ 
   & holds & 0.921 & 0.040 & 3.86 & 0.54 & 0.929 & 0.020 & 3.91 & 0.31 \\ 
  Local ASR & weakly violated & 0.922 & 0.033 & 3.84 & 0.50 & 0.859 & 0.025 & 3.04 & 0.22 \\ 
   & strongly violated & 0.918 & 0.038 & 3.78 & 0.47 & 0.767 & 0.031 & 2.39 & 0.17 \\ 
   & holds & 0.894 & 0.136 & 3.84 & 1.73 & 0.872 & 0.107 & 3.23 & 0.35 \\ 
  CQR & weakly violated & 0.885 & 0.142 & 3.81 & 2.35 & 0.879 & 0.104 & 3.26 & 0.36 \\ 
   & strongly violated & 0.882 & 0.143 & 3.74 & 1.45 & 0.867 & 0.121 & 3.29 & 0.46 \\ 
   & & \multicolumn{4}{c}{Federated II (ours)}&\multicolumn{4}{c}{Target site only}\\
  \cmidrule(lr){3-6}\cmidrule(lr){7-10}
 & holds & 0.908 & 0.018 & 4.17 & 0.50 & 0.907 & 0.023 & 4.17 & 0.62 \\ 
  ASR & weakly violated & 0.907 & 0.018 & 4.13 & 0.47 & 0.908 & 0.022 & 4.18 & 0.57 \\ 
   & strongly violated & 0.903 & 0.021 & 4.06 & 0.50 & 0.906 & 0.022 & 4.16 & 0.60 \\ 
   & holds & 0.923 & 0.034 & 3.88 & 0.50 & 0.922 & 0.047 & 3.95 & 0.74 \\ 
  Local ASR & weakly violated & 0.925 & 0.029 & 3.87 & 0.45 & 0.926 & 0.041 & 3.97 & 0.68 \\ 
   & strongly violated & 0.921 & 0.032 & 3.82 & 0.43 & 0.925 & 0.041 & 3.97 & 0.71 \\ 
   & holds & 0.904 & 0.121 & 3.77 & 1.37 & 0.849 & 0.179 & 3.47 & 0.81 \\ 
  CQR & weakly violated & 0.899 & 0.119 & 3.75 & 1.86 & 0.857 & 0.172 & 3.51 & 0.77 \\ 
   & strongly violated & 0.896 & 0.123 & 3.69 & 1.14 & 0.857 & 0.172 & 3.51 & 0.80 \\ 
   & & \multicolumn{4}{c}{Federated III}&\multicolumn{4}{c}{Equal weights}\\
  \cmidrule(lr){3-6}\cmidrule(lr){7-10}
 & holds & 0.911 & 0.019 & 4.26 & 0.54 & 0.914 & 0.020 & 4.36 & 0.60 \\ 
  ASR & weakly violated & 0.910 & 0.019 & 4.23 & 0.52 & 0.913 & 0.021 & 4.33 & 0.60 \\ 
   & strongly violated & 0.908 & 0.020 & 4.19 & 0.54 & 0.911 & 0.022 & 4.31 & 0.63 \\ 
   & holds & 0.932 & 0.030 & 4.04 & 0.55 & 0.940 & 0.030 & 4.24 & 0.87 \\ 
  Local ASR & weakly violated & 0.934 & 0.028 & 4.03 & 0.54 & 0.940 & 0.029 & 4.28 & 1.76 \\ 
   & strongly violated & 0.932 & 0.028 & 4.01 & 0.50 & 0.939 & 0.030 & 4.30 & 1.88 \\ 
   & holds & 0.941 & 0.086 & 4.03 & 1.14 & 0.957 & 0.074 & 4.67 & 1.77 \\ 
  CQR & weakly violated & 0.944 & 0.071 & 3.97 & 0.97 & 0.956 & 0.069 & 4.67 & 2.34 \\ 
   & strongly violated & 0.941 & 0.088 & 3.96 & 0.87 & 0.958 & 0.067 & 4.80 & 3.17 \\ 
   \bottomrule
\end{tabular}
\begin{tablenotes}
       \item CFS: conformal score; CCOD: common conditional outcome distribution
       \item CP: coverage probability; wd: width; s.d.: standard deviation (over 500 replications)
   \end{tablenotes}
\caption{$n_k=1000$, weakly heterogeneous covariate distribution}\label{tab:weak_n1000}
\end{table}

\begin{table}[H]
\centering
\scriptsize 
\begin{tabular}{llrrrrrrrrrr}
  \toprule
CFS & CCOD & CP & s.d.(CP) & wd & s.d.(wd) & CP & s.d.(CP) & wd & s.d.(wd) \\ 
  \midrule
  & & \multicolumn{8}{c}{\bf Homoscedasticity where $\sigma(x) = 1$}\\
  \addlinespace
  & & \multicolumn{4}{c}{Federated I}&\multicolumn{4}{c}{Pooled sample}\\
  \cmidrule(lr){3-6}\cmidrule(lr){7-10}
& holds & 0.898 & 0.016 & 3.27 & 0.14 & 0.900 & 0.008 & 3.29 & 0.05 \\ 
  ASR & weakly violated & 0.898 & 0.015 & 3.28 & 0.13 & 1.000 & 0.000 & 9.20 & 0.16 \\ 
   & strongly violated & 0.896 & 0.017 & 3.27 & 0.15 & 1.000 & 0.000 & 24.47 & 0.56 \\ 
   & holds & 0.898 & 0.016 & 3.28 & 0.14 & 0.900 & 0.011 & 3.29 & 0.09 \\ 
  Local ASR & weakly violated & 0.898 & 0.016 & 3.28 & 0.14 & 0.829 & 0.016 & 2.74 & 0.08 \\ 
   & strongly violated & 0.897 & 0.017 & 3.28 & 0.14 & 0.733 & 0.020 & 2.23 & 0.08 \\ 
   & holds & 0.899 & 0.019 & 3.30 & 0.17 & 0.899 & 0.012 & 3.29 & 0.10 \\ 
  CQR & weakly violated & 0.899 & 0.019 & 3.30 & 0.17 & 0.900 & 0.015 & 3.30 & 0.13 \\ 
   & strongly violated & 0.899 & 0.020 & 3.30 & 0.19 & 0.900 & 0.015 & 3.31 & 0.13 \\  
   & & \multicolumn{4}{c}{Federated II (ours)}&\multicolumn{4}{c}{Target site only}\\
  \cmidrule(lr){3-6}\cmidrule(lr){7-10}
  & holds & 0.898 & 0.013 & 3.28 & 0.12 & 0.900 & 0.014 & 3.30 & 0.12 \\ 
  ASR & weakly violated & 0.899 & 0.013 & 3.29 & 0.11 & 0.901 & 0.015 & 3.30 & 0.13 \\ 
   & strongly violated & 0.897 & 0.014 & 3.27 & 0.12 & 0.900 & 0.014 & 3.30 & 0.13 \\ 
   & holds & 0.899 & 0.015 & 3.29 & 0.13 & 0.901 & 0.018 & 3.31 & 0.17 \\ 
  Local ASR & weakly violated & 0.899 & 0.015 & 3.28 & 0.13 & 0.900 & 0.020 & 3.31 & 0.18 \\ 
   & strongly violated & 0.898 & 0.015 & 3.29 & 0.13 & 0.901 & 0.018 & 3.32 & 0.17 \\ 
   & holds & 0.900 & 0.017 & 3.30 & 0.15 & 0.901 & 0.020 & 3.31 & 0.19 \\ 
  CQR & weakly violated & 0.900 & 0.016 & 3.31 & 0.15 & 0.902 & 0.021 & 3.34 & 0.20 \\ 
   & strongly violated & 0.900 & 0.017 & 3.31 & 0.16 & 0.902 & 0.021 & 3.34 & 0.20 \\ 
   & & \multicolumn{4}{c}{Federated III}&\multicolumn{4}{c}{Equal weights}\\
  \cmidrule(lr){3-6}\cmidrule(lr){7-10}
 & holds & 0.901 & 0.014 & 3.31 & 0.12 & 0.902 & 0.015 & 3.32 & 0.14 \\ 
  ASR & weakly violated & 0.902 & 0.014 & 3.31 & 0.13 & 0.903 & 0.015 & 3.33 & 0.14 \\ 
   & strongly violated & 0.900 & 0.014 & 3.30 & 0.13 & 0.901 & 0.016 & 3.32 & 0.15 \\ 
   & holds & 0.903 & 0.016 & 3.33 & 0.14 & 0.905 & 0.017 & 3.35 & 0.16 \\ 
  Local ASR & weakly violated & 0.903 & 0.015 & 3.33 & 0.14 & 0.905 & 0.016 & 3.35 & 0.15 \\ 
   & strongly violated & 0.903 & 0.015 & 3.33 & 0.14 & 0.905 & 0.016 & 3.35 & 0.16 \\ 
   & holds & 0.909 & 0.018 & 3.40 & 0.18 & 0.914 & 0.022 & 3.46 & 0.24 \\ 
  CQR & weakly violated & 0.909 & 0.016 & 3.40 & 0.16 & 0.914 & 0.019 & 3.46 & 0.22 \\ 
   & strongly violated & 0.909 & 0.017 & 3.40 & 0.17 & 0.913 & 0.020 & 3.46 & 0.23 \\  
   \midrule
  & & \multicolumn{8}{c}{\bf Heteroscedasticity where $\sigma(x) = -\log(x)$}\\
  \addlinespace
  & & \multicolumn{4}{c}{Federated I}&\multicolumn{4}{c}{Pooled sample}\\
  \cmidrule(lr){3-6}\cmidrule(lr){7-10}
& holds & 0.906 & 0.014 & 4.09 & 0.35 & 0.907 & 0.008 & 4.10 & 0.14 \\ 
  ASR & weakly violated & 0.906 & 0.015 & 4.08 & 0.37 & 0.985 & 0.003 & 9.46 & 0.17 \\ 
   & strongly violated & 0.903 & 0.016 & 4.04 & 0.37 & 1.000 & 0.000 & 24.53 & 0.52 \\ 
   & holds & 0.927 & 0.022 & 3.84 & 0.33 & 0.932 & 0.014 & 3.90 & 0.20 \\ 
  Local ASR & weakly violated & 0.925 & 0.023 & 3.86 & 0.33 & 0.856 & 0.018 & 3.02 & 0.15 \\ 
   & strongly violated & 0.924 & 0.023 & 3.82 & 0.33 & 0.763 & 0.021 & 2.36 & 0.12 \\ 
   & holds & 0.861 & 0.124 & 3.23 & 0.43 & 0.888 & 0.058 & 3.19 & 0.21 \\ 
  CQR & weakly violated & 0.860 & 0.132 & 3.26 & 0.47 & 0.890 & 0.057 & 3.21 & 0.23 \\ 
   & strongly violated & 0.861 & 0.125 & 3.23 & 0.42 & 0.880 & 0.081 & 3.21 & 0.28 \\ 
   & & \multicolumn{4}{c}{Federated II (ours)}&\multicolumn{4}{c}{Target site only}\\
  \cmidrule(lr){3-6}\cmidrule(lr){7-10}
 & holds & 0.907 & 0.012 & 4.09 & 0.29 & 0.907 & 0.013 & 4.10 & 0.33 \\ 
  ASR & weakly violated & 0.906 & 0.013 & 4.09 & 0.31 & 0.908 & 0.013 & 4.13 & 0.33 \\ 
   & strongly violated & 0.904 & 0.013 & 4.05 & 0.30 & 0.906 & 0.013 & 4.11 & 0.34 \\ 
   & holds & 0.928 & 0.020 & 3.85 & 0.30 & 0.928 & 0.026 & 3.90 & 0.42 \\ 
  Local ASR & weakly violated & 0.926 & 0.021 & 3.87 & 0.30 & 0.927 & 0.026 & 3.92 & 0.42 \\ 
   & strongly violated & 0.926 & 0.020 & 3.84 & 0.29 & 0.929 & 0.026 & 3.93 & 0.43 \\ 
   & holds & 0.872 & 0.108 & 3.23 & 0.36 & 0.853 & 0.138 & 3.24 & 0.46 \\ 
  CQR & weakly violated & 0.873 & 0.110 & 3.26 & 0.40 & 0.861 & 0.134 & 3.30 & 0.48 \\ 
   & strongly violated & 0.874 & 0.105 & 3.25 & 0.36 & 0.861 & 0.140 & 3.30 & 0.49 \\ 
   & & \multicolumn{4}{c}{Federated III}&\multicolumn{4}{c}{Equal weights}\\
  \cmidrule(lr){3-6}\cmidrule(lr){7-10}
  & holds & 0.910 & 0.013 & 4.18 & 0.34 & 0.911 & 0.014 & 4.22 & 0.38 \\ 
  ASR & weakly violated & 0.909 & 0.014 & 4.18 & 0.35 & 0.910 & 0.014 & 4.21 & 0.39 \\ 
   & strongly violated & 0.907 & 0.014 & 4.14 & 0.35 & 0.908 & 0.015 & 4.17 & 0.38 \\ 
   & holds & 0.935 & 0.020 & 3.98 & 0.33 & 0.937 & 0.021 & 4.04 & 0.38 \\ 
  Local ASR & weakly violated & 0.932 & 0.020 & 3.98 & 0.32 & 0.934 & 0.022 & 4.03 & 0.37 \\ 
   & strongly violated & 0.932 & 0.021 & 3.97 & 0.33 & 0.934 & 0.022 & 4.02 & 0.39 \\ 
   & holds & 0.919 & 0.080 & 3.49 & 0.42 & 0.927 & 0.082 & 3.63 & 0.56 \\ 
  CQR & weakly violated & 0.921 & 0.080 & 3.51 & 0.44 & 0.929 & 0.080 & 3.64 & 0.54 \\ 
   & strongly violated & 0.918 & 0.083 & 3.48 & 0.39 & 0.926 & 0.086 & 3.60 & 0.48 \\ 
   \bottomrule
\end{tabular}
\begin{tablenotes}
       \item CFS: conformal score; CCOD: common conditional outcome distribution
       \item CP: coverage probability; wd: width; s.d.: standard deviation (over 500 replications)
   \end{tablenotes}
\caption{$n_k=3000$, weakly heterogeneous covariate distribution}
\end{table}\label{tab:weak_n3000}

\begin{figure}[H]
    \centering
    \includegraphics[width=0.96\textwidth]{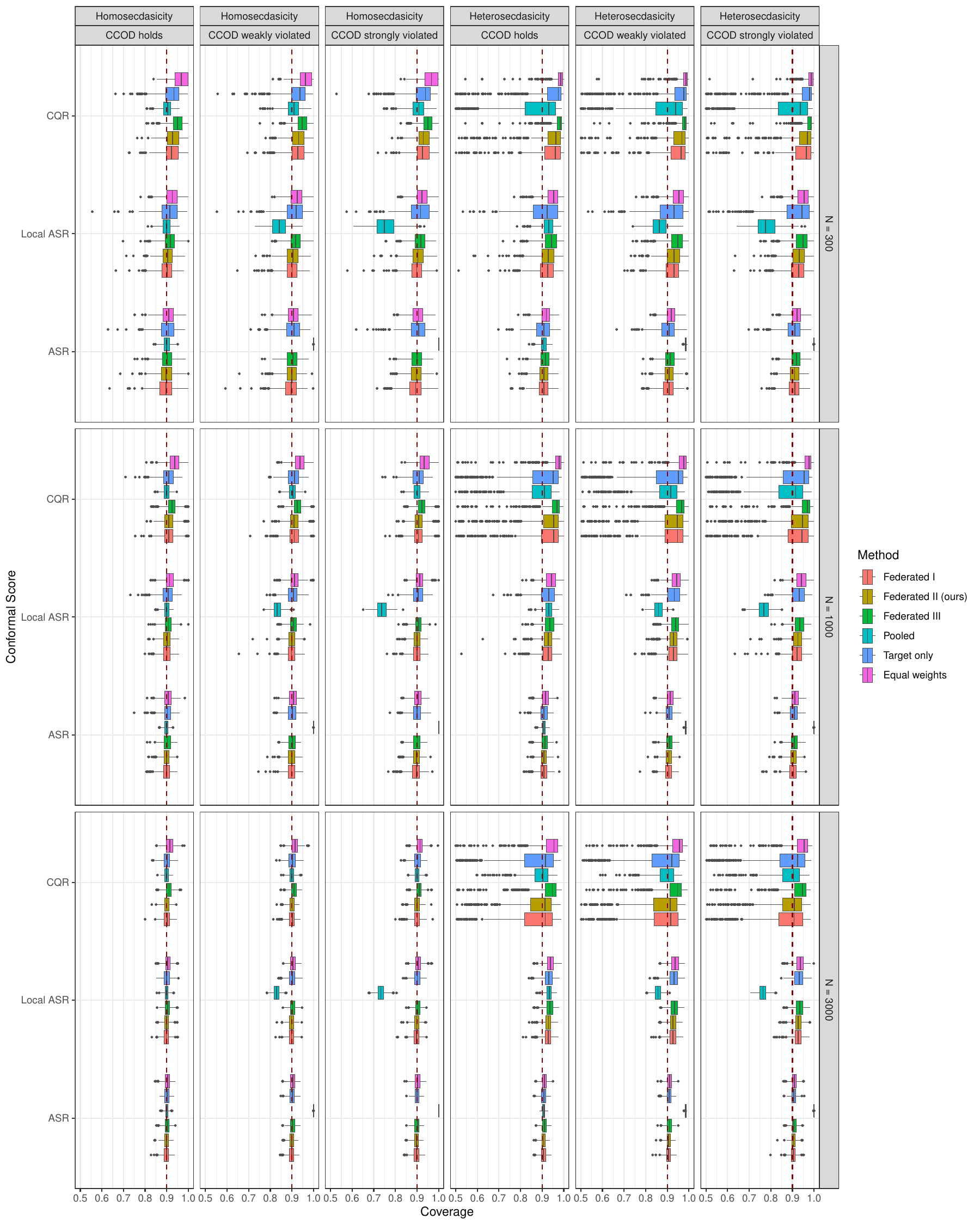}
    \caption{Boxplots of coverage probability, under weakly heterogeneous covariate distributions}
    \label{fig:box_cov_weak}
\end{figure}

\begin{figure}[H]
    \centering
    \includegraphics[width=0.96\textwidth]{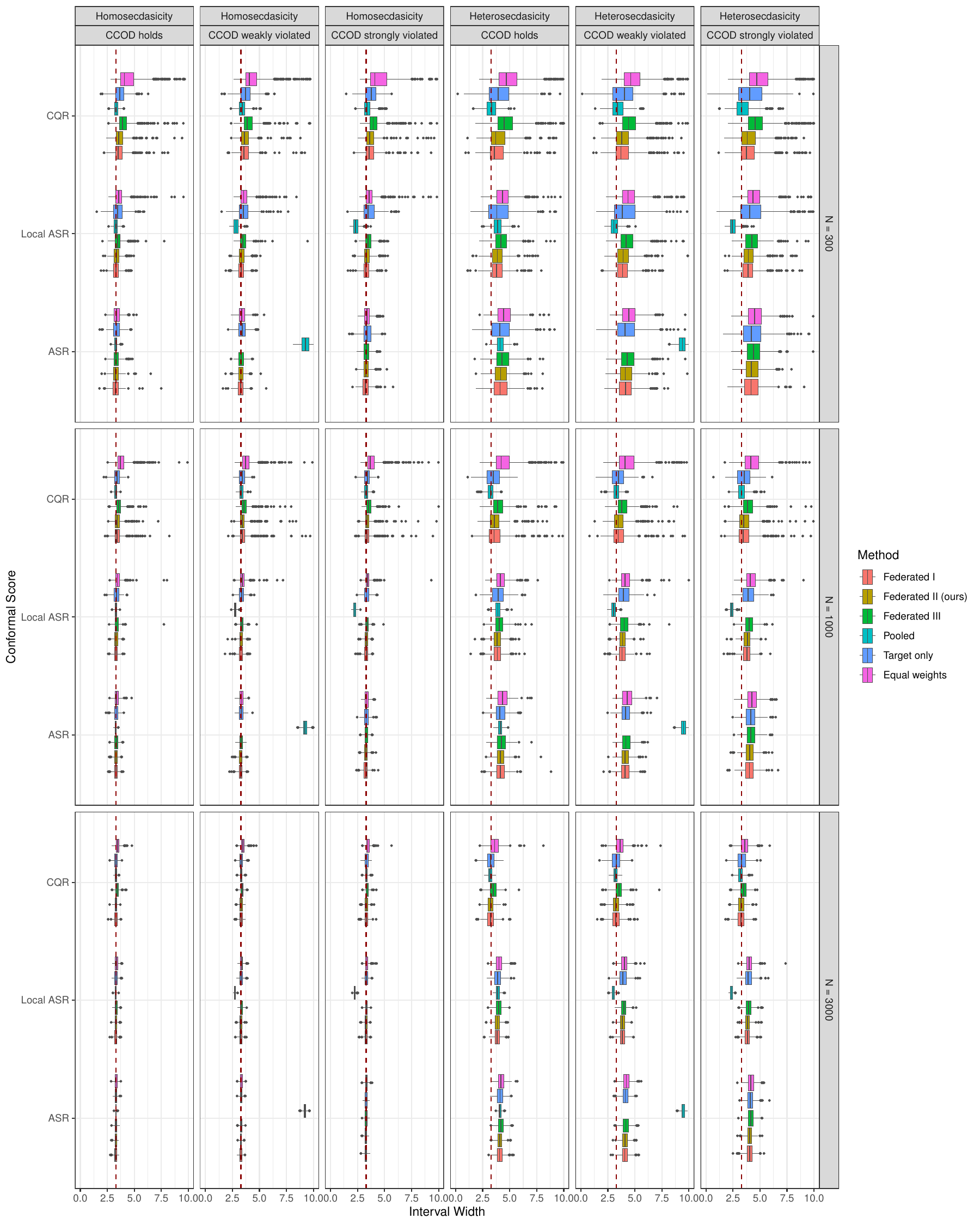}
    \caption{Boxplots of prediction interval width, under weakly heterogeneous covariate distributions}
    \label{fig:box_len_weak}
\end{figure}


\begin{table}[H]
\centering
\scriptsize 
\begin{tabular}{llrrrrrrrrrr}
  \toprule
CFS & CCOD & CP & s.d.(CP) & wd & s.d.(wd) & CP & s.d.(CP) & wd & s.d.(wd) \\ 
  \midrule
  & & \multicolumn{8}{c}{\bf Homoscedasticity where $\sigma(x) = 1$}\\
  \addlinespace
  & & \multicolumn{4}{c}{Federated I}&\multicolumn{4}{c}{Pooled sample}\\
  \cmidrule(lr){3-6}\cmidrule(lr){7-10}
 & holds & 0.891 & 0.043 & 3.26 & 0.37 & 0.901 & 0.018 & 3.31 & 0.16 \\ 
  ASR & weakly violated & 0.893 & 0.039 & 3.28 & 0.40 & 1.000 & 0.000 & 8.85 & 0.50 \\ 
   & strongly violated & 0.893 & 0.041 & 3.29 & 0.39 & 1.000 & 0.000 & 24.23 & 1.49 \\ 
   & holds & 0.899 & 0.043 & 3.36 & 0.44 & 0.901 & 0.023 & 3.33 & 0.22 \\ 
  Local ASR & weakly violated & 0.897 & 0.044 & 3.34 & 0.48 & 0.843 & 0.037 & 2.86 & 0.27 \\ 
   & strongly violated & 0.895 & 0.049 & 3.34 & 0.47 & 0.754 & 0.056 & 2.35 & 0.30 \\ 
   & holds & 0.923 & 0.064 & 4.55 & 6.23 & 0.903 & 0.025 & 3.35 & 0.24 \\ 
  CQR & weakly violated & 0.922 & 0.055 & 4.49 & 5.18 & 0.904 & 0.043 & 3.41 & 0.42 \\ 
   & strongly violated & 0.924 & 0.067 & 4.72 & 7.05 & 0.904 & 0.042 & 3.41 & 0.44 \\ 
   & & \multicolumn{4}{c}{Federated II (ours)}&\multicolumn{4}{c}{Target site only}\\
  \cmidrule(lr){3-6}\cmidrule(lr){7-10}
   & holds & 0.893 & 0.038 & 3.28 & 0.34 & 0.898 & 0.048 & 3.35 & 0.47 \\ 
  ASR & weakly violated & 0.895 & 0.034 & 3.29 & 0.34 & 0.898 & 0.043 & 3.34 & 0.41 \\ 
   & strongly violated & 0.896 & 0.036 & 3.30 & 0.35 & 0.896 & 0.049 & 3.35 & 0.47 \\ 
   & holds & 0.903 & 0.040 & 3.39 & 0.43 & 0.904 & 0.058 & 3.50 & 0.66 \\ 
  Local ASR & weakly violated & 0.900 & 0.041 & 3.37 & 0.46 & 0.904 & 0.056 & 3.48 & 0.64 \\ 
   & strongly violated & 0.900 & 0.043 & 3.37 & 0.44 & 0.905 & 0.058 & 3.50 & 0.63 \\ 
   & holds & 0.927 & 0.051 & 4.38 & 4.96 & 0.918 & 0.058 & 3.68 & 0.66 \\ 
  CQR & weakly violated & 0.926 & 0.045 & 4.33 & 4.12 & 0.920 & 0.062 & 3.71 & 0.71 \\ 
   & strongly violated & 0.928 & 0.055 & 4.52 & 5.62 & 0.920 & 0.060 & 3.71 & 0.69 \\ 
   & & \multicolumn{4}{c}{Federated III}&\multicolumn{4}{c}{Equal weights}\\
  \cmidrule(lr){3-6}\cmidrule(lr){7-10}
 & holds & 0.897 & 0.036 & 3.31 & 0.34 & 0.907 & 0.037 & 3.44 & 0.46 \\ 
  ASR & weakly violated & 0.899 & 0.032 & 3.32 & 0.32 & 0.907 & 0.036 & 3.44 & 0.44 \\ 
   & strongly violated & 0.900 & 0.034 & 3.34 & 0.35 & 0.908 & 0.039 & 3.46 & 0.47 \\ 
   & holds & 0.914 & 0.038 & 3.59 & 1.08 & 0.929 & 0.042 & 4.27 & 3.11 \\ 
  Local ASR & weakly violated & 0.912 & 0.037 & 3.55 & 0.89 & 0.930 & 0.042 & 4.70 & 7.53 \\ 
   & strongly violated & 0.913 & 0.039 & 3.53 & 0.57 & 0.930 & 0.042 & 4.30 & 4.39 \\ 
   & holds & 0.949 & 0.033 & 4.59 & 3.44 & 0.965 & 0.034 & 9.82 & 52.91 \\ 
  CQR & weakly violated & 0.946 & 0.035 & 4.55 & 3.13 & 0.960 & 0.056 & 6.84 & 8.83 \\ 
   & strongly violated & 0.950 & 0.031 & 4.65 & 4.31 & 0.966 & 0.032 & 6.82 & 7.94 \\ 
   \midrule
  & & \multicolumn{8}{c}{\bf Heteroscedasticity where $\sigma(x) = -\log(x)$}\\
  \addlinespace
  & & \multicolumn{4}{c}{Federated I}&\multicolumn{4}{c}{Pooled sample}\\
  \cmidrule(lr){3-6}\cmidrule(lr){7-10}
 & holds & 0.903 & 0.038 & 4.23 & 1.05 & 0.906 & 0.017 & 4.10 & 0.42 \\ 
  ASR & weakly violated & 0.902 & 0.035 & 4.16 & 0.90 & 0.983 & 0.004 & 9.22 & 0.52 \\ 
   & strongly violated & 0.904 & 0.034 & 4.21 & 0.95 & 1.000 & 0.000 & 24.37 & 1.53 \\ 
   & holds & 0.918 & 0.053 & 3.92 & 0.89 & 0.925 & 0.032 & 3.89 & 0.49 \\ 
  Local ASR & weakly violated & 0.919 & 0.053 & 3.97 & 0.92 & 0.870 & 0.044 & 3.19 & 0.45 \\ 
   & strongly violated & 0.919 & 0.052 & 3.96 & 0.91 & 0.789 & 0.058 & 2.55 & 0.40 \\ 
   & holds & 0.915 & 0.134 & 9.13 & 47.58 & 0.850 & 0.155 & 3.29 & 0.56 \\ 
  CQR & weakly violated & 0.916 & 0.127 & 4.57 & 8.47 & 0.854 & 0.169 & 3.43 & 0.72 \\ 
   & strongly violated & 0.915 & 0.126 & 5.86 & 16.05 & 0.839 & 0.183 & 3.48 & 0.91 \\ 
   & & \multicolumn{4}{c}{Federated II (ours)}&\multicolumn{4}{c}{Target site only}\\
  \cmidrule(lr){3-6}\cmidrule(lr){7-10}
  & holds & 0.905 & 0.034 & 4.23 & 0.97 & 0.899 & 0.046 & 4.24 & 1.29 \\ 
  ASR & weakly violated & 0.905 & 0.031 & 4.20 & 0.84 & 0.905 & 0.044 & 4.39 & 1.33 \\ 
   & strongly violated & 0.906 & 0.030 & 4.22 & 0.87 & 0.902 & 0.042 & 4.28 & 1.27 \\ 
   & holds & 0.920 & 0.049 & 3.96 & 0.91 & 0.901 & 0.091 & 4.12 & 1.67 \\ 
  Local ASR & weakly violated & 0.923 & 0.050 & 4.04 & 0.94 & 0.913 & 0.085 & 4.31 & 1.70 \\ 
   & strongly violated & 0.923 & 0.049 & 4.03 & 0.93 & 0.913 & 0.081 & 4.32 & 1.71 \\ 
   & holds & 0.925 & 0.120 & 8.10 & 38.03 & 0.855 & 0.213 & 3.97 & 1.46 \\ 
  CQR & weakly violated & 0.926 & 0.116 & 4.49 & 6.77 & 0.877 & 0.197 & 4.18 & 1.51 \\ 
   & strongly violated & 0.929 & 0.107 & 5.53 & 12.80 & 0.877 & 0.192 & 4.18 & 1.49 \\ 
   & & \multicolumn{4}{c}{Federated III}&\multicolumn{4}{c}{Equal weights}\\
  \cmidrule(lr){3-6}\cmidrule(lr){7-10}
  & holds & 0.910 & 0.031 & 4.38 & 1.00 & 0.918 & 0.031 & 4.69 & 1.25 \\ 
  ASR & weakly violated & 0.909 & 0.030 & 4.34 & 0.86 & 0.915 & 0.030 & 4.52 & 0.96 \\ 
   & strongly violated & 0.910 & 0.030 & 4.35 & 0.86 & 0.917 & 0.029 & 4.61 & 0.99 \\ 
   & holds & 0.936 & 0.044 & 4.27 & 1.02 & 0.951 & 0.038 & 4.99 & 2.81 \\ 
  Local ASR & weakly violated & 0.936 & 0.046 & 4.41 & 1.81 & 0.948 & 0.041 & 5.64 & 9.58 \\ 
   & strongly violated & 0.938 & 0.047 & 4.42 & 1.66 & 0.950 & 0.040 & 5.40 & 7.17 \\ 
   & holds & 0.958 & 0.097 & 7.31 & 18.71 & 0.975 & 0.066 & 16.39 & 86.25 \\ 
  CQR & weakly violated & 0.958 & 0.092 & 5.51 & 7.78 & 0.969 & 0.084 & 13.30 & 53.84 \\ 
   & strongly violated & 0.969 & 0.057 & 6.29 & 9.57 & 0.979 & 0.043 & 10.97 & 24.62 \\ 
   \bottomrule
\end{tabular}
\begin{tablenotes}
       \item CFS: conformal score; CCOD: common conditional outcome distribution
       \item CP: coverage probability; wd: width; s.d.: standard deviation (over 500 replications)
   \end{tablenotes}
\caption{$n_k=300$, strongly heterogeneous covariate distribution}\label{tab:strg_n300}
\end{table}

\begin{table}[H]
\centering
\scriptsize 
\begin{tabular}{llrrrrrrrrrr}
  \toprule
CFS & CCOD & CP & s.d.(CP) & wd & s.d.(wd) & CP & s.d.(CP) & wd & s.d.(wd) \\ 
  \midrule
  & & \multicolumn{8}{c}{\bf Homoscedasticity where $\sigma(x) = 1$}\\
  \addlinespace
  & & \multicolumn{4}{c}{Federated I}&\multicolumn{4}{c}{Pooled sample}\\
  \cmidrule(lr){3-6}\cmidrule(lr){7-10}
 & holds & 0.895 & 0.027 & 3.27 & 0.23 & 0.900 & 0.011 & 3.30 & 0.09 \\ 
  ASR & weakly violated & 0.896 & 0.026 & 3.27 & 0.23 & 1.000 & 0.000 & 8.77 & 0.26 \\ 
   & strongly violated & 0.895 & 0.026 & 3.27 & 0.22 & 1.000 & 0.000 & 23.95 & 0.85 \\ 
   & holds & 0.897 & 0.030 & 3.30 & 0.26 & 0.901 & 0.015 & 3.31 & 0.14 \\ 
  Local ASR & weakly violated & 0.899 & 0.027 & 3.30 & 0.25 & 0.839 & 0.021 & 2.80 & 0.14 \\ 
   & strongly violated & 0.896 & 0.028 & 3.29 & 0.25 & 0.739 & 0.036 & 2.26 & 0.17 \\ 
   & holds & 0.907 & 0.044 & 3.60 & 1.54 & 0.901 & 0.015 & 3.32 & 0.14 \\ 
  CQR & weakly violated & 0.909 & 0.045 & 3.60 & 1.44 & 0.902 & 0.022 & 3.33 & 0.21 \\ 
   & strongly violated & 0.906 & 0.042 & 3.57 & 1.66 & 0.900 & 0.025 & 3.33 & 0.25 \\ 
   & & \multicolumn{4}{c}{Federated II (ours)}&\multicolumn{4}{c}{Target site only}\\
  \cmidrule(lr){3-6}\cmidrule(lr){7-10}
   & holds & 0.897 & 0.022 & 3.28 & 0.20 & 0.902 & 0.024 & 3.34 & 0.24 \\ 
  ASR & weakly violated & 0.898 & 0.022 & 3.28 & 0.20 & 0.901 & 0.025 & 3.32 & 0.24 \\ 
   & strongly violated & 0.896 & 0.022 & 3.28 & 0.20 & 0.899 & 0.025 & 3.31 & 0.24 \\ 
   & holds & 0.900 & 0.026 & 3.32 & 0.23 & 0.906 & 0.029 & 3.40 & 0.31 \\ 
  Local ASR & weakly violated & 0.900 & 0.025 & 3.31 & 0.23 & 0.903 & 0.031 & 3.35 & 0.32 \\ 
   & strongly violated & 0.898 & 0.025 & 3.30 & 0.23 & 0.902 & 0.031 & 3.36 & 0.31 \\ 
   & holds & 0.910 & 0.036 & 3.57 & 1.22 & 0.909 & 0.031 & 3.43 & 0.33 \\ 
  CQR & weakly violated & 0.911 & 0.037 & 3.56 & 1.14 & 0.905 & 0.033 & 3.40 & 0.33 \\ 
   & strongly violated & 0.908 & 0.035 & 3.53 & 1.32 & 0.905 & 0.035 & 3.40 & 0.35 \\ 
   & & \multicolumn{4}{c}{Federated III}&\multicolumn{4}{c}{Equal weights}\\
  \cmidrule(lr){3-6}\cmidrule(lr){7-10}
 & holds & 0.898 & 0.022 & 3.29 & 0.20 & 0.904 & 0.028 & 3.36 & 0.28 \\ 
  ASR & weakly violated & 0.900 & 0.021 & 3.30 & 0.20 & 0.907 & 0.027 & 3.39 & 0.28 \\ 
   & strongly violated & 0.899 & 0.021 & 3.30 & 0.20 & 0.906 & 0.025 & 3.38 & 0.26 \\ 
   & holds & 0.904 & 0.023 & 3.36 & 0.23 & 0.916 & 0.029 & 3.52 & 0.41 \\ 
  Local ASR & weakly violated & 0.907 & 0.023 & 3.38 & 0.26 & 0.919 & 0.029 & 3.59 & 0.62 \\ 
   & strongly violated & 0.904 & 0.023 & 3.36 & 0.25 & 0.917 & 0.031 & 3.71 & 3.48 \\ 
   & holds & 0.922 & 0.028 & 3.63 & 0.61 & 0.942 & 0.034 & 4.03 & 1.02 \\ 
  CQR & weakly violated & 0.925 & 0.026 & 3.65 & 0.68 & 0.944 & 0.031 & 4.03 & 0.90 \\ 
   & strongly violated & 0.921 & 0.028 & 3.61 & 0.75 & 0.943 & 0.033 & 4.11 & 1.30 \\ 
   \midrule
  & & \multicolumn{8}{c}{\bf Heteroscedasticity where $\sigma(x) = -\log(x)$}\\
  \addlinespace
  & & \multicolumn{4}{c}{Federated I}&\multicolumn{4}{c}{Pooled sample}\\
  \cmidrule(lr){3-6}\cmidrule(lr){7-10}
 & holds & 0.904 & 0.025 & 4.10 & 0.60 & 0.907 & 0.010 & 4.10 & 0.23 \\ 
  ASR & weakly violated & 0.903 & 0.022 & 4.07 & 0.56 & 0.983 & 0.003 & 9.11 & 0.28 \\ 
   & strongly violated & 0.903 & 0.025 & 4.07 & 0.58 & 1.000 & 0.000 & 24.14 & 0.85 \\ 
   & holds & 0.918 & 0.040 & 3.80 & 0.54 & 0.929 & 0.020 & 3.90 & 0.31 \\ 
  Local ASR & weakly violated & 0.921 & 0.041 & 3.85 & 0.61 & 0.868 & 0.025 & 3.12 & 0.23 \\ 
   & strongly violated & 0.920 & 0.037 & 3.84 & 0.56 & 0.781 & 0.036 & 2.48 & 0.22 \\ 
   & holds & 0.862 & 0.168 & 3.83 & 2.02 & 0.866 & 0.111 & 3.22 & 0.35 \\ 
  CQR & weakly violated & 0.868 & 0.165 & 3.86 & 2.11 & 0.856 & 0.121 & 3.20 & 0.39 \\ 
   & strongly violated & 0.874 & 0.166 & 4.17 & 4.44 & 0.854 & 0.143 & 3.28 & 0.50 \\ 
   & & \multicolumn{4}{c}{Federated II (ours)}&\multicolumn{4}{c}{Target site only}\\
  \cmidrule(lr){3-6}\cmidrule(lr){7-10}
  & holds & 0.906 & 0.021 & 4.12 & 0.52 & 0.907 & 0.023 & 4.18 & 0.60 \\ 
  ASR & weakly violated & 0.904 & 0.019 & 4.08 & 0.49 & 0.906 & 0.022 & 4.12 & 0.58 \\ 
   & strongly violated & 0.904 & 0.021 & 4.09 & 0.50 & 0.905 & 0.023 & 4.14 & 0.60 \\ 
   & holds & 0.921 & 0.033 & 3.84 & 0.49 & 0.924 & 0.043 & 3.98 & 0.73 \\ 
  Local ASR & weakly violated & 0.924 & 0.035 & 3.87 & 0.56 & 0.924 & 0.043 & 3.96 & 0.72 \\ 
   & strongly violated & 0.923 & 0.032 & 3.87 & 0.52 & 0.923 & 0.047 & 3.97 & 0.76 \\ 
   & holds & 0.879 & 0.143 & 3.76 & 1.59 & 0.854 & 0.177 & 3.48 & 0.81 \\ 
  CQR & weakly violated & 0.884 & 0.141 & 3.77 & 1.67 & 0.849 & 0.183 & 3.46 & 0.81 \\ 
   & strongly violated & 0.889 & 0.143 & 4.03 & 3.53 & 0.849 & 0.177 & 3.46 & 0.82 \\ 
   & & \multicolumn{4}{c}{Federated III}&\multicolumn{4}{c}{Equal weights}\\
  \cmidrule(lr){3-6}\cmidrule(lr){7-10}
    & holds & 0.908 & 0.021 & 4.18 & 0.55 & 0.915 & 0.025 & 4.43 & 0.77 \\ 
  ASR & weakly violated & 0.908 & 0.019 & 4.17 & 0.51 & 0.915 & 0.023 & 4.44 & 0.73 \\ 
   & strongly violated & 0.907 & 0.020 & 4.16 & 0.53 & 0.914 & 0.024 & 4.43 & 0.73 \\ 
   & holds & 0.929 & 0.031 & 3.98 & 0.52 & 0.942 & 0.035 & 4.38 & 0.96 \\ 
  Local ASR & weakly violated & 0.932 & 0.031 & 4.03 & 0.67 & 0.947 & 0.032 & 4.58 & 2.28 \\ 
   & strongly violated & 0.930 & 0.032 & 3.99 & 0.56 & 0.947 & 0.033 & 4.49 & 1.24 \\ 
   & holds & 0.942 & 0.079 & 4.02 & 1.08 & 0.964 & 0.064 & 5.00 & 3.18 \\ 
  CQR & weakly violated & 0.937 & 0.088 & 4.01 & 1.54 & 0.960 & 0.074 & 4.81 & 2.30 \\ 
   & strongly violated & 0.940 & 0.080 & 4.07 & 1.72 & 0.964 & 0.065 & 5.11 & 2.77 \\ 
   \bottomrule
\end{tabular}
\begin{tablenotes}
       \item CFS: conformal score; CCOD: common conditional outcome distribution
       \item CP: coverage probability; wd: width; s.d.: standard deviation (over 500 replications)
   \end{tablenotes}
\caption{$n_k=1000$, strongly heterogeneous covariate distribution}\label{tab:strg_n1000}
\end{table}

\begin{table}[H]
\centering
\scriptsize 
\begin{tabular}{llrrrrrrrrrr}
  \toprule
CFS & CCOD & CP & s.d.(CP) & wd & s.d.(wd) & CP & s.d.(CP) & wd & s.d.(wd) \\ 
  \midrule
  & & \multicolumn{8}{c}{\bf Homoscedasticity where $\sigma(x) = 1$}\\
  \addlinespace
  & & \multicolumn{4}{c}{Federated I}&\multicolumn{4}{c}{Pooled sample}\\
  \cmidrule(lr){3-6}\cmidrule(lr){7-10}
 & holds & 0.898 & 0.017 & 3.28 & 0.15 & 0.899 & 0.008 & 3.28 & 0.06 \\ 
  ASR & weakly violated & 0.897 & 0.016 & 3.27 & 0.14 & 1.000 & 0.000 & 8.74 & 0.15 \\ 
   & strongly violated & 0.897 & 0.017 & 3.27 & 0.15 & 1.000 & 0.000 & 23.78 & 0.49 \\ 
   & holds & 0.898 & 0.018 & 3.28 & 0.15 & 0.899 & 0.011 & 3.29 & 0.09 \\ 
  Local ASR & weakly violated & 0.898 & 0.019 & 3.28 & 0.17 & 0.835 & 0.015 & 2.78 & 0.08 \\ 
   & strongly violated & 0.898 & 0.018 & 3.28 & 0.15 & 0.732 & 0.022 & 2.22 & 0.09 \\ 
   & holds & 0.900 & 0.024 & 3.31 & 0.22 & 0.899 & 0.011 & 3.29 & 0.10 \\ 
  CQR & weakly violated & 0.899 & 0.026 & 3.31 & 0.24 & 0.899 & 0.015 & 3.29 & 0.13 \\ 
   & strongly violated & 0.899 & 0.024 & 3.31 & 0.21 & 0.899 & 0.016 & 3.30 & 0.15 \\ 
   & & \multicolumn{4}{c}{Federated II (ours)}&\multicolumn{4}{c}{Target site only}\\
  \cmidrule(lr){3-6}\cmidrule(lr){7-10}
   & holds & 0.899 & 0.015 & 3.28 & 0.13 & 0.900 & 0.014 & 3.30 & 0.13 \\ 
  ASR & weakly violated & 0.898 & 0.014 & 3.28 & 0.12 & 0.900 & 0.014 & 3.30 & 0.13 \\ 
   & strongly violated & 0.898 & 0.014 & 3.28 & 0.12 & 0.900 & 0.014 & 3.30 & 0.12 \\ 
   & holds & 0.898 & 0.016 & 3.28 & 0.14 & 0.900 & 0.018 & 3.30 & 0.17 \\ 
  Local ASR & weakly violated & 0.899 & 0.017 & 3.29 & 0.15 & 0.901 & 0.019 & 3.31 & 0.18 \\ 
   & strongly violated & 0.899 & 0.016 & 3.29 & 0.13 & 0.900 & 0.018 & 3.31 & 0.17 \\ 
   & holds & 0.901 & 0.020 & 3.32 & 0.19 & 0.901 & 0.020 & 3.32 & 0.18 \\ 
  CQR & weakly violated & 0.900 & 0.022 & 3.31 & 0.20 & 0.903 & 0.021 & 3.34 & 0.19 \\ 
   & strongly violated & 0.900 & 0.020 & 3.31 & 0.18 & 0.903 & 0.020 & 3.34 & 0.18 \\ 
   & & \multicolumn{4}{c}{Federated III}&\multicolumn{4}{c}{Equal weights}\\
  \cmidrule(lr){3-6}\cmidrule(lr){7-10}
 & holds & 0.901 & 0.015 & 3.30 & 0.13 & 0.903 & 0.017 & 3.33 & 0.16 \\ 
  ASR & weakly violated & 0.900 & 0.015 & 3.30 & 0.13 & 0.903 & 0.018 & 3.33 & 0.17 \\ 
   & strongly violated & 0.901 & 0.013 & 3.30 & 0.12 & 0.902 & 0.017 & 3.33 & 0.16 \\ 
   & holds & 0.902 & 0.016 & 3.32 & 0.15 & 0.907 & 0.019 & 3.37 & 0.19 \\ 
  Local ASR & weakly violated & 0.903 & 0.017 & 3.33 & 0.15 & 0.907 & 0.019 & 3.37 & 0.19 \\ 
   & strongly violated & 0.903 & 0.014 & 3.33 & 0.13 & 0.906 & 0.018 & 3.37 & 0.17 \\ 
   & holds & 0.910 & 0.018 & 3.41 & 0.21 & 0.918 & 0.022 & 3.51 & 0.28 \\ 
  CQR & weakly violated & 0.908 & 0.020 & 3.39 & 0.20 & 0.915 & 0.025 & 3.48 & 0.28 \\ 
   & strongly violated & 0.909 & 0.017 & 3.41 & 0.17 & 0.915 & 0.022 & 3.48 & 0.25 \\ 
   \midrule
  & & \multicolumn{8}{c}{\bf Heteroscedasticity where $\sigma(x) = -\log(x)$}\\
  \addlinespace
  & & \multicolumn{4}{c}{Federated I}&\multicolumn{4}{c}{Pooled sample}\\
  \cmidrule(lr){3-6}\cmidrule(lr){7-10}
 & holds & 0.906 & 0.017 & 4.09 & 0.40 & 0.908 & 0.007 & 4.11 & 0.13 \\ 
  ASR & weakly violated & 0.905 & 0.015 & 4.06 & 0.37 & 0.983 & 0.003 & 9.04 & 0.18 \\ 
   & strongly violated & 0.904 & 0.016 & 4.07 & 0.38 & 1.000 & 0.000 & 23.98 & 0.51 \\ 
   & holds & 0.925 & 0.024 & 3.83 & 0.37 & 0.932 & 0.014 & 3.91 & 0.20 \\ 
  Local ASR & weakly violated & 0.924 & 0.024 & 3.85 & 0.35 & 0.866 & 0.019 & 3.11 & 0.15 \\ 
   & strongly violated & 0.924 & 0.028 & 3.85 & 0.37 & 0.773 & 0.023 & 2.42 & 0.14 \\ 
   & holds & 0.855 & 0.141 & 3.28 & 0.53 & 0.886 & 0.060 & 3.19 & 0.22 \\ 
  CQR & weakly violated & 0.861 & 0.142 & 3.30 & 0.53 & 0.877 & 0.078 & 3.19 & 0.25 \\ 
   & strongly violated & 0.863 & 0.134 & 3.29 & 0.49 & 0.868 & 0.098 & 3.19 & 0.32 \\ 
   & & \multicolumn{4}{c}{Federated II (ours)}&\multicolumn{4}{c}{Target site only}\\
  \cmidrule(lr){3-6}\cmidrule(lr){7-10}
  & holds & 0.907 & 0.014 & 4.10 & 0.33 & 0.908 & 0.014 & 4.14 & 0.34 \\ 
  ASR & weakly violated & 0.905 & 0.013 & 4.07 & 0.31 & 0.907 & 0.014 & 4.11 & 0.33 \\ 
   & strongly violated & 0.905 & 0.014 & 4.08 & 0.32 & 0.907 & 0.013 & 4.13 & 0.34 \\ 
   & holds & 0.927 & 0.021 & 3.85 & 0.32 & 0.930 & 0.027 & 3.93 & 0.44 \\ 
  Local ASR & weakly violated & 0.925 & 0.021 & 3.86 & 0.32 & 0.927 & 0.026 & 3.92 & 0.41 \\ 
   & strongly violated & 0.926 & 0.023 & 3.87 & 0.33 & 0.930 & 0.025 & 3.95 & 0.44 \\ 
   & holds & 0.866 & 0.123 & 3.27 & 0.44 & 0.850 & 0.144 & 3.26 & 0.50 \\ 
  CQR & weakly violated & 0.870 & 0.124 & 3.29 & 0.45 & 0.863 & 0.127 & 3.30 & 0.46 \\ 
   & strongly violated & 0.876 & 0.113 & 3.29 & 0.42 & 0.863 & 0.132 & 3.30 & 0.50 \\ 
   & & \multicolumn{4}{c}{Federated III}&\multicolumn{4}{c}{Equal weights}\\
  \cmidrule(lr){3-6}\cmidrule(lr){7-10}
   & holds & 0.909 & 0.014 & 4.17 & 0.34 & 0.912 & 0.017 & 4.27 & 0.45 \\ 
  ASR & weakly violated & 0.908 & 0.014 & 4.15 & 0.35 & 0.911 & 0.017 & 4.24 & 0.47 \\ 
   & strongly violated & 0.908 & 0.013 & 4.16 & 0.32 & 0.910 & 0.015 & 4.22 & 0.41 \\ 
   & holds & 0.933 & 0.021 & 3.96 & 0.36 & 0.939 & 0.024 & 4.11 & 0.49 \\ 
  Local ASR & weakly violated & 0.931 & 0.022 & 3.97 & 0.37 & 0.936 & 0.026 & 4.10 & 0.52 \\ 
   & strongly violated & 0.933 & 0.020 & 3.98 & 0.33 & 0.937 & 0.023 & 4.07 & 0.43 \\ 
   & holds & 0.914 & 0.090 & 3.49 & 0.43 & 0.933 & 0.083 & 3.75 & 0.60 \\ 
  CQR & weakly violated & 0.916 & 0.085 & 3.50 & 0.43 & 0.929 & 0.090 & 3.74 & 0.60 \\ 
   & strongly violated & 0.921 & 0.082 & 3.51 & 0.41 & 0.934 & 0.082 & 3.73 & 0.53 \\ 
   \bottomrule
\end{tabular}
\begin{tablenotes}
       \item CFS: conformal score; CCOD: common conditional outcome distribution
       \item CP: coverage probability; wd: width; s.d.: standard deviation (over 500 replications)
   \end{tablenotes}
\caption{$n_k=3000$, strongly heterogeneous covariate distribution}\label{tab:strg_n3000}
\end{table}

\begin{figure}[H]
    \centering
    \includegraphics[width=0.96\textwidth]{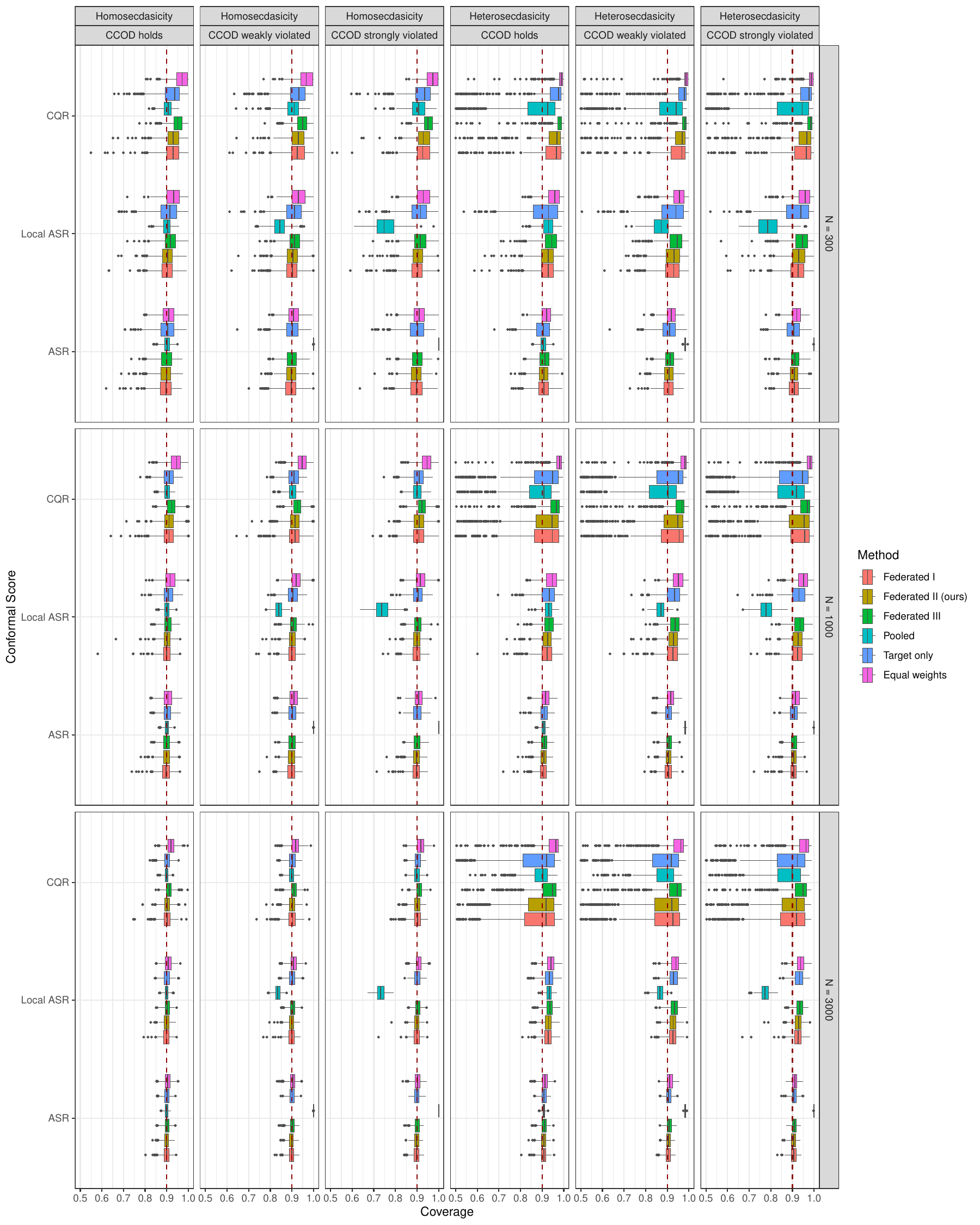}
    \caption{Boxplots of coverage probability, under strongly heterogeneous covariate distributions}
    \label{fig:box_cov_strg}
\end{figure}

\begin{figure}[H]
    \centering
    \includegraphics[width=0.96\textwidth]{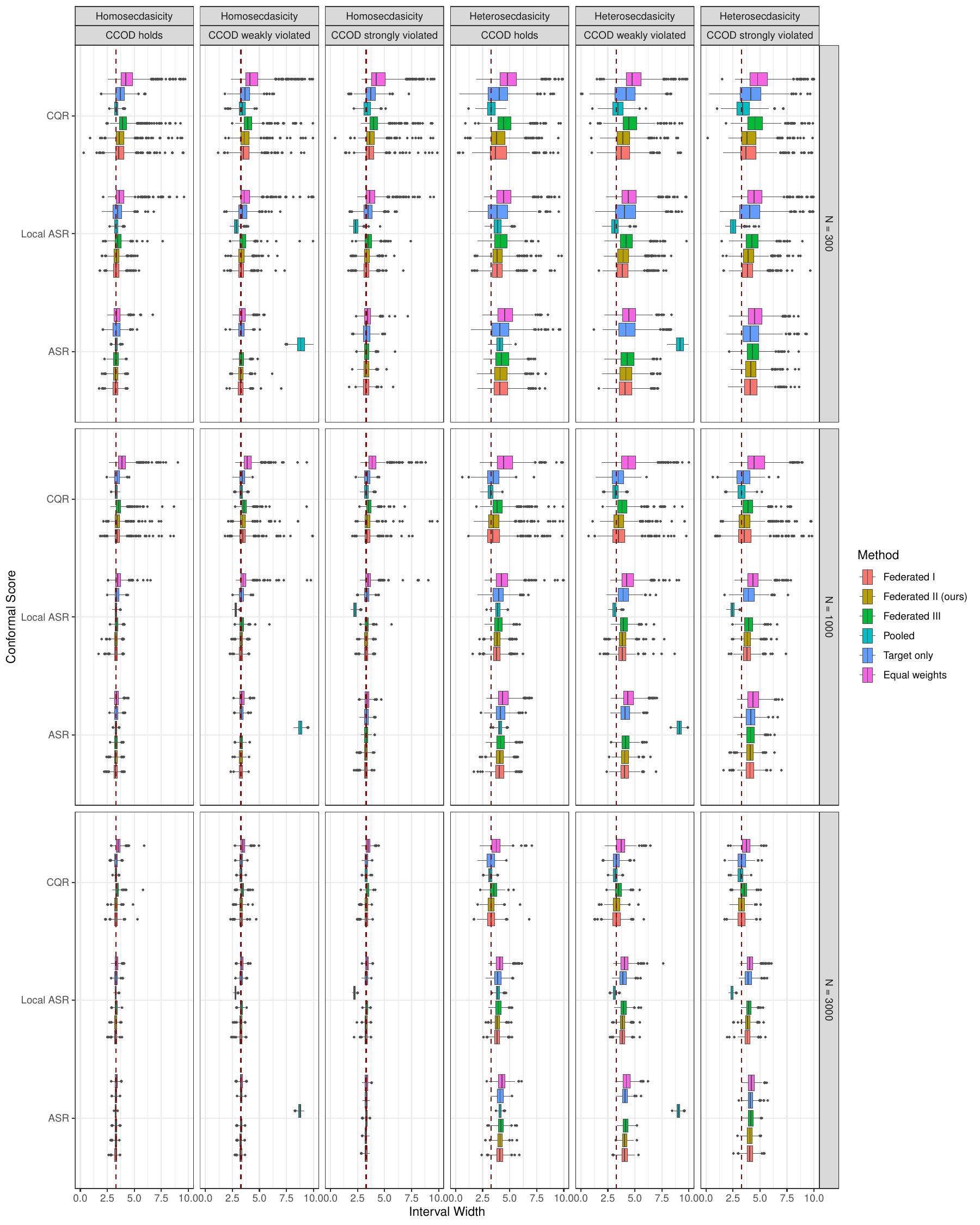}
    \caption{Boxplots of prediction interval width, under strongly heterogeneous covariate distributions}
    \label{fig:box_len_strg}
\end{figure}


\subsection{Local coverage over covariate values and scatterplots of federated weights}\label{subapx:condCov}

In this section, we provide the conditional coverage and federated weights plots. 

Figure \ref{fig:localCov} shows the plots of local coverage of the constructed prediction intervals over a grid of $X\in[0,1]$, where the sample size is set to be $n_k=3000$ for all sites. We used the smoothing method and published R code by \citet{lei2018distribution} for these plots. We can see that under homoscedasticity, the local coverage is constant (a horizontal line) over the covariate values by a given conformal score. Most of these horizontal lines are close to $0.9$, except for the pooled sample. The three federated weights consistently performed well under homoscedasticity. Furthermore, under heteroscedasticity, we can see the local coverage when the value of $X$ is too small always deviates from the nominal level by all methods and conformal scores, which makes sense as $-\log x\to\infty$ when $x\to 0$. When $X$ is sufficiently larger than $0$, the local coverage increases. Among the three conformal scores, ASR is the most sensitive one to the change in variance, and does not have coverage close to $0.9$ almost everywhere. This confirms findings in \citet{lei2018distribution}. The other two conformal scores are more robust against the heteroscedastic variance. When $X \in [0.1,0.6]$, their local coverages are close to $0.9$, except for the pooled sample method.  

In addition, Figure \ref{fig:wts_all} shows three federated weights vs. $\chi^2_k$ values using data of $n_k=3000$ and under heteroscedasticity, where we only plotted weights corresponding to $\chi_k^2\in[0,0.5]$ for illustration purposes, i.e., some weights corresponding to $\chi_k^2>0.5$ are not shown. As can be seen from the upper 9 panels when CCOD holds, in every case, all weights are clustered more or less around $0.2$. When covariate distributions are heterogeneous, the weights distributions become more complex, but generally when $\chi_k^2$ is smaller, there are larger weights in each panel. Also, there are obviously some larger weights ($>0.2$, i.e., above the red dashed lines) in site 1; about half of the weights are below $0.2$ for both sites 2 and 3, and most weights for site 4 are close to $0$. Although site 3 has some surprisingly large weights, it also shows a more unstable pattern of weights, which might be a reflection of its heterogeneity to the target site. Overall, the trend of weights fits the expectation of our method: the bigger the difference to the target site, the smaller (or the less stable) the weights.

\begin{figure}[H]
    \centering
    \includegraphics[width=1\textwidth]{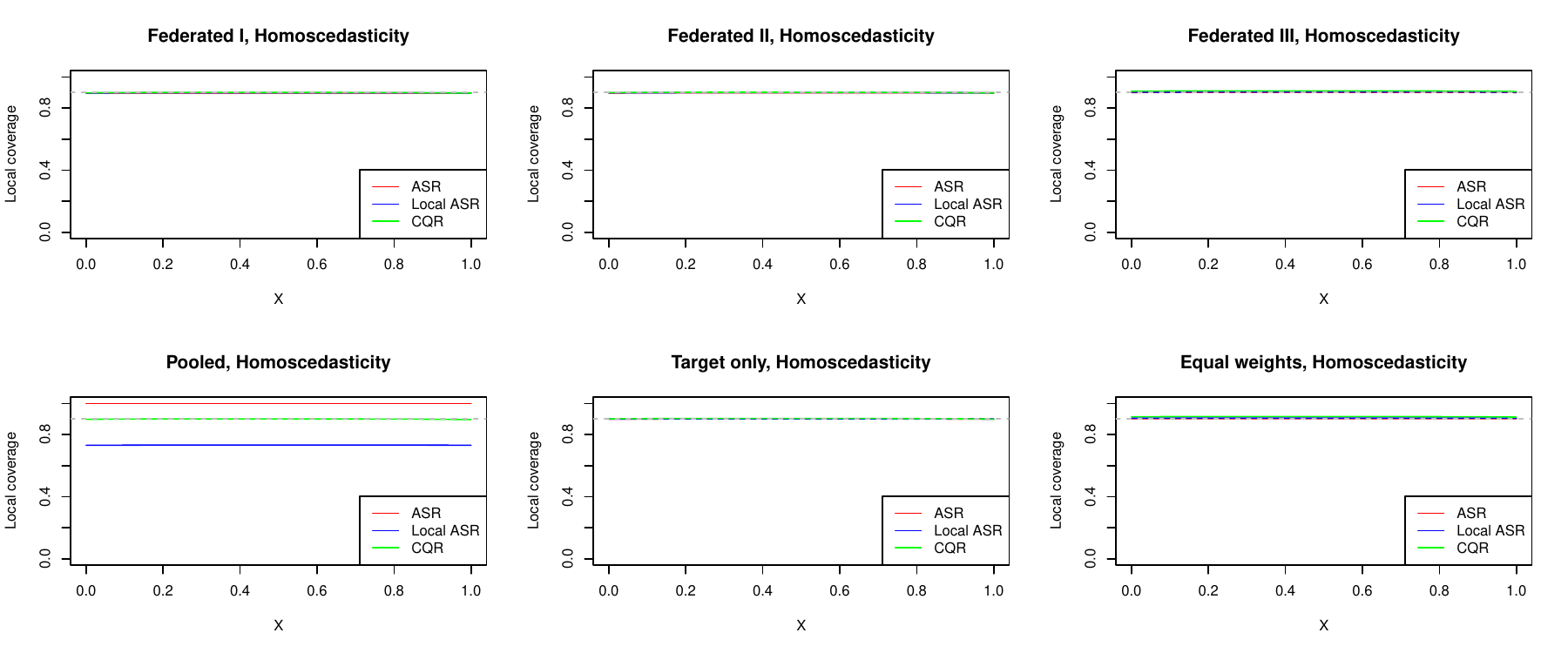}
    \includegraphics[width=1\textwidth]{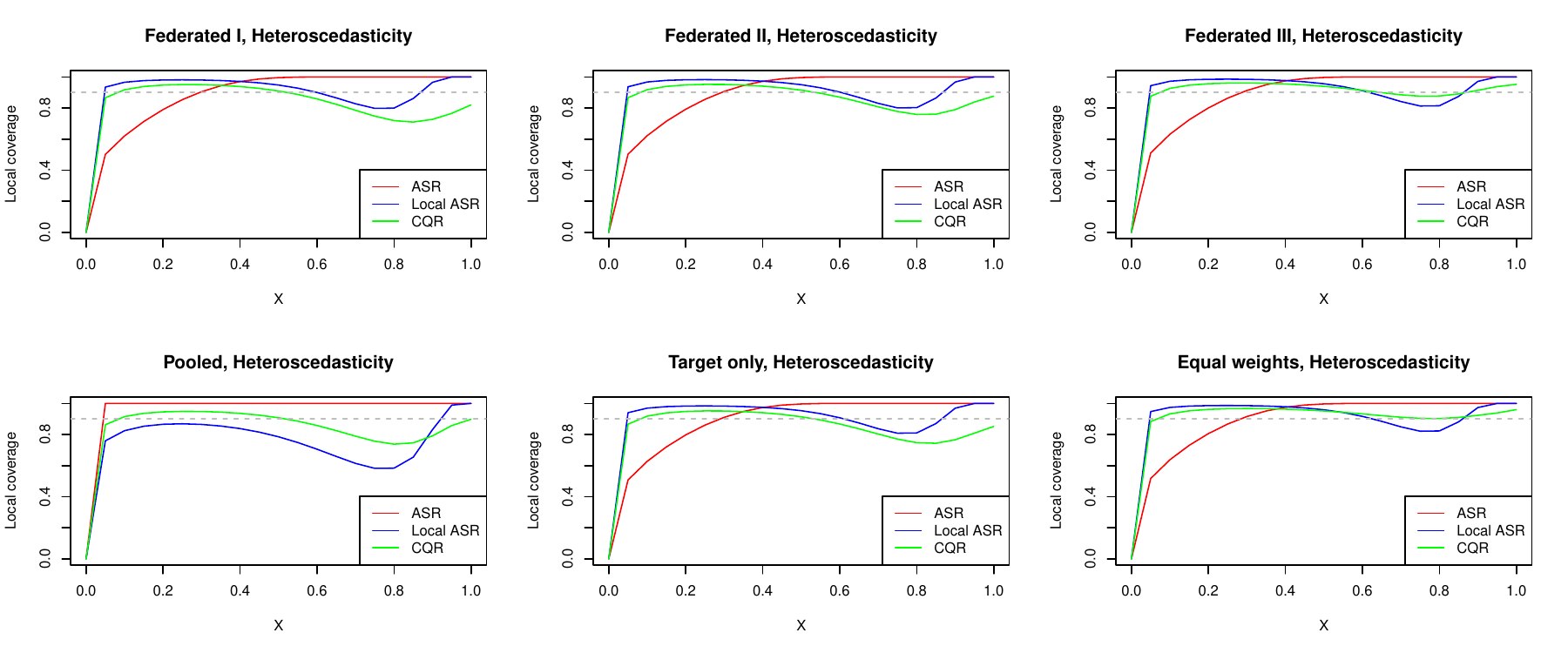}
    \caption{Local coverages, under CCOD is strongly violated and strongly heterogeneous covariate distributions and $n_k=3000$}
    \label{fig:localCov}
\end{figure}



\begin{figure}[H]
    \centering
    \includegraphics[width=0.96\textwidth]{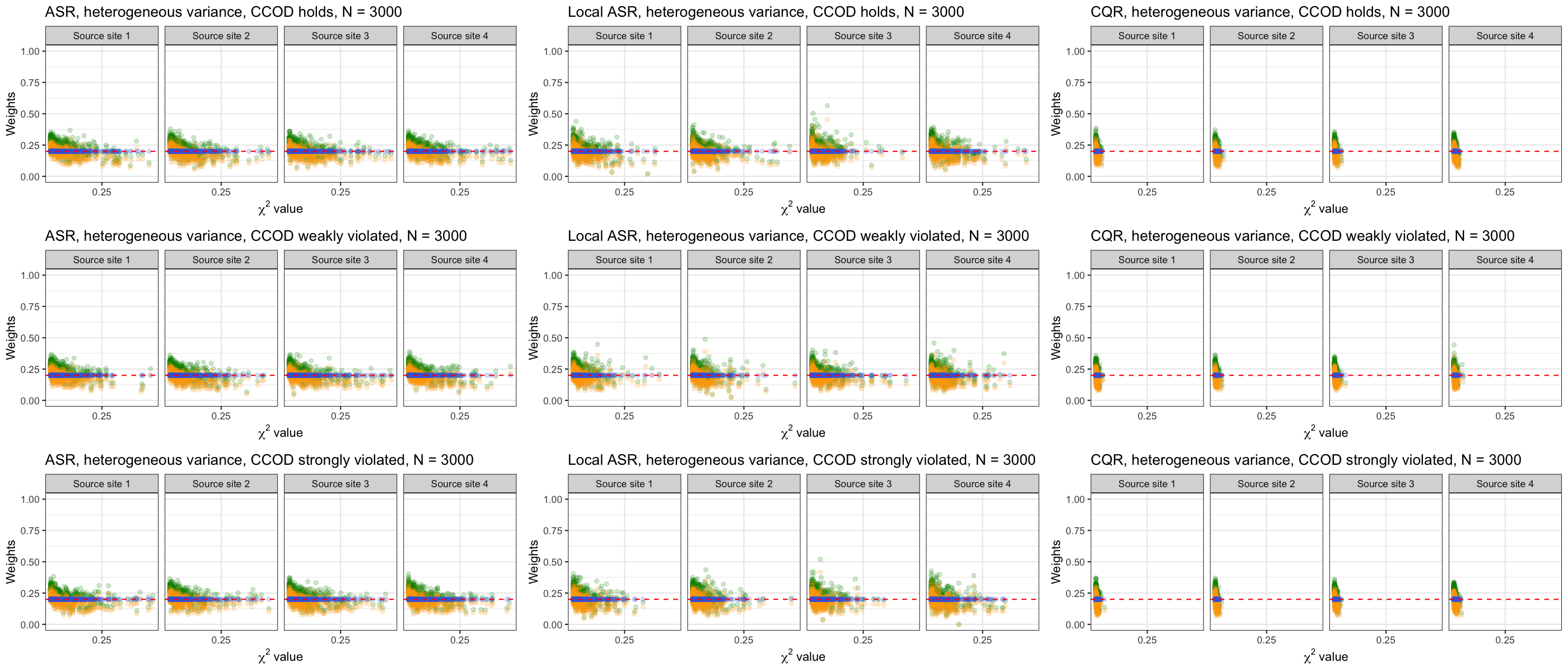}
    \includegraphics[width=0.96\textwidth]{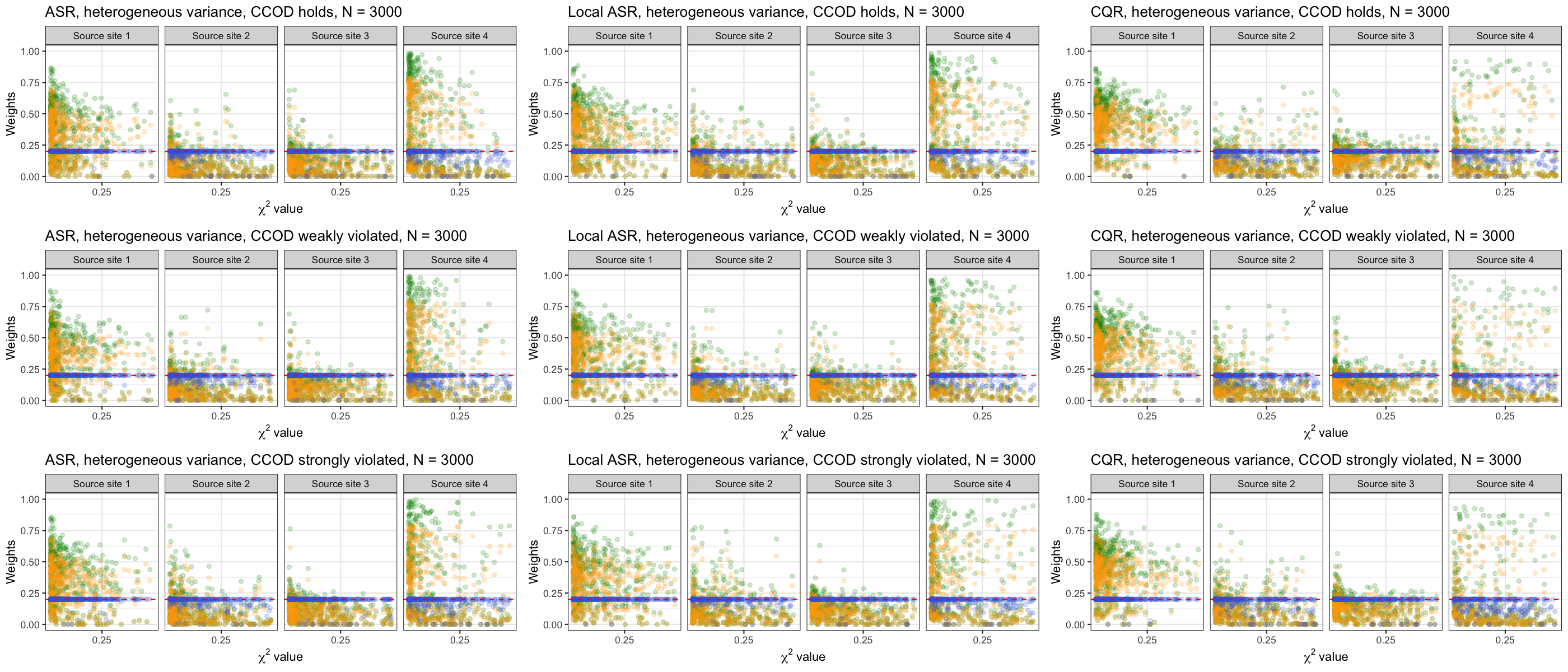}
    \includegraphics[width=0.96\textwidth]{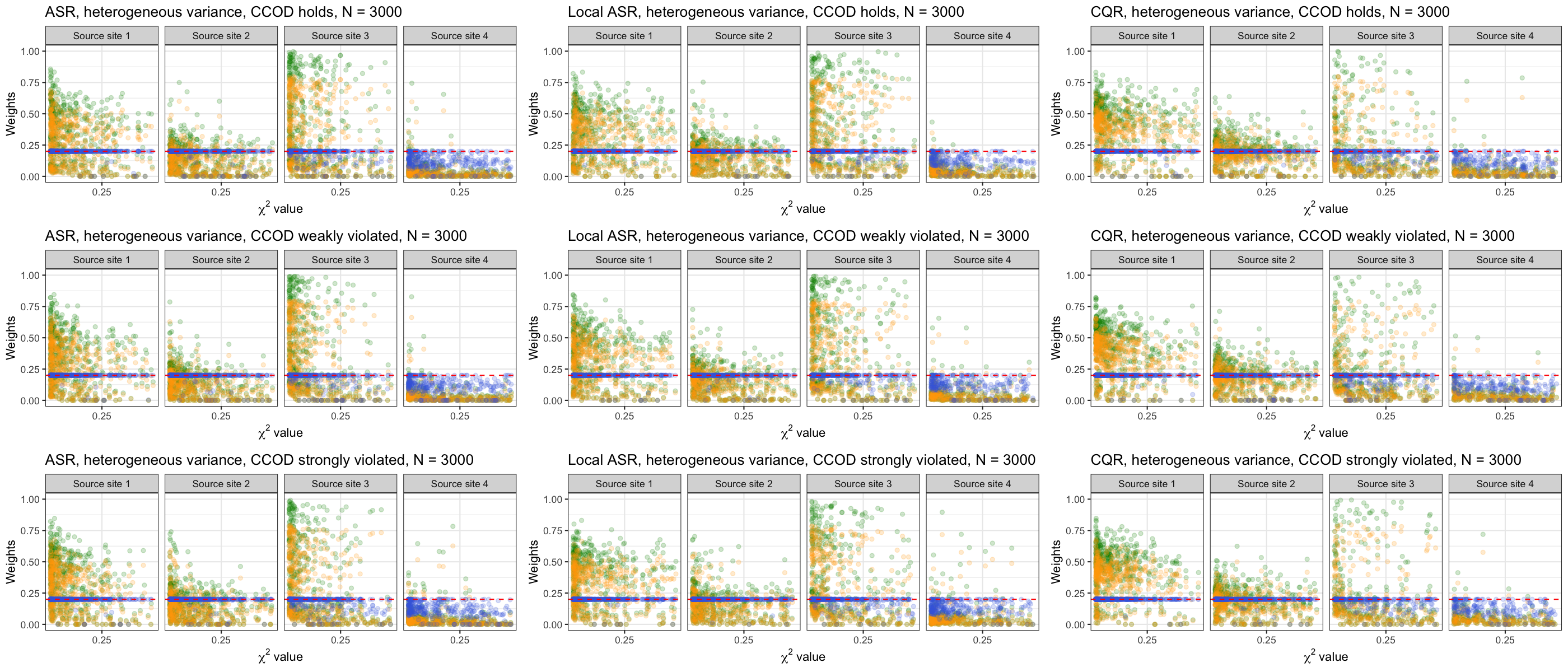}
    \caption{Weights vs. $\chi_k^2$ values, using $n_k=3000$ data under heteroscedasticity. The green points are by Federated I, the orange points are by Federated II (ours), the blue points are by Federated III, and the red dashed lines are for a reference line weights $= 0.2$. }
    \label{fig:wts_all}
\end{figure}

\begin{figure}[H]
    \centering
    \includegraphics[width=0.96\textwidth]{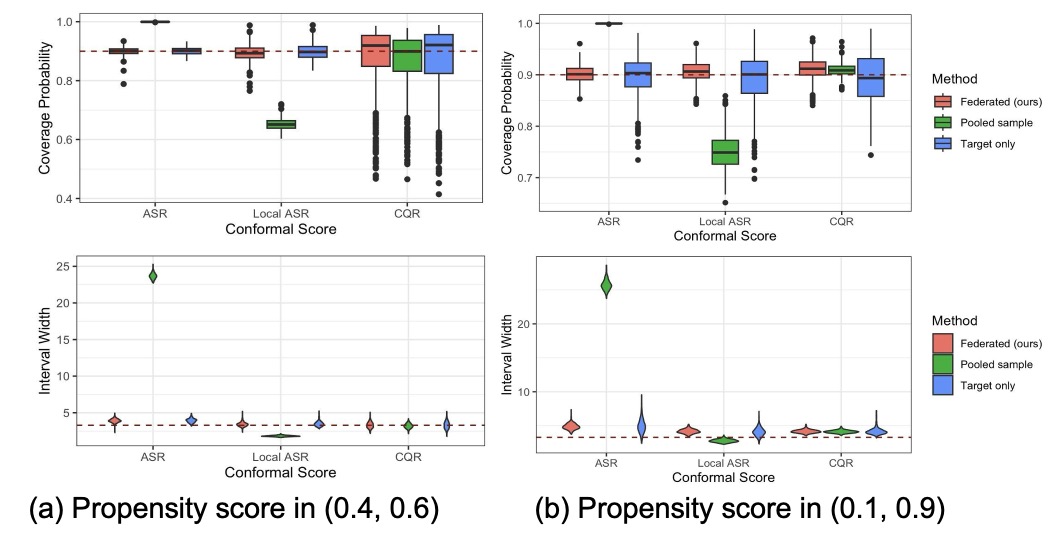}
    \caption{Comparison of coverage probabilities and average interval width when modifying the propensity score of observing the outcome between $(0.4,0.6)$ (panel (a)) and $(0.1,0.9)$ (panel (b)).}
    \label{newfig}
\end{figure}

\end{document}